\documentclass[12pt,authoryear,review]{elsarticle}

\makeatletter

 \def\ps@pprintTitle{%
 \let\@oddhead\@empty
 \let\@evenhead\@empty
 \def\@oddfoot{\centerline{\thepage}}%
 \let\@evenfoot\@oddfoot}
\makeatother
\usepackage{amsmath,amssymb,amsfonts}
\usepackage{graphicx}
\graphicspath{{images/}}

\usepackage{caption}
\usepackage{subcaption}

\usepackage{color}
\usepackage{bigstrut}
\usepackage{multirow}
\usepackage{pdflscape}
\usepackage{ulem}
\usepackage{dsfont}

\usepackage{hhline}
\usepackage{epstopdf}

\usepackage{verbatim}

\usepackage{tikz}
\usepackage[toc,page]{appendix}
\usepackage{hyperref}

\journal{SIAM J. Financial Mathematics}

\newtheorem{theorem}{Theorem}

\newtheorem{proposition}[theorem]{Proposition}

\newdefinition{remark}{Remark}
\newdefinition{assumption}{Assumption}
\newproof{proof}{Proof}
\newproof{pot}{Proof of Theorem \ref{thm2}}

\newcommand{\PP}{{\mathbb P}}

\newcommand{\QQ}{{\mathbb Q}}
\newcommand{\RRR}{{\mathbb R}}

\newcommand{\EE}{{\mathbb E}}
\newcommand{\FF}{{\mathcal F}}
\newcommand{\LL}{{\mathcal L}}
\newcommand{\DD}{{\mathcal D}}

\newcommand{\BB}{{\mathcal B}}
\newcommand{\GG}{{\mathcal G}}

\newcommand{\bkappa}{{\boldsymbol\kappa}}

\newcommand{\bp}{{\boldsymbol p}}

\newcommand{\bW}{{\boldsymbol W}}

\newcommand{\mH}{{\mathcal H}}

\newcommand{\bzero}{{\boldsymbol{0}}}

\newcommand{\bQ}{{\boldsymbol Q}}
\newcommand{\bmQ}{\boldsymbol {\mathcal Q}}

\newcommand{\bnu}{{\boldsymbol\nu}}

\newcommand{\calV}{\mathcal V}

\newcommand{\brho}{{\boldsymbol\rho}}

\newcommand{\bbq}{{\boldsymbol q}}
\newcommand{\by}{{\boldsymbol p}}

\newcommand{\hX}{{\widehat{X}}}
\newcommand{\hY}{{\widehat{Y}}}
\newcommand{\hZ}{{\widehat{Z}}}

\newcommand{\hmu}{{\widehat{\mu}}}
\newcommand{\hbeta}{{\widehat{\beta}}}

\newcommand{\tX}{{\widetilde{X}}}
\newcommand{\tY}{{\widetilde{Y}}}
\newcommand{\tZ}{{\widetilde{Z}}}

\newcommand{\mcX}{{\mathcal X}}

\newcommand{\mfh}{{\mathfrak{h}}}
\newcommand{\mfbQ}{\boldsymbol{\mathfrak{Q}}}
\newcommand{\mfQ}{{\mathfrak{Q}}}

\newcommand{\mfbp}{\boldsymbol{\mathfrak{p}}}

\newcommand{\mfbq}{\boldsymbol{\mathfrak{q}}}
\newcommand{\mfq}{{\mathfrak{q}}}

\newcommand{\mfX}{{\mathfrak{X}}}

\newcommand{\mfY}{{\mathfrak{Y}}}

\newcommand{\mrmH}{{\mathrm{H}}}

\newcommand{\whk}{{\widehat{k}}}

\newcommand{\cartea}[1]{{\color{black} #1}} 
\newcommand{\tianyi}[1]{{\color{black}#1}} 

\begin{document}

\begin{frontmatter}

\title {\textbf{Trading Foreign Exchange Triplets}
\tnoteref{t1}
}

\tnotetext[t1]{SJ would like to acknowledge the support of the Natural Sciences and Engineering Research Council of Canada (NSERC), [funding reference numbers RGPIN-2018-05705 and RGPAS-2018-522715]. We thank seminar participants at the Fields Institute, PIMS Summer School in Mathematical Finance, LABEX--Paris, SIAM (Austin 2017), Dublin City University, Technical University Berlin, Humboldt University of Berlin. This work first appeared on \url{https://ssrn.com/abstract=3054656}.}

\author[author1]{\'Alvaro Cartea}
\ead{alvaro.cartea@maths.ox.ac.uk}
\author[author2]{Sebastian Jaimungal}
\ead{sebastian.jaimungal@utoronto.ca}
\author[author2]{Tianyi Jia}
\ead{tianyi.jia@mail.utoronto.ca}
\address[author1]{Mathematical Institute, University of Oxford, Oxford, UK\\  Oxford-Man Institute of Quantitative Finance, Oxford, UK}
\address[author2]{Department of Statistical Sciences, University of Toronto, Toronto, Canada}

\begin{abstract}
We develop the optimal trading strategy for a foreign exchange (FX) broker who must liquidate a large position in an illiquid currency pair. To maximize revenues, the broker considers trading in a currency triplet which consists of the illiquid pair and two other liquid currency pairs. The liquid pairs in the triplet are chosen so that one of the pairs is redundant. The broker is risk-neutral and accounts for model ambiguity in the FX rates to make her  strategy robust to model misspecification. When the broker is ambiguity neutral (averse) the trading strategy in each pair is independent (dependent) of the inventory in the other two pairs in the triplet.   We employ simulations to illustrate how the robust strategies perform. For a range of ambiguity aversion parameters, we find the mean Profit and Loss (P\&L) of the strategy increases and the standard deviation of the P\&L decreases as ambiguity aversion increases.
\end{abstract}


\end{frontmatter}

\section{Introduction}

Trading activity in the  foreign exchange (FX) market is vast. On average, the daily turnover is in excess of 5 trillion USD, where approximately a third of this turnover is in the FX spot market, and the remainder is in FX derivatives: forwards, swaps, and options (see \cite{BIS}). All currencies are traded in the FX spot market,  but most of the turnover is in a handful of  pairs (USD/EUR, USD/JPY, USD/GBP, USD/AUD, USD/CAD). The market quotes  a rate for those who buy (bid) or sell (ask) the currency pair, where the difference between the ask and the bid prices is the quoted spread and the average of the bid and ask prices is the mid-exchange rate.  A currency pair is traded by simultaneously taking a long position in one currency and a short position in the other currency of the pair. The conversion rate between the two currencies is given by the bid (resp. ask) if the investor is going long (resp. short) the currency pair.

In this paper we show how a broker liquidates a large position in an illiquid currency pair. One approach is to trade only the pair the broker aims to liquidate. An alternative is for the broker to trade in three currency pairs, one of which is the illiquid pair and the   additional   pairs are such that the three pairs are formed by combinations of only three currencies.  Therefore, by no-arbitrage,  one of the currency pairs can be replicated by taking positions in the other two  pairs, and hence the dynamics of the three currency pairs exhibit strong co-movements. Ideally,  at least one of the two additional pairs is very liquid (heavily traded),  so this ``triangle'' dependence between the pairs is considered by the broker to devise liquidation strategies that offset the illiquidity in one pair with the two other more liquid pairs.

There is no central Exchange framework in FX, so brokers trade in electronic communication networks (ECNs) with multiple liquidity providers and via other channels with their own pool of clients, see \cite{CarteaJaiWalLastLook} and \cite{OomenAggregator}.  We refer to the broker's trading activities on all ECNs as trading in the `Exchange'.  The broker can control her own liquidity taking orders sent to the Exchange, but has little or no  control on the arrival rate  of the orders of her clients.

The broker is risk-neutral and acknowledges that her model of mid-exchange rates of the currency pairs, which is represented by a reference measure $\mathbb P$, may be misspecified. She deals with this model uncertainty, also referred to as ambiguity aversion, by considering alternative models when she develops the optimal execution strategy, see \cite{CDJ} and \cite{cartea2016speculative}. These alternative models are characterized by probability measures that are absolutely continuous with respect to the reference model $\mathbb P$. The decision to reject the reference measure is based on a penalty that the broker incurs if she adopts an alternative model. The magnitude of the penalty depends on the broker's degree of ambiguity aversion and is based on a ‘measure’ of the distance between the reference and the alternative measure.

The broker solves an optimal control problem where the objective is to liquidate a large position in an illiquid pair, and does   by submitting liquidity taking orders to the Exchange in all three currency pairs. These orders have a temporary impact in the quote currency of the pair. That is, when the broker executes trades in the Exchange she receives worse rates than the quoted mid-exchange rate.  The broker simultaneously fills orders from her pool of clients. These orders are filled at the mid-exchange rate and clients pay a service fee to the broker, where the size of the fee is proportional to the size of the order in the currency pair.

We show that when the broker is ambiguity neutral, i.e., fully trusts the reference model, the optimal trading strategy in each currency pair is independent of the level of inventory in the other two pairs. Thus, the broker's strategy does not employ the other two currency pairs of the triplet to optimally execute a position in an illiquid pair. On the other hand, when the broker is ambiguity averse,  the inventory position in each pair affects the trading strategy in the other two pairs. As the level of  ambiguity aversion increases,  the speed at which the strategy builds positions in the liquid pairs of the triplet increases.

We also demonstrate that the ambiguity averse broker makes her model robust to misspecification by  adopting a candidate measure where the drifts of the mid-exchange rates are different from those of the reference model. Specifically,  when the broker's initial position in the illiquid pair is short (resp. long), the  broker devises a strategy where the drift in the mid-exchange rates of the currency pairs in the triplet, relative to the model under the reference measure $\PP$,  is larger (resp. smaller). This affects the speed of trading in the robust model, relative to that of the reference model,  because the expected growth in the mid-exchange rate of the currency pair affects the value of the inventory. For example, if a currency pair is expected to appreciate, then, everything else being equal,  it is optimal to increase the position in that pair.

We use simulations,  based on parameters calibrated to FX data, to illustrate the performance of the strategy. Our results show that when  the broker makes her model robust to misspecification,  the  mean Profit and Loss  (P\&L) of the strategy increases and  the standard deviation of the P\&L decreases.

To the best of our knowledge this is the first paper that shows how FX brokers manage large positions in currency pairs. Our framework can be extended to FX trading algorithms to make markets and those designed to take speculative positions in currency pairs   with strong co-dependence. Although the literature on algorithmic trading in the FX market is scant, there is a large body of work that looks at optimal liquidation, and other algorithmic strategies in equity markets. For comprehensive treatments of algorithmic trading and microstructure issues in equity markets see \cite{lehalle2013market}, \cite{TheBook}, \cite{gueant2016financial}, \cite{abergel2016limit}.

The remainder of the paper proceeds as follows. Section \ref{sec: model setup} presents the broker's reference model for the currency pairs and discusses the broker's cashflows from dealing currency pairs with her pool of clients and from executing trades in the FX Exchange. Section \ref{sec: optimal trading under reference measure} solves the control problem for the ambiguity neutral broker.  Section \ref{sec: ambiguity aversion} shows how the broker introduces model ambiguity to make the model robust to misspecification and solves the  control problem of the ambiguity averse  broker and provides the optimal speed to trade the currency pairs in the triplet. Section \ref{sec: performance of strategy} illustrates the performance of the strategy and Section \ref{sec: conclusions} concludes. Finally,  some proofs are collected in the Appendix.

\section{Model Setup}\label{sec: model setup}

FX traders buy and sell currency pairs. Trading in a pair of currencies involves  the simultaneous sale of one currency and purchase of another currency and the price at which this transaction is done is  the mid-exchange rate plus (minus)  half the quoted spread when longing (shorting) the pair. In this paper we focus on a triplet of currency pairs that links the exchange rates of three  currencies. We denote the three currencies by $\{1,\,2,\,3\}$, and denote the mid-exchange rates of the currency pairs by $X_t$, $Y_t$, and $Z_t$, where $t$ denotes time. Here $X_t$ is the mid-exchange rate for the pair $(2,1)$, $Y_t$ is the mid-exchange rate for the pair $(3,1)$, and $Z_t$ is the mid-exchange rate for the pair $(2,3)$.

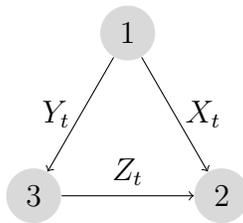
\begin{figure}
\begin{center}
\begin{tikzpicture}[]
\node [circle, fill=gray!30] (C3) at (0, 0) {$3$};
\node [circle, fill=gray!30] (C2) at +(0: 2.5) {$2$};
\node [circle, fill=gray!30] (C1) at +(60: 2.5) {$1$};
\draw [->] (C1) -- (C2) node [midway, right] {$X_t$};
\draw [->] (C1) -- (C3) node [midway, left] {$Y_t$};
\draw [->] (C3) -- (C2) node [midway, above] {$Z_t$};
\end{tikzpicture}
\end{center}
\caption{Currencies and exchange rates for pairs in the FX triplet. The quote currency (the start of the arrow) is sold at the rate shown on the link, and one unit of the base currency (the end of the arrow) purchased. For example, when the broker goes long the currency pair $(2,1)$ in a frictionless market,  she buys one unit of the base currency $2$ and sells $X_t$ units of quote currency $1$. \label{fig:FX-triplets-diagram}
}
\end{figure}

When the broker goes long the currency pair $(2,1)$ in a frictionless market (i.e., no fees and zero spread),  she buys one unit of currency $2$, known as the base currency in the pair, and sells $X_t$ units of currency $1$,  known as the quote currency in the pair. Similarly, the cost of one unit of currency $3$ is $Y_t$ units of currency $1$; and the cost of one unit of currency $2$ is $Z_t$ units of currency $3$. By no-arbitrage, the mid-exchange rates for the reverse direction of the pairs, i.e., for $(1,2)$, $(1,3)$, and $(3,2)$, are given by $1/X_t$, $1/Y_t$, $1/Z_t$,  respectively. See Figure \ref{fig:FX-triplets-diagram}.

One of the three currency pairs is redundant because it can be obtained by simultaneously executing transactions in the other two pairs. For example, buying  one unit of the pair $(2,1)$ and buying $X_t$ units of the pair $(1,3)$, is equivalent to buying one unit of the pair $(2,3)$. Thus, by no-arbitrage the mid-exchange rates satisfy
\begin{equation}\label{eqn: Z in terms of Y and X}
Z_t =  X_t /Y_t\,.
\end{equation}

We focus on a triplet with a redundant currency pair  because  brokers who need to unwind positions in highly illiquid currency pairs may rely on the triplet to enhance the financial performance of the liquidation strategy. For example, assume that a broker needs to unwind a position in one illiquid pair only. One approach is to devise an optimal liquidation strategy based on exclusively trading the illiquid pair. An alternative is to devise a liquidation strategy where the broker also trades in two other pairs which are more liquid. These two additional pairs are chosen to form, in conjunction with the illiquid pair, a triangular triplet. The dynamics of the three pairs exhibit strong co-dependence,  and the broker may offset some of the costs arising from  the illiquidity in one pair by taking positions in the other, more liquid, pairs.

\subsection{Mid-exchange rate dynamics}

To devise the optimal trading strategy, the broker  assumes a model for the dynamics of the triplet. This model is described by the completed probability space $(\Omega, \FF =\{\FF_t\}_{t\in[0,T]}, \PP)$, where $\FF$ is the natural filtration generated by the processes $X=(X_t)_{t\in[0,T]}$ and $Y=(Y_t)_{t\in[0,T]}$, which satisfy the stochastic differential equations (SDEs)
\begin{subequations}
	\begin{align}
	{dX_t} &= \mu_x\,{X_t}\,dt+\sigma_x\,{X_t}\,dW_t^x\,, \label{eqn:dXtGBM2}\\
	{dY_t} &= \mu_y\,{Y_t}\,dt+\sigma_y\,{Y_t}\,dW_t^y\,,\label{eqn:dYtGBM2}
	\end{align}
where $\sigma_x>0$, $\sigma_y>0$, $\mu_x$, $\mu_y$   are constants, $W^x=\left(W^x_t\right)_{t\in[0,T]}$ and $W^y=\left(W^y_t\right)_{t\in[0,T]}$ are $\PP$--Brownian motions with correlation $\rho\in[-1,1]$. The dynamics of the  mid-exchange rate $Z=\left(Z_t\right)_{t\in[0,T]}$  are implied by the no-arbitrage relationship \eqref{eqn: Z in terms of Y and X} and are given by the SDE
\begin{equation}\label{eqn: dynamics of Z under P}
  dZ_t= \mu_z\,{Z_t}  \,dt + \sigma_x\,{Z_t} \,dW_t^x - \sigma_y\,{Z_t} \,dW_t^y\,,\qquad \mbox{where}\qquad  \mu_z = \mu_x - \mu_y + \sigma_y^2 - \rho\,\sigma_x\,\sigma_y\,.
\end{equation}%
\end{subequations}%

The exchange rates of the triplet satisfy identity \eqref{eqn: Z in terms of Y and X} if there are no trading fees, the quoted spread is zero (i.e., best bid and best ask prices in the LOB are the same), and the market is arbitrage-free. On the other hand, if the product of the three exchange rates is greater than one, there is an arbitrage opportunity,  see \cite{fenn2009mirage}.  For example,  assume that $Z$ is the exchange rate for the pair EUR/USD, $Y$ is for the pair USD/CHF, and $X$ is for the pair EUR/CHF, and assume that $Z^{bid}_t\,Y^{bid}_t \,1/X^{ask}_t = f $, where $f>1$. If a trader initially holds one Euro,  the arbitrage consists of: (i) with one Euro purchase $Z^{bid}_t$ USD, (ii) which buys   $Y^{bid}_t$ units CFH, (iii) and sell this position for  $1/X^{ask}_t $ Euros. Clearly, this chain of transactions will `convert' one Euro into $f$ Euros with no risk.

We  could assume that the exchange rate $Z$ is as \eqref{eqn: dynamics of Z under P} plus noise so that the product of the three rates  is a stochastic process which is most of the time below 1 and very seldom above 1. However, to keep the model as simple as possible,  we assume that the product  of the mid-exchange rates of the currency pairs of the triplet is 1 for all $t$.

\subsection{Speed of trading, order flow, and inventory}

We assume the broker uses the Exchange exclusively to take liquidity (i.e., does not provide liquidity to the Exchange), and provides liquidity to her pool of clients by filling the orders of clients.  Thus, the broker controls the speed at which she trades aggressively in  the Exchange, but cannot control the rate at which she fills the liquidity taking orders of her clients.  The broker trades in the triplet of currencies over the time window  $[0,\,T]$ and aims at holding zero inventory in all pairs by the terminal date  $T$.  We denote by  $\bnu=(\bnu_t)_{t\in[0,T]}$ the vector-valued process of execution speeds in the Exchange, for the pairs $x,\,y,\,z$, and write $\bnu_t=(\nu^x_t,\,\nu^y_t,\,\nu^z_t)^\intercal$, where $^\intercal$ is the transpose operator.

The vector-valued process $\bQ^\bnu=(\bQ^\bnu_t)_{t\in[0,T]}$  represents  the controlled inventory in each currency pair  of the triplet, which results from trading in the Exchange and from trading with her own pool of clients. That is, $\bQ^\bnu$ tracks the units of base currency in each currency pair -- for modelling purposes, it is convenient to use ``inventory of currency pair'' instead of ``units of base currency'' to track trades in different currency pairs that share a common base currency. At $t=0$, the broker's initial holding in the currency pairs is denoted by $\bQ^\bnu_0=(Q^{x,\bnu}_0,\,Q^{y,\bnu}_0,\,Q^{z,\bnu}_0)^\intercal$. To understand our notation (in a frictionless market)  if the first entry of the inventory triplet is  $Q^x_0=100$, then the broker is short $100\times X_0 $ units of currency $1$ and long  100  units of currency 2. Similarly, if $Q^x_0=-100$  the broker is long $100\times X_0 $  units of currency $1$ and short 100 units of currency 2.

When the broker's initial inventory is $\bQ^\bnu_0$, her initial portfolio of currency pairs consists of $Q^x_0 + Q^z_0$ units of currency $2$, and $Q^y_0$ units of currency $3$. Then, the value of the portfolio in currency $1$ is $Q^x_0 \times X_0 + Q^y_0\times Y_0 + Q^z_0 \times Z_0\times Y_0$ in a frictionless market, and in the absence of arbitrage, the value of the  portfolio is $Q^x_0 \times X_0 + Q^y_0\times Y_0 + Q^z_0 \times X_0$. As mentioned previously, two inventory processes, $Q^{x,\bnu}_t$ and $Q^{z,\bnu}_t$, track the units of the base currency $2$  in the pairs $(2,1)$ and $(2,3)$.

The following example illustrates how the broker's speed of trading affects inventory. Assume that over the time step $\Delta t$ the strategy is to execute in the Exchange the amount $\nu^x_t\,\Delta t>0$ in the currency pair $(2,1)$. Thus the post-transaction inventory becomes $Q^{x,\bnu}_{t+\Delta t}= Q^{x,\bnu}_t - \nu^x_t\,\Delta t$. If we further assume that  over the time interval $(t,t+\Delta t]$ the mid-exchange rate $X_t$ does not change, then the effect of the execution strategy is to reduce by:   $\nu^x_t\,\Delta t\times X_t$ the amount the broker is short in currency $1$,   and  $\nu^x_t\,\Delta t$ the amount the broker is long in currency $2$.

While the broker's target is to hold zero units in the base currency of all currency pairs by the terminal date, i.e., $\bQ^\bnu_T = \bzero$, she continues to fill buy and sell orders of the clients in the three pairs, and this affects the  broker's inventory. We assume the  liquidity taking  orders of the clients arrive according to independent compound Poisson processes. Let $O^{k,\pm} = (O_{t}^{k,\pm})_{t\in[0,T]}$ represent the order flow from clients in the currency pair $k \in \{x,\,y,\,z\}$. Specifically,
\begin{equation}
	O_{t}^{k,\pm} = \sum_{n=1}^{N_{t}^{k,\pm}}\xi_{n}^{k,\pm}\,,\label{eqn:client order}
\end{equation}
where market sell is denoted by $+$ and market buy is denoted by $-$. Here,   $\{\xi^{k,\pm}_n\}_{n=1,2,\dots} \stackrel{i.i.d.}{\sim} F^{k,\pm}$, where $F^{k,\pm}$ is the distribution function, with support on $[0,\infty)$,  of the volume of the clients' buy and sell orders in currency pair $k$. The counting processes  $N^{k,\pm} = (N_t^{k,\pm})_{t\in[0,T]}$ are independent Poisson processes with intensities $\lambda^{k,\pm}$, and  represent the number of market (buy and sell) orders received by the broker in each currency pair through time. The Poisson processes, the size of the trades, and the Brownian motions $W^{x,y}$ are all mutually independent. All distribution functions are under the reference  measure $\PP$. At this point, we extend the completed filtered probability space $(\Omega, \,\FF=\{\FF_t\}_{t\in[0,T]},\,\mathbb P)$, so that henceforth, $\FF_t$ represents the natural filtration generated by $(W^x,\,W^y,\, O^{x,\pm},\, O^{y,\pm},\, O^{z,\pm})$.

For example, when the broker fills a client's  market sell order of volume $\xi^{x,+}$ and   there is no brokerage service fee, the broker's  inventory position in the pair $x$ increases by $\xi^{x,+}$; that is,  the broker has shorted an additional $\xi^{x,+}\times X_t$  units of currency $1$ and longed an additional $\xi^{x,+}$ units of currency $2$.

In all, the broker's inventory is affected by the controlled executions sent  to the Exchange, and the clients' filled market orders. The dynamics of the inventory in each pair are given by
\begin{equation}
Q_T^{k,\bnu} = Q_t^{k,\bnu} - \int_t^T \nu_u^k\,du-\int_t^T\int_{\RRR}r\,J^{k,-}(du,dr) + \int_t^T\int_{\RRR}r\,J^{k,+}(du,dr)\,,\label{eqn:Q dynamics}
\end{equation}
where $J^{k,\pm}$  is the jump measure of $O_{k,t}^{\pm}$,  with intensity measure $\lambda^{k,\pm}\,F^{k,\pm}(dr)\,du$.

\subsection{Temporary mid-exchange rate impact on quote currency}

The broker's trades in the Exchange receive temporary price impact. When the broker sends executions to the Exchange, she enters a position in the quote currency at a value that is worse than the quoted mid-exchange rate. For example, if at time $t$ the broker sends  to the Exchange a sell order for $r$ units of the currency pair $(2,1)$, instead of receiving $r\times X_t$ units of the quote currency $1$ (in exchange of $r$ units of the base currency $2$), she receives  $r\times X_t\,(1-a_x\,r)$ units of currency $1$ (in exchange of   $r$ units of currency $2$). The  parameter $a_x\geq 0$  denotes the temporary impact of the broker's order for the currency pair $(2,1)$, i.e., the broker's order walks the FX book of the Exchange. The impact on the mid-exchange rate is temporary because we assume that liquidity in the Exchange's order book is immediately replenished (or replenished faster than the broker submits orders).

On the other hand, when the broker fills  buy and sell orders from her clients, the broker charges a spread around the mid-exchange rates $X_t,\,Y_t,\,Z_t$. For example, if a client sends a sell order of $r$ units in the pair $(2,1)$  the client receives the rate $X_t\,\left(1 - c_x\,r\right)r$, which is worse than the mid-exchange rate $X_t$ if $c_x>0$. We refer to this spread as a brokerage service fee.

\subsection{Broker's marked-to-market cash position}

At every instant in time the broker's cash position is  marked-to-market in currency 1 and the cumulative value is denoted by $\mcX^\bnu =(\mcX^\bnu)_{t\in[0,T]}$. This cash position results from  the market making activity  with her  clients, the  liquidity taking orders she sends to the FX Exchange,    and  is given by
\begin{subequations}\label{eqn: cash process}
\begin{eqnarray}\nonumber
	\mcX_T^{\bnu} &=& \mcX_t^{\bnu} \\ \label{eqn: cash process a}
&&+ \int_{t}^{T}\Big[\hX_u\, \nu_u^x + \hY_u \, \nu_u^y + \hZ_u\,\nu_u^z\,\Big]\,du\\\label{eqn: cash process b}
	&&+\int_t^T\int_{\RRR^+}\tX^+_u\,r\,J^{x,-}(du,dr)-\int_t^T\int_{\RRR^+}\tX^-_u\,r\,J^{x,+}(du,dr)\\ \label{eqn: cash process c}
	&&+\int_t^T\int_{\RRR^+}\tY^+_u\,r\,J^{y,-}(du,dr)-\int_t^T\int_{\RRR^+}\tY^-_u\,r\,J^{y,+}(du,dr)\\ \label{eqn: cash process d}
	&&+\int_t^T\int_{\RRR^+} \tZ^+_u\,r\,J^{z,-}(du,dr)-\int_t^T\int_{\RRR^+} \tZ^-_u\,r\,J^{z,+}(du,dr)\,,
\end{eqnarray}
\end{subequations}
where
\begin{subequations}\label{eqn: notation for execution prices}
\begin{eqnarray}\label{eqn: notation for execution prices x}
\hX_t &=& X_t\,(1 - a_x\,\nu_t^x)\,,
\qquad\,\,\,\,\tX^\pm _t = X_t\,(1\pm c_x^\mp \,r)\,,\\ \label{eqn: notation for execution prices y}
\hY_t &=& Y_t\,(1 - a_y\,\nu_t^y)\,,\qquad\quad\, \tY^\pm _t = Y_t\,(1\pm c_y^\mp \,r)\,,\\ \label{eqn: notation for execution prices z}
\hZ_t &=& Y_t\,Z_t\,(1 - a_z\,\nu_t^z)\,, \qquad\tZ^\pm _t = Y_t\,Z_t\,(1\pm c_z^\mp \,r)\,.
\end{eqnarray}
\end{subequations}
Recall that $a_k\geq 0$ is the temporary impact on the mid-exchange rate due to the trades sent by the broker to  the Exchange, and the brokerage service fee parameter $c_k^{\pm}\geq 0$ scales with the order size to obtain the spread the broker charges to her clients.

Line \eqref{eqn: cash process a}  represents the cash-flows from the broker's executions in the Exchange when trading in the three currency pairs and marked-to-market in currency 1. In particular, the first, second, and third terms are the cash-flows received/paid by the broker when trading the currency pair $(2,1)$, $(3,1)$, and $(2,3)$, respectively. Recall that a positive execution speed corresponds to selling the corresponding currency pair, so, e.g., if  the execution speed  $\nu^x $ is positive, the broker is selling the pair $(2,1)$, i.e., sell base currency 2 and purchase quote currency 1.

Lines \eqref{eqn: cash process b} to \eqref{eqn: cash process d} represent the cash-flows from the broker's trading activity with her pool of clients. For example, line \eqref{eqn: cash process b} is the cash-flow from trading currency pair $(2,1)$ and marked-to-market in the quote currency 1.

\section{Optimal trading in triplet of currency pairs}\label{sec: optimal trading under reference measure}

In this section we pose and solve the broker's control problem when the objective is to liquidate an inventory position in the currency pairs by the terminal date $T$, and the broker is confident about the reference model $\PP$. The broker maximizes terminal expected wealth resulting from trading in the Exchange and dealing currency pairs with her clients. The results in this section are used as a benchmark to compare how the broker's strategy changes when the broker acknowledges that the reference model $\PP$ is misspecified (see Section \ref{sec: ambiguity aversion}).

\subsection{Value function and optimal trading}\label{subsec: value function and optimal trading in measure P}

The broker's performance criterion is
\begin{equation}\label{eqn: performance criteria under P}
H^{\bnu}(t, \mcX,\by,\bbq) =  \EE_{t,\mcX,\by,\bbq}^{\PP}\bigg[\mcX_T^\bnu+\ell\left((X_T,Y_T),\bQ^{\bnu}_T\right)\bigg]\,,
\end{equation}
where $\by=(x,\,y)$, $\bbq=(q^x,\,q^y,\,q^z)$,
\begin{equation}\label{eqn:trmPnLty1}
\ell(\by, \bbq) = \,x\,q^x\,\left(1-\alpha_x\,q^x \right)+y\,q^y\,\left(1-\alpha_y\,q^y \right)+x\,q^z\,\left(1-\alpha_z\,q^z \right)\,,\quad\alpha_{x,y,z}\geq 0\,.
\end{equation}
The terminal payment $\ell(\by, \bbq)$ denotes the cash proceeds from liquidating any terminal inventory in the Exchange minus a penalty paid by the broker, marked-to-market in currency 1.  The terminal payment includes a penalty of the form $\alpha_k\,q^2_k$, which  consists of fees stemming from crossing the spread and a non-financial penalty included by the broker to tweak the liquidation strategy. For example, in the absence of order flow from  clients, in the limit $\alpha_k\to\infty$  the optimal strategy guarantees full liquidation by the terminal date. The state variable $z$ is not included as an argument in the performance criterion  because by no-arbitrage  it is a redundant state variable. \cartea{Finally, note that the broker is risk-neutral because  the performance criterion \eqref{eqn: performance criteria under P} optimises expected wealth, i.e.,  the broker is not sensitive to the risk in the outcome of terminal wealth.}

The broker's value function is
\begin{eqnarray}\label{eqn: value function}
		H(t,\mcX,\by,\bbq) &=& \sup_{\bnu\in\calV}
H^{\bnu}(t, \mcX,\by,\bbq)
\end{eqnarray}
where the set of admissible controls is
\begin{equation}
	\calV = \left\{\bnu:\bnu\textrm{ is }\FF\textrm{-predictable, and } \EE^{\PP}\left[\,\int_0^T (\nu_s)^2\,ds\,\right]<\infty \right\}\,.\label{eqn:admissible strategy 1}
\end{equation}

To solve the optimal control problem in \eqref{eqn: value function}, the dynamic programming principle holds, and the value function is the unique solution of the Hamilton-Jacobi-Bellman (HJB) equation
\begin{subequations}
\label{eqn:DPE-zero-ambiguity}
\begin{eqnarray}\label{eqn: DPE under P}
\begin{split}
0=&\left(\partial_t + \LL^{\by} + \sup_{\bnu}\LL^{\bnu}\right)H + \sum_{k \in \{x,y,z\}, i=\pm}
\int_{\RRR^+}\,\Delta^{k,-i}_{t,\mcX,\by,\bbq,r} H\,\lambda^{k,i}\,F^{k,i}(dr)
\,,
\end{split}
\end{eqnarray}
with terminal condition
\begin{equation}\label{eqn:termCon P}
H(T, \mcX,\by,\bbq) = \mcX + x\,q^x\,\left(1-\alpha_x\,q^x \right)+y\,q^y\,\left(1-\alpha_y\,q^y \right)+x\,q^z\,\left(1-\alpha_z\,q^z \right)\,,
\end{equation}
\end{subequations}
where the various operators are defined as follows:
\begin{subequations}
	\begin{align}
\begin{split}
\LL^{\by}H =&\;\;\; \mu_x\,x\,\partial_xH + \mu_y\,y\,\partial_yH
\\
& +\tfrac{1}{2}\,\sigma_x^2\,x^2\,\partial_{x x}H
  + \tfrac{1}{2}\,\sigma_y^2\,y^2\,\partial_{y  y }H
  + \rho\,\sigma_x\,\sigma_y\,x\,y\,\partial_{ x y }H\,,
\end{split}
\label{LxyAmbAver P}
\\
\LL^{\bnu}H =&\,\sum_{k \in \{x,y,z\}}\,\bigg[-a_k\,\whk\,\partial_{\mcX} H\,(\nu^k)^2 + (\whk\,\partial_{\mcX} H - \partial_{q^k} H)\,\nu^k\bigg]\label{LvAmbAver P}
\\
\Delta^{k,+}_{t,\mcX,\by,\bbq, r}H(t, \mcX,\by,\bbq) =&\, H(t,\mcX + \tilde k^+ \, r, \by, \bbq - r\,\bold{1}^k)-H(t, \mcX,\by,\bbq)\,,\label{Deltak+ P}\\
\Delta^{k,-}_{t,\mcX,\by,\bbq,r}H(t, \mcX,\by,\bbq) =&\, H(t,\mcX - \tilde k^- \, r, \by, \bbq + r\,\bold{1}^k)-H(t, \mcX,\by,\bbq)\,,\label{Deltak- P}
\end{align}
\end{subequations}
$\whk = k\,1_{\{k\in\{x,y\}\}}+x\,1_{\{k=z\}}$, $\bold{1}^x = (1, 0, 0)^\intercal$, $\bold{1}^y = (0, 1, 0)^\intercal$, $\bold{1}^z = (0, 0, 1)^\intercal$, and   $\tilde k^{\pm}\in \{\tilde x, \tilde y, \tilde z\}$ denotes the mid-exchange rates defined in \eqref{eqn: notation for execution prices}.

\begin{proposition}\label{prop10}
The DPE \eqref{eqn: DPE under P} admits the solution
\begin{equation}\label{eqn:ansatzH under P}
\begin{split}
H(t,\mcX,\by,\bbq) = \mcX &+ x\,(q^x+q^z)+y\,q^y-h_0^x(t)\,x-h_0^y(t)\,y -h_1^x(t)\,x\,q^x\\
&-h_1^y(t)\,y\,q^y-h_1^z(t)\,x\,q^z
-h_2^x(t)\,x\,(q^x)^2-h_2^y(t)\,y\,(q^y)^2-h_2^z(t)\,x\,(q^z)^2\,,
\end{split}
\end{equation}
where $h_{2}^{x,y,z}(t)$, $h_{1}^{x,y,z}(t)$, $h_{0}^{x,y}(t)$ are deterministic functions of time given explicitly in \eqref{eqn:mfh2k}, \eqref{eqn:mfh1k}, \eqref{eqn:mfh0k}.
\end{proposition}
For a proof see Appendix \ref{appendix_I}. 

\begin{theorem} \label{thm:VerificationP}
If $|\mu_{\whk}|<| \alpha_k/a_k|$, then the value function of problem \eqref{eqn: value function} is given by \eqref{eqn:ansatzH under P} and  the controls $\bnu^*$ are optimal and are given componentwise by
\begin{equation}
\nu^{k,*}_t = \tfrac{1}{2\,a_k} \left( h_1^k(t) + 2\,h_2^k(t)\,Q^{k,\bnu^*}_t\right)\,,\qquad
k\in\{x,y,z\}\,,
\label{eqn:opt-controls-P}
\end{equation}
where $h_{1,2}^k(t)$ are provided in \eqref{eqn:mfh2k} and \eqref{eqn:mfh1k}.
\end{theorem}
For a proof see Appendix \ref{sec:Proof-Verification-P}. 
From the above theorem, one sees that the optimal strategy does not have any cross effects. That is, the speed of execution in one pair is independent of the inventory held in the other two pairs.

As a result of Theorem \ref{thm:VerificationP}, when $Q_0^x = 0$, $Q_0^y = 0$, $Q_0^z\neq 0$, $\mu_x=0$, $\mu_y=0$,  $\lambda^{x,\pm} = 0$ and $\lambda^{x,\pm} = 0$, the ambiguity neutral broker will only trade currency pair $(2,3)$  \cartea{-- recall that in our model, an ambiguity neutral broker is equivalent to a risk-neutral one}. Another observation is that the optimal trading strategy does not depend on mid-exchange rates and is a deterministic function of time if there are no orders from clients. However, below we show that the ambiguity averse broker will always trade the three currency pairs of the triplet regardless of the initial inventory and drift of the mid-exchange rates.

The following proposition shows  the execution strategy as the terminal date approaches and the terminal liquidation penalty is arbitrarily large.

\begin{proposition}
	\label{prop11}
(i) When $\lambda^{k,\pm} = 0$, the optimal trading strategy remains admissible in the limit as the terminal inventory penalty parameter  $\alpha_k\to+\infty$ (for $k\in\{x,y,z\}$). Moreover, near the end of the trading horizon, we have
	\begin{equation}
	\lim_{\alpha_k\to+\infty} \nu_t^{k,*} = (T-t)^{-1}\,Q^{k,\bnu^*}_t + o(T-t)\,,\label{eqn:nuk t to T}
	\end{equation}
which results in complete liquidation of all inventories by $T$.

(ii) When the arrival rate of the orders from clients is positive, i.e., $\lambda^{k,\pm} > 0$, the optimal trading strategy is not admissible in the limit as $\alpha_k\to+\infty$ (for $k\in\{x,y,z\}$).
\end{proposition}
For a proof see Appendix \ref{appendix_L}.

\section{Ambiguity Aversion on the Mid-Exchange Rates}\label{sec: ambiguity aversion}

The broker assumes that the  dynamics of the three currency pairs are given by the reference measure $\PP$, see  SDEs \eqref{eqn:dXtGBM2}-\eqref{eqn: dynamics of Z under P}. However, the broker is not fully confident about the reference model for the mid-exchange rates, so she  considers alternative measures to make the model robust to misspecification, see \cite{CDJ}. This ambiguity about model choice, or model uncertainty,  has an effect on the optimal strategy employed by the broker when trading in the three currency pairs. We incorporate this ambiguity about model choice in two steps. First, we characterize alternative measures that describe the mid-exchange rate dynamics, and then we determine how the broker decides between employing the reference measure $\PP$ or one of the alternative measures.

The broker considers candidate measures $\QQ$ that are equivalent to $\PP$ and characterized by the Radon-Nikodym derivative
\begin{equation}
\frac{d\QQ(\bkappa)}{d\PP} = \exp\bigg\{\frac{1}{2}\int_0^{T}\bkappa_u^\intercal\,\brho^{-1}\,\bkappa_u\,du+\int_0^{T}\bkappa_u^\intercal\,\brho^{-1}\,d\bW_u^{\bkappa}\bigg\}\,,\label{eqn:RN derivative}
\end{equation}
where
$\brho=
\begin{bmatrix}
    1       & \rho \\
    \rho       & 1\\
\end{bmatrix}\,,
$
and $\bkappa=((\kappa_t^x ,\kappa_t^y )^\intercal)_{t\in[0,T]}$ is a two-dimensional $\mathcal F$-adapted process.

We denote by $\bmQ$ the class of alternative measures
\[
\bmQ=\bigg\{\mathbb Q(\bkappa)\;|\;\bkappa\text{ is }\mathcal F-\text{adapted and }
\mathbb E^\PP\left[\exp\left\{\tfrac{1}{2}\textstyle\int_0^{T}\bkappa_u^\intercal\,\brho^{-1}\,\bkappa_u\,du\right\}\right]<\infty
\bigg\}\,,
\]
and write
\begin{subequations}
	\begin{align}
	\label{eqn:dXtGBM2 under Q}
dX_t =& X_t\left(\mu_x\,dt + \sigma_x\,\kappa^{x}_t\,dt+ \sigma_x\, dW_t^{x,\bkappa}\right) \,,\\\label{eqn:dYtGBM2 under Q}
dY_t =& Y_t\left( \mu_y\,dt +  \sigma_y\,\kappa^{y}_t\,dt+ \sigma_y\, dW_t^{y,\bkappa}\right)\,,\\\label{eqn:dZtGBM2 under Q}
dZ_t =& Z_t\left(\left(\mu_z + \sigma_x\,\kappa^{x}_t - \sigma_y\,\kappa^{y}_t\right)\,dt + \sigma_x\, dW_t^{x,\bkappa} - \sigma_y\, dW_t^{y,\bkappa}\right)  \,,
	\end{align}
\end{subequations}
where  $W^{x,y,\bkappa}=(W^{x,y,\bkappa}_t)_{t\in[0,T]}$ are $\QQ(\bkappa)$-Brownian motions, and $[W^{x,\bkappa}, W^{y,\bkappa}]_t=\rho\,t$, $\rho\in [\,-1,\,1\,]$.

Next, the broker  penalizes deviations from the reference measure using the relative entropy from $t$ to $T$, i.e.,
\begin{equation}
\begin{split}
\mH_{t,T}(\QQ\,|\,\PP)
=
\mathbb E^{\mathbb Q}_t\left[\log\bigg(\bigg(\frac{d\QQ}{d\PP}\bigg)_T\,\bigg/\,\bigg(\frac{d\QQ}{d\PP}\bigg)_t\bigg)
\right]\,.
\end{split}
\end{equation}
Note that $\mH_{t,T}(\QQ(\bkappa)\,|\,\PP)\geq 0$ and $\mH_{t,T}(\QQ(\bkappa)\,|\,\PP)=0$ if and only if $\bkappa = \boldsymbol{0}$, i.e., if and only if $\QQ(\bkappa)=\PP$ almost surely.

\subsection{Performance criterion and value function}\label{sec: performance criteria and value function}

As before, the broker's  aim is to liquidate the position in the currency pairs by the terminal date $T$ while maximizing expected terminal wealth and considering candidate models that are penalized using relative entropy. The broker's performance criterion is
\begin{equation}\label{eqn: performance criteria}
H^{\bnu}(t, \mcX,\by,\bbq) =
\inf_{\QQ\in\bmQ}
\left\{\EE_{t,\mcX,\by,\bbq}^{\QQ}
\left[
\;
\mcX_T^\bnu+\ell((X_T,Y_T),\bQ^{\bnu}_T)
\;
\right]
+ \tfrac{1}{\varphi}\, \mH_{t,T}(\QQ\,|\,\PP)
\right\}\,,
\end{equation}
where  $\ell(\by, \bbq)$  is as in \eqref{eqn:trmPnLty1} and recall that $\by=(x,\,y)$, $\bbq=(q^x,\,q^y,\,q^z)$.

Here the parameter $\varphi $ is a non-negative constant that represents the broker's  degree of ambiguity aversion. If the broker is confident about the reference measure $\PP$, then the ambiguity aversion parameter $\varphi$ is  small and any deviation from the reference model is very costly. In the extreme   $\varphi\to 0$, the broker is very confident about the reference measure,  so she chooses $\PP$ because the penalty that results from rejecting the reference measure is too high. \cartea{Thus, when $\varphi\to 0$, the broker is ambiguity neutral and the problem is equivalent to solving the value function in \eqref{eqn: value function} discussed above when the broker is risk-neutral.}

On the other hand, if the broker is very ambiguous about the reference model, considering alternative models results in a very small penalty. In the extreme   $\varphi\to \infty$, deviations from the reference model are costless, so  the broker  considers the worst case scenario.

\cartea{One can also assume that the broker is risk-averse and employs the performance criterion
\begin{equation*}
H^{\bnu}(t, \mcX,\by,\bbq) =  \EE_{t,\mcX,\by,\bbq}^{\PP}\bigg[U\left(\mcX_T^\bnu+\ell\left((X_T,Y_T),\bQ^{\bnu}_T\right)\right)\bigg]\,,
\end{equation*}
where, as above, $\by=(x,\,y)$, $\bbq=(q^x,\,q^y,\,q^z)$ and  $\ell(\by, \bbq)$ is as in \eqref{eqn:trmPnLty1}. Here, $U(\,\cdot\,)$ is the broker's  utility function, which is concave. If the utility function is linear,  we obtain the risk-neutral performance criterion  in \eqref{eqn: value function} discussed above.  In general, when $U$ is concave, the broker's problem becomes more difficult to solve. We note, however,  that under certain assumptions, ambiguity aversion and exponential utilities are equivalent, see e.g., \cite{schweizer2010minimal} for details. }

\subsubsection{The dynamic programming equation}
The broker's value function is
\begin{eqnarray}\label{eqn:valueFunAmbAver}
H(t,\mcX,\by,\bbq) &=& \sup_{\bnu\in\calV}\,H^{\bnu}(t, \mcX,\by,\bbq)\,,
\end{eqnarray}
where the set of admissible controls is as in \eqref{eqn:admissible strategy 1}. To solve the optimal control problem in \eqref{eqn:valueFunAmbAver}, the dynamic programming principle holds, and the value function is the unique viscosity solution of the Hamilton-Jacobi-Bellman-Isaacs (HJBI) equation
\begin{subequations}
\begin{equation}\label{eqn:DPEAmbAver}
0=\left(\partial_t + \LL^{\by} + \sup_{\bnu}\LL^{\bnu} + \inf_{\bkappa}\LL^{\bkappa}\right)H
+ \sum_{k\in \{x,y,z\},i=\pm}\int_{\RRR^+}\,\Delta^{k,-i}_{t,\mcX,\by,\bbq,r}H\,\lambda^{k,i}\,F^{k,i}(dr)
\,,
\end{equation}
with terminal condition
\begin{equation}\label{eqn:termConAmbAver}
H(T, \mcX,\by,\bbq) = \mcX + x\,q^x\,\left(1-\alpha_x\,q^x \right)+y\,q^y\,\left(1-\alpha_y\,q^y \right)+x\,q^z\,\left(1-\alpha_z\,q^z \right)\,,
\end{equation}
\label{eqn:HJBI-full}
\end{subequations}
where $\LL^\by$, $\LL^\bnu$, and $\Delta^{k,\pm}_{t,\mcX,\by,\bbq, r}$ are given in \eqref{LxyAmbAver P}--\eqref{Deltak- P},
\begin{equation*}
	\LL^{\bkappa}H = \bkappa^\intercal\,\DD H+\frac{1}{2\,\varphi}\,\bkappa^\intercal\,\brho^{-1}\,\bkappa\,,\qquad\text{and }\qquad
\DD H = (\sigma_x\,x\,\partial_{x}H\textrm{, }\sigma_y\,y\,\partial_{y}H )^\intercal\,.\label{eqn:DAmbAver}
\end{equation*}

In the following proposition we show that the optimal trading speed does not  depend on the cash position.

\begin{proposition}
	\label{prop1}
Write the value function as
\begin{equation}\label{eqn:RemoveWealthAmbAver}
	H(t, \mcX,\by,\bbq) = \mcX + \mrmH(t,\by,\bbq)\,.
\end{equation}
Then, from \eqref{eqn:RemoveWealthAmbAver} and the HJBI \eqref{eqn:HJBI-full}, the function $\mrmH$ satisfies
\begin{subequations}
\label{eqn:DPEAmbiguityAverse}
\begin{equation}\label{eqn:DPERemoveWealthAmbAver}
0 =\,\left(\partial_t + \LL^{\by} + \GG_1 + \GG_2\right)\mrmH
+ \sum_{k\in \{x,y,z\},i=\pm}\int_{\RRR^+}\,\left(i\,\tilde k^i +\Delta^{k,-i}_{t,\mcX,\by,\bbq,r}\mrmH\right)\,\lambda^{k,i}\,F^{k,i}(dr),
\end{equation}
subject to the terminal condition
\begin{equation}\label{eqn:termConRemoveWealthAmbAver}
\mrmH(T, \by,\bbq) = x\,q^x\,\left(1-\alpha_x\,q^x \right)+y\,q^y\,\left(1-\alpha_y\,q^y \right)+x\,q^z\,\left(1-\alpha_z\,q^z \right)\,,
\end{equation}
\label{eqn:RemoveWealth HJBI-full}
\end{subequations}
and the operators $\GG_1$ and $\GG_2$ act as follows:
\begin{subequations}
	\begin{align}
	\GG_1\,\mrmH =&\, \sum_{k\in \{x,y,z\}}\,\frac{(\whk-\partial_{q^k}\mrmH)^2}{4\,a_k\,\whk}\,,
	\label{eqn:G1H}\\
	\GG_2\,\mrmH =&\, -\tfrac{\varphi}{2}\,\bigg(\sigma_x^2\,x^2\,(\partial_x\mrmH)^2+\sigma_y^2\,y^2\,(\partial_y\mrmH)^2+2\,\rho\,\sigma_x\,\sigma_y\,x\,y\,\partial_x\mrmH\,\partial_y\mrmH\bigg)\,.\label{eqn:G2H}
	\end{align}
\end{subequations}

Furthermore, the optimal speed of trading in feedback form is
\begin{equation}\label{eqn:vxyz}
\nu^{k,*} = \tfrac{1}{2\,a_k\,\whk}\,(\whk-\partial_{q^k}\mrmH)\,,
\end{equation}
and the optimal measure in feedback form is
\begin{eqnarray}\label{eqn:kappa}
\bkappa^{*} = -\varphi\,\begin{bmatrix} 1 & \rho\\ \rho & 1
\end{bmatrix}\,\begin{bmatrix}\sigma_x\,x\,\partial_x\mrmH\\\sigma_y\,y\,\partial_y\mrmH\end{bmatrix}.
\end{eqnarray}
\end{proposition}
\begin{proof}
Substitute \eqref{eqn:RemoveWealthAmbAver} into \eqref{eqn:DPEAmbAver} and maximize the term $\LL^{\bnu}H$ to obtain the optimal trading speed \eqref{eqn:vxyz}. As this term is quadratic in $\bnu$ (with negative quadratic coefficient), it is trivial to verify that the supremum $\bnu^*$ satisfies the first order condition (FOC) in \eqref{eqn:vxyz}. Similarly, we obtain \eqref{eqn:kappa} by the FOC of the $\LL^{\bkappa}H$ term in \eqref{eqn:DPEAmbAver}, which  is quadratic in $\bkappa$ with positive coefficients for $(\kappa^x)^2$ and $(\kappa^y)^2$, thus we obtain the minimizer $\bkappa^*$. Then, by substituting \eqref{eqn:RemoveWealthAmbAver}, \eqref{eqn:vxyz}, and \eqref{eqn:kappa} into \eqref{eqn:HJBI-full} results in \eqref{eqn:RemoveWealth HJBI-full}. \qed
\end{proof}

\subsection{Expansion with respect to the ambiguity parameter $\varphi$}

The HJBI \eqref{eqn:HJBI-full} is nonlinear and we cannot obtain a solution in closed-form. We employ
perturbation methods, similar to those used by \cite{Lorig2016Portfolio} and \cite{Fouque2017Portfolio} in portfolio optimization problems,  to approximate the value function with the expansion
	\begin{equation}\label{eqn:Hexpansion}
	H(t, \mcX,\by,\bbq) = \mcX +
	H_0(t,\by,\bbq) +
		\varphi \, H_1(t,\by,\bbq)
	+ \varepsilon(t,\by,\bbq)\,.
	\end{equation}
By construction, we anticipate that $\varepsilon$ is $o(\varphi)$  (see discussion after Proposition \ref{prop4}), and in the following propositions we provide closed-form solutions for the terms $H_0(t,\by,\bbq)$  and $H_1(t,\by,\bbq)$.
\begin{proposition}
	\label{prop2}
	In the limit $\varphi\downarrow 0$, the value function of the ambiguity averse broker $H(t,\mcX,\by,\bbq)\to \mcX + H_0(t,\by,\bbq)$, where
	\begin{equation}\label{eqn:ansatzH0}
	\begin{split}
	H_0(t,\by,\bbq) = x\,(q^x+q^z)&+y\,q^y-h_{00}^x(t)\,x-h_{00}^y(t)\,y-h_{01}^x(t)\,x\,q^x
-h_{01}^y(t)\,y\,q^y\\
&-h_{01}^z(t)\,x\,q^z-h_{02}^x(t)\,x\,(q^x)^2
-h_{02}^y(t)\,y\,(q^y)^2-h_{02}^z(t)\,x\,(q^z)^2\,,
	\end{split}
	\end{equation}
$h_{00}^{x,y}(t)=h_{0}^{x,y}(t)$, $h_{01}^{x,y,z}(t)=h_{1}^{x,y,z}(t)$, $h_{02}^{x,y,z}(t)=h_{2}^{x,y,z}(t)$ are deterministic functions, and $h_{0}^{x,y}(t)$, $h_{1}^{x,y,z}(t)$, $h_{2}^{x,y,z}(t)$ are in the appendix,  see \eqref{eqn:mfh2k}, \eqref{eqn:mfh1k}, \eqref{eqn:mfh0k}.
\end{proposition}
\begin{proof}
It is straightforward to see from Proposition \ref{prop1}, that as $\varphi\downarrow0$, \eqref{eqn:DPEAmbiguityAverse} reduces to the PDE for the ambiguity neutral case in \eqref{eqn:DPE-zero-ambiguity}, and the result follows from Proposition \ref{prop10}.
\end{proof}

\begin{proposition}
	\label{prop14}
	Let $H(t,\mcX,\by,\bbq)$ be as in the expansion \eqref{eqn:Hexpansion}, then $H_1(t,\by,\bbq)$ satisfies
\begin{equation}\label{eqn:PDE-H1}
    \begin{split}
        0 =(\partial_t + \LL^{\bp})H_1 &- \sum_{k \in \{x,y,z\}} \!\! \tfrac{\whk-\partial_{q^k}H_0}{2\,a_k\,\whk}\,\partial_{q^k}H_1
        \\
        & + \sum_{k\in \{x,y,z\},i=\pm}\,\int_{\RRR^+}\,\Delta^{k,-i}_{t,\mcX,\by,\bbq,r}H_1\,\lambda^{k,i}\,F^{k,i}(dr)
        - \tfrac{1}{2}\,\DD H_0^\intercal\,\brho\,\DD H_0\,,
    \end{split}
\end{equation}
with terminal condition $H_1(T,\by,\bbq) = 0$.
\end{proposition}
\begin{proof}
Substitute expansion \eqref{eqn:Hexpansion} into \eqref{eqn:HJBI-full}, expand to first order in $\varphi$, use the fact that $H_0$ satisfies \eqref{eqn:DPE-zero-ambiguity}, and  equate the result to zero to obtain \eqref{eqn:PDE-H1}.
\qed
\end{proof}

Note that given  $H_0$,  the equation in \eqref{eqn:PDE-H1} is a linear PIDE. Thus we use a Feynman-Kac probabilistic representation to  write the solution of $H_1$:
\begin{equation}
\begin{split}
H_1(t,\by,\bbq) = \int_t^T\,\EE_{t,\by,\bbq}\left[\left(- \tfrac{1}{2}\,\DD H_0^\intercal\,\brho\,\DD H_0\right)(u,X_u,Y_u,\mfbQ_u)\right]\,du\,.\label{eqn:H1expectation}
\end{split}
\end{equation}
The term $\DD H_0$ can be computed explicitly from \eqref{eqn:ansatzH0}, and
here $\mfbQ_t = (\mfQ^x_t,\mfQ^y_t,\mfQ^z_t)^\intercal$ are independent auxiliary processes that satisfy the SDEs
\begin{equation}\label{eqn:auxQ dynamics}
	d\mfQ^k_t = \beta_{k}(t,\mfQ^k_t)\,dt - \int_{\RRR}rJ^{k,-}(dr,dt) + \int_{\RRR}rJ^{k,+}(dr,dt)\,,\quad\,k\in\{x,y,z\}\,,
\end{equation}
with the function
\begin{equation}\label{eqn:betak}
	\beta_{k}(t,q) = -\tfrac{1}{2\,a_k}\,\left(h_{01}^k(t)+2\,h_{02}^k(t)\,q\right)\,.\nonumber
\end{equation}

Based on the form of \eqref{eqn:ansatzH0} and \eqref{eqn:H1expectation}, we propose the ansatz
\begin{equation}\label{eqn: ansatzH1}
\begin{split}
H_1(t,\by,\bbq) = &\,H_{11}(t,q^x,q^z)\,x^2 + H_{12}(t,q^y)\,y^2 + H_{13}(t,q^x,q^y,q^z)\,x\,y\,.
\end{split}
\end{equation}
Insert \eqref{eqn: ansatzH1} into \eqref{eqn:PDE-H1}, collect terms proportional to $x^2$, $y^2$, $x\,y$, and  equate each term to zero to obtain expressions for $H_{11}$, $H_{12}$, $H_{13}$.
\begin{proposition}
We have the following representation for $H_{11}$, $H_{12}$, $H_{13}$:
\begin{subequations}
\begin{align}
H_{11}(t,q^x,q^z)&= -\int_t^Te^{(2\,\mu_x+\sigma_x^2)\,(u-t)}\,\tfrac{1}{2}\,\sigma_x^2\;
\EE_{t,q^x,q^z}\left[\,\bigg(\partial_xH_0\,(u,\mfQ^x_u,\mfQ^z_u)\bigg)^2\right]\,du\,,
\label{eqn: solution H11}
\\
H_{12}(t,q^y) &= -\int_t^Te^{(2\,\mu_y+\sigma_y^2)\,(u-t)}\,\tfrac{1}{2}\,\sigma_y^2\;
\EE_{t,q^y}\left[\,\bigg(\partial_yH_0\,(u,\mfQ^y_u)\bigg)^2\right]\,du\,,
\label{eqn: solution H12}
\\
\begin{split}
H_{13}(t,\bbq) &= -\int_t^T e^{(\mu_x+\mu_y+\rho\,\sigma_x\,\sigma_y)\,(u-t)}\\
&\hspace{3em}\times \,\rho\,\sigma_x\,\sigma_y\,\EE_{t,q^x,q^y,q^z}\left[\,\bigg(\partial_xH_0\,(u,\mfQ^x_u,\mfQ^z_u)\bigg)\,\bigg(\partial_yH_0\,(u,\mfQ^y_u)\bigg)\right]\,du\,.
\label{eqn: solution H13}
\end{split}
\end{align}
\end{subequations}
\end{proposition}
\begin{proof} 
Substitute \eqref{eqn: ansatzH1} into \eqref{eqn:PDE-H1} and equate $x^2$ terms to zero, then
	$H_{11}(t,q^x,q^z)$ satisfies
	\begin{equation}\label{eqn:eqn H11}
	\begin{split}
	0 = \partial_t H_{11} &+ (2\,\mu_x + \sigma_x^2)\,H_{11} - \sum_{k\in \{x,z\}}\,\frac{(x-\partial_{q^k}H_0)}{2\,a_k\,x}\,\partial_{q^k}H_{11}\\
	&+ \sum_{k\in \{x,z\},i=\pm}\,\int_{\RRR^+}\,\Delta^{k,-i}_{t,\mcX,\by,\bbq,r}H_{11}\,\lambda^{k,i}\,F^{k,i}(dr) - \tfrac{1}{2}\,\sigma_x^2\,(\partial_x H_0)^2\,,
	\end{split}
\end{equation}
with terminal condition $H_{11}(T,q^x,q^z) = 0$. Then, use the Feynman-Kac Theorem to show that  \eqref{eqn: solution H11} is a solution to the equation above. Similarly, we equate to zero the terms proportional to   $y^2$ and $x\,y$ and obtain the functions $H_{12}$ and $H_{13}$. \qed
\end{proof}

\begin{proposition}
	\label{prop4}
	Let $H(t,\mcX,\by,\bbq)$ be as in the expansion \eqref{eqn:Hexpansion}, then the residual function $\varepsilon(t,\by,\bbq)$ satisfies the PIDE
\begin{equation}\label{eqn:eqn error}
	0 =(\partial_t+\LL^{\by}+\sum^4_{i=1}\BB_{i})\,\varepsilon + \sum_{k\in \{x,y,z\},i=\pm}\,\int_{\RRR^+}\,\left(\Delta^{k,-i}_{t,\mcX,\by,\bbq,r}\varepsilon\right)\;\lambda^{k,i}\,F^{k,i}(dr)\\
	+ f_{\varepsilon}(t,\by,\bbq)\,,
	\end{equation}
	with terminal condition $\varepsilon(T,\by,\bbq) = 0$,
	and where the operators are defined as
	\begin{subequations}
		\begin{align}
	\BB_{1}\,\varepsilon =&\, \sum_{k \in \{x,y,z\}} \tfrac{1}{4\,a_k\,\whk}(\partial_{q^k}\varepsilon)^2\,,\label{eqn:BB1}\\
	\BB_{2}\,\varepsilon =&\, - \sum_{k \in \{x,y,z\}} \tfrac{1}{2\,a_k\,\whk}\left(\whk-\partial_{q^k}(H_0 + \varphi\,H_1)\right)\,\partial_{q^k}\varepsilon\,,\label{eqn:BB2}\\
	\BB_{3}\,\varepsilon =&\, - \tfrac{1}{2}\,\varphi\,
\left[\DD \varepsilon\right]^\intercal\,\brho\,\DD \varepsilon\,,\label{eqn:BB3}\\
	\BB_{4}\,\varepsilon =&\, -\varphi\,\left[\DD
(H_0 + \varphi\,H_1)\right]^\intercal\,\brho\,\DD \varepsilon\,,
\label{eqn:BB4}\\
	f_{\varepsilon}(t,\by,\bbq) =&\,
\varphi^2 \left\{
\sum_{k \in \{x,y,z\}}\tfrac{1}{4\,a_k\,\whk}\,(\partial_{q^k}H_1)^2
-\DD H_0^\intercal\,\brho\,\DD H_1
- \tfrac{1}{2}\,\varphi\,\DD H_1^\intercal\,\brho\,\DD H_1
\right\}\,.\label{eqn:fepsilon}
	\end{align}
	\end{subequations}
\end{proposition}
\begin{proof}
	Substitute expansion \eqref{eqn:Hexpansion} into \eqref{eqn:HJBI-full} and collect powers of $\varphi$, and obtain \eqref{eqn:eqn error} because $H_0$ satisfies \eqref{eqn:DPE-zero-ambiguity} and $H_1$ satisfies \eqref{eqn:PDE-H1}. \qed
\end{proof}

We observe that \eqref{eqn:eqn error} is a semilinear PIDE, and
\begin{equation*}
	\BB_{1}\,\varepsilon = \sup_{\bnu} \LL^{\bnu}\varepsilon\,,\qquad	\BB_{3}\,\varepsilon = \inf_{\bkappa} \LL^{\bkappa}\varepsilon\,,
\end{equation*}
where
\begin{equation*}
	\LL^{\bnu}\varepsilon = \sum_{k \in \{x,y,z\}}\,\bigg[-a_k\,\whk\,(\nu^k)^2 - (\partial_{q^k} \varepsilon)\,\nu^k\bigg]\,,\qquad
	\LL^{\bkappa}\varepsilon = \bkappa^\intercal\,\DD \varepsilon+\frac{1}{2\,\varphi}\,\bkappa^\intercal\,\brho^{-1}\,\bkappa\,,
\end{equation*}
and the extremes are achieved at
\begin{equation}\label{eqn:nu epsilon star}
	\nu^{k,*} = \tfrac{1}{2\,a_k\,\whk}(- \partial_{q^k}\varepsilon)
\qquad\text{and}\qquad\bkappa^{*} = -\varphi\,\begin{bmatrix} 1 & \rho\\ \rho & 1
\end{bmatrix}\,\begin{bmatrix}\sigma_x\,x\,\left(\partial_x\varepsilon\right)\\\sigma_y\,y\,\left(\partial_y\varepsilon\right)\end{bmatrix}\,.
\end{equation}

The PIDE in \eqref{eqn:eqn error} is the HJBI equation of the stochastic control problem:
\begin{eqnarray}\label{eqn:epsilon sup inf}
	\varepsilon(t,\by,\bbq) &=& \sup_{\bnu}\inf_{\bkappa}\EE_{t,\by,\bbq}^{\QQ(\bkappa)}\left[\int_t^T\,g_\varepsilon(\mfbp_u,\bnu_u) + \frac{1}{2\,\varphi}\,(\bkappa_u)^\intercal\,\brho^{-1}\,\bkappa_u + f_\varepsilon(u,\mfbp_u,\mfbq_u)\,du\right]\,,
\end{eqnarray}
where
\begin{eqnarray}\label{eqn:gepsilon}
g_\varepsilon(\bp,\bnu) &=& \sum_{k \in \{x,y,z\}}\,(-a_k)\,\whk\,(\nu^k)^2\,,
\end{eqnarray}
and $\mfbp_u = (\mfX_u,\mfY_u)$, $\mfbq_u = (\mfq^x_u,\mfq^y_u,\mfq^z_u)$ are auxiliary processes that satisfy the SDEs:
\begin{eqnarray}
	d\mfX_t &=& \mfX_t\left(\mu_x\,dt + \sigma_x\,\kappa^{x}_t\,dt+ \sigma_x\, dW_t^{x,\bkappa}\right) + \hmu_x(t,\mfbp_t,\mfbq_t)\,dt \,,\label{eqn:dXtGBM2 under Q epsilon}\\
	d\mfY_t &=& \mfY_t\left( \mu_y\,dt +  \sigma_y\,\kappa^{y}_t\,dt+ \sigma_y\, dW_t^{y,\bkappa}\right) + \hmu_y(t,\mfbp_t,\mfbq_t)\,dt\,,\label{eqn:dYtGBM2 under Q epsilon}\\
	d\mfq^k_t &=& \hbeta_{k}(t,\mfbp_t,\mfbq_t,\nu^k_t)\,dt - \int_{\RRR}rJ^{k,-}(dr,dt) + \int_{\RRR}rJ^{k,+}(dr,dt)\,,\quad\,k\in\{x,y,z\}\,,
\end{eqnarray}
with the functions
\begin{eqnarray*}
H_\varphi &=& H_0 + \varphi\,H_1\,,\\
\hmu_x(t,\by,\bbq) &=& -\,\varphi\,\left(\sigma_x^2\,x^2\,\partial_x H_\varphi + \rho\,\sigma_x\,\sigma_y\,x\,y\,\partial_y H_\varphi\right)\,(t,\by,\bbq)\,,\\
\hmu_y(t,\by,\bbq) &=& -\,\varphi\,\left(\sigma_y^2\,y^2\,\partial_y H_\varphi + \rho\,\sigma_x\,\sigma_y\,x\,y\,\partial_x H_\varphi\right)\,(t,\by,\bbq)\,,\\
\hbeta_{k}(t,\by,\bbq,\nu^k) &=& -\,\tfrac{1}{2\,a_k\,\whk}\left(\whk-\partial_{q^k}H_\varphi(t,\by,\bbq)\right) - \nu^k\,,\label{eqn:hatmuk hatbetak}
\end{eqnarray*}
and initial conditions $X_0$, $Y_0$, $Q^x_0$, $Q^y_0$, $Q^z_0$.

Suppose the optimal processes $(\bnu_t^*, \bkappa_t^*, \mfbp_u^*, \mfbq_u^*)_{t\in[0,T]}$ are admissible, then  \eqref{eqn:epsilon sup inf} is indeed
\begin{eqnarray}\label{eqn:epsilon optimal}
\varepsilon(t,\by,\bbq) &=& \EE_{t,\by,\bbq}^{\QQ(\bkappa^*)}\left[\int_t^T\,g_\varepsilon(\mfbp_u^*,\bnu_u^*) + \frac{1}{2\,\varphi}\,(\bkappa_u^*)^\intercal\,\brho^{-1}\,\bkappa_u^* + f_\varepsilon(u,\mfbp_u^*,\mfbq_u^*)\,du\right].
\end{eqnarray}

In addition, here we assume $\varepsilon$ is $O(\varphi^{1+r})$ for some $r>0$. Then, because of the initial data, we have $(\mfbp_t^*)_{t\in[0,T]}$ and $(\mfbq_t^*)_{t\in[0,T]}$ are $O(1)$. By \eqref{eqn:nu epsilon star}  we have $(\bnu_t^*)_{t\in[0,T]}$ is $O(\varphi^{1+r})$ and $(\bkappa_t^*)_{t\in[0,T]}$ is $O(\varphi^{2+r})$. By \eqref{eqn:gepsilon} and \eqref{eqn:fepsilon},  the first, second, and third terms of the integrand in \eqref{eqn:epsilon optimal} are $O(\varphi^{2+2\,r})$, $O(\varphi^{3+2\,r})$, and $O(\varphi^{2})$, respectively. Thus, the term $\varepsilon$ should be $O(\varphi^{2})$, which is consistent with our assumption that $O(\varphi^{1+r})$ for some $r>0$.

Next, insert \eqref{eqn:Hexpansion} into the optimal trading speeds \eqref{eqn:vxyz} and write
\begin{equation}
	\nu^{k,*} = \tfrac{1}{2\,a_k\,\whk} \left[ (\whk-\partial_{q^k}H_0)- \varphi\,\partial_{q^k}H_1- \partial_{q^k}\varepsilon \right]\,,\label{eqn:vk expansion}
\end{equation}
so the optimal measure that arises from ambiguity aversion \eqref{eqn:kappa} is characterized by
\begin{eqnarray}\label{eqn:kappa expansion}
\quad\bkappa^{*} = -\varphi\,\begin{bmatrix} 1 & \rho\\ \rho & 1
\end{bmatrix}\,\begin{bmatrix}\sigma_x\,x\,\left(\partial_xH_0+\varphi\,\partial_xH_1+\partial_x\varepsilon\right)\\\sigma_y\,y\,\left(\partial_yH_0+\varphi\,\partial_yH_1+\partial_y\varepsilon\right)\end{bmatrix}\,.
\end{eqnarray}

Thus, we  approximate the optimal controls $\bnu^*$ and $\bkappa^*$  with
\begin{equation}\label{eqn:vkaprox}
	\tilde{\nu}^{k,*} =  \tfrac{1}{2\,a_k\,\whk} \left[ (\whk-\partial_{q^k}H_0)- \varphi\,\partial_{q^k}H_1\right]\quad\text{and}\quad \tilde \bkappa^{*} = -\varphi\,\begin{bmatrix} 1 & \rho\\ \rho & 1
\end{bmatrix}\,\begin{bmatrix}\sigma_x\,x\,\left(\partial_xH_0+\varphi\,\partial_xH_1\right)\\\sigma_y\,y\,\left(\partial_yH_0+\varphi\,\partial_yH_1\right)\end{bmatrix}\,.
\end{equation}

We observe that the speeds of trading in \eqref{eqn:vkaprox} exhibit cross effects. The coefficient of ambiguity aversion parameter $\varphi$ depends on the price and inventory of each pair in the triplet. In the limit $\varphi\downarrow0$  we recover the ambiguity neutral trading speeds \eqref{eqn:opt-controls-P}.

The proposition below provides  trading speeds when the strategy in \eqref{eqn:vkaprox} approaches  the terminal date,  the  liquidation penalty is very large, and the broker does not fill orders from her pool of clients. In such a case, the strategy consists of TWAP and a term that arises from ambiguity aversion that shows the explicit cross-dependence between the liquidation speed in a pair and the inventory position in the three currency pairs.

\begin{proposition}
	\label{prop6}
	If the broker does not fill any orders from her clients, i.e., $\lambda^{k,\pm}=0$, then the  trading strategy in \eqref{eqn:vkaprox} remains admissible in the limit $\alpha_k\to+\infty$ (for $k\in\{x,y,z\}$). Moreover, close to the end of the trading horizon, the trading speed can be expressed as
	\begin{equation}\label{eqn: tilde nu k}
	\lim_{\substack{\alpha_k\to +\infty\,,\\\,k\in\{x,y,z\}}}\tilde{\nu}_t^{k,*} = \frac{Q_t^{k,\tilde \bnu^*}}{T-t} \left(1+ \varphi\,\sigma_{\whk}\,
C_{\whk}(X_t,Y_t,Q_t^{x,\tilde \bnu^*},Q_t^{y,\tilde \bnu^*},Q_t^{z,\tilde \bnu^*}) \right)+ o(T-t)\,,
	\end{equation}
	where
	\begin{equation}\label{eqn:Ck fcn}
	\begin{split}
	C_{x}(x,y,q^x,q^y,q^z) =&\, \sigma_x\,\left[a_x\,x\,(q^x)^2 + a_z\,x\,\,(q^z)^2\right] + \rho\,\sigma_y\,a_y\,y\,(q^y)^2\,,\\
	C_{y}(x,y,q^x,q^y,q^z) =&\, \rho\,\sigma_x\,\left[a_x\,x\,(q^x)^2 + a_z\,x\,\,(q^z)^2\right] + \sigma_y\,a_y\,y\,(q^y)^2\,,
	\end{split}	
	\end{equation}
	and $C_x+C_y\geq 0$ for any $x$, $y$, $q^x$, $q^y$, $q^z$, and recall that $\whk = k\,1_{\{k\in\{x,y\}\}} + x\,1_{\{k=z\}}$.
\end{proposition}
For a proof see  Appendix \ref{appendix_G}.

\section{Performance of strategy}\label{sec: performance of strategy}
In this section we employ simulations to show the performance of the strategy, where we employ the trading speed in \eqref{eqn:vkaprox}. We choose the following triplet of currency pairs: (NZD, USD), (AUD, USD), (NZD, AUD). In our notation, USD is currency 1, NZD is currency 2, and AUD is currency 3. Therefore, we denote the mid-exchange rate for the three pairs by $X_t$, $Y_t$, and $Z_t$ respectively. Historically, in this triplet, the pair (AUD, USD) is the most liquid, followed by (NZD, USD), and the least liquid pair of the triplet is  (NZD, AUD).

The broker  trades in lots of currency pairs. Each lot  consists of  $10^6$ currency pairs. The initial lot position is $Q^x_0 = 0\,, Q^y_0 = 0\,, Q^z_0 = 200$, the trading window is $T=1$ hour and we run 10,000 simulations.

\noindent\textbf{Statistical measure.} We assume that under the (true) statistical measure, denoted by $\tilde \PP$, the dynamics of the pairs $x$, $y$, $z$ satisfy the SDEs
\begin{eqnarray}\label{eqn:dXtGBM2 true}
	dX_t &=& \tilde \mu_x\,X_t\,dt+\tilde\sigma_x\,X_t\,d\widetilde W_t^x\,, \\ \label{eqn:dYtGBM2 true}
	dY_t &=& \tilde\mu_y\,Y_t\,dt+\tilde\sigma_y\,Y_t\,d\widetilde W_t^y\,,\\\label{eqn:dZtGBM2 true}
	dZ_t &=& \tilde \mu_z\,Z_t \,dt + \tilde \sigma_x\,Z_t\,d\widetilde W_t^x - \tilde \sigma_y\,Z_t\,d\widetilde W_t^y\,,
\end{eqnarray}
where, by no-arbitrage,
\begin{equation}\label{eqn: params of dZ in P true}
\tilde \mu_z = \tilde \mu_x - \tilde \mu_y + \tilde \sigma_y^2 - \tilde \rho\,\tilde \sigma_x\,\tilde \sigma_y\,,
\end{equation}
and where $\widetilde W^x=(\widetilde W^x_t)_{t\in[0,T]}$ and $\widetilde W^y=(\widetilde W^y_t)_{t\in[0,T]}$ are $\tilde \PP$-Brownian motions. Moreover, $d[\widetilde W^x, \widetilde W^y]_t=\tilde\rho\,dt$. Here,  $\tilde\mu_x$, $\tilde\mu_y$, $\tilde \rho$, $\tilde\sigma_x>0$, $\tilde\sigma_y>0$ are constants, which are not necessarily the same as those above in the reference model.

\noindent\textbf{Parameter estimates and reference measure.} We use FX data from  HistData.com  to obtain the parameters of the mid-exchange rate dynamics under statistical measure: $\tilde\mu_k$, $\tilde\sigma_k$, $k=\{x,\,y,\,z\}$, and $\tilde\rho$. We employ data between the trading hours 00:00 am to 13:59 pm Easter Standard time (without daylight savings adjustment), on the 8th of September, 2016. The first row of Table \ref{tab: parameter estimates} shows the maximum likelihood estimates of parameters used to simulate mid-exchange rates under the statistical measure $\tilde \PP$.    The second row of the table shows the parameters used by the broker under the reference measure $\PP$. These parameters are obtained by assuming that $\mu_k = 0$, $\sigma_k = \tilde\sigma_k$, $k=\{x,\,y\}$, $\rho = \tilde\rho$, and the parameters $\mu_z$, $\sigma_z$ are determined by the no-arbitrage condition \eqref{eqn: Z in terms of Y and X}. The time unit is hours.

\begin{table}
\scriptsize
\begin{tabular}{|c|c|c|c|c|c|c|c|c|c|c|}%
  \hline
   & $\sigma_x \times 10 ^{-3}$ & $\sigma_y\times 10 ^{-3}$ & $\sigma_z\times 10 ^{-3}$ & $\rho$ & $\mu_x\times 10^{-4}$ & $\mu_y\times 10^{-4}$ & $\mu_z\times 10^{-4}$ & $X_0^{NZD,USD}$ & $Y_0^{AUD,USD}$ & $Z^{NZD,AUD}_0$ \\  \hline\hline
   $\tilde \PP$ & $1.70 $ & $1.56 $ & $1.02 $ & 0.78 & -6.491659& -3.155159 & -3.299338 & 0.7459 & 0.7678  & 0.9715\\
   $ \PP$ & $1.70 $ & $1.56 $ & $1.08 $ & 0.78 & 0& 0 &  0.003650 & 0.7459 & 0.7678  & 0.9715\\
  \hline
\end{tabular}\caption{Parameter estimates and reference measure}\label{tab: parameter estimates}
\end{table}

\noindent\textbf{Other model parameters.} We do not have access to limit order book data and market order activity. These data are useful to estimate the impact of orders on exchange rates, and to estimate the arrival rate and size of liquidity taking orders -- see \cite{MArtinSamFX} for a recent study of order flow in FX markets. Thus, for illustrative purposes we assume that the temporary impact on the mid-exchange rate for the three pairs is
\begin{eqnarray*}
a_x = 5\times 10^{-8}\,,\quad a_y = 1\times 10^{-8}\,,\quad a_z = 1\times 10^{-7}\,.
\end{eqnarray*}
Here we assume that the most liquid pair is that with the deepest limit order book (i.e., everything else being equal, it shows the smallest temporary price impact), followed by the second most liquid,  and finally, the illiquid pair shows the shallowest book, i.e., highest temporary price impact.

Brokerage service fees are given by
			\begin{equation*}
				c_x^{\pm} = \tfrac{1}{2}\,a_x\,,\quad c_y^{\pm} = \tfrac{1}{2}\,a_y\,,\quad c_z^{\pm} = \tfrac{1}{2}\,a_z\,,
			\end{equation*}
and the terminal liquidation penalty parameters are $\alpha_k=a_k\times 10^{6}$.

We assume that the sizes of clients' orders are exponentially distributed with mean $\theta^\pm$. The arrival rates and expected sizes of the clients' orders are given by
			\begin{equation*}
				\lambda^{x,\pm} = 60\,,\quad \lambda^{y,\pm} = 90\,,\quad \lambda^{z,\pm} = 6\,;\quad \theta^{x,\pm} = 2\,,\quad \theta^{y,\pm} = 1\,,\quad \theta^{z,\pm} = 10\,.
			\end{equation*}

\subsection{Behavior of strategy}

\begin{figure}[]
\begin{center}
\includegraphics[scale=0.33]{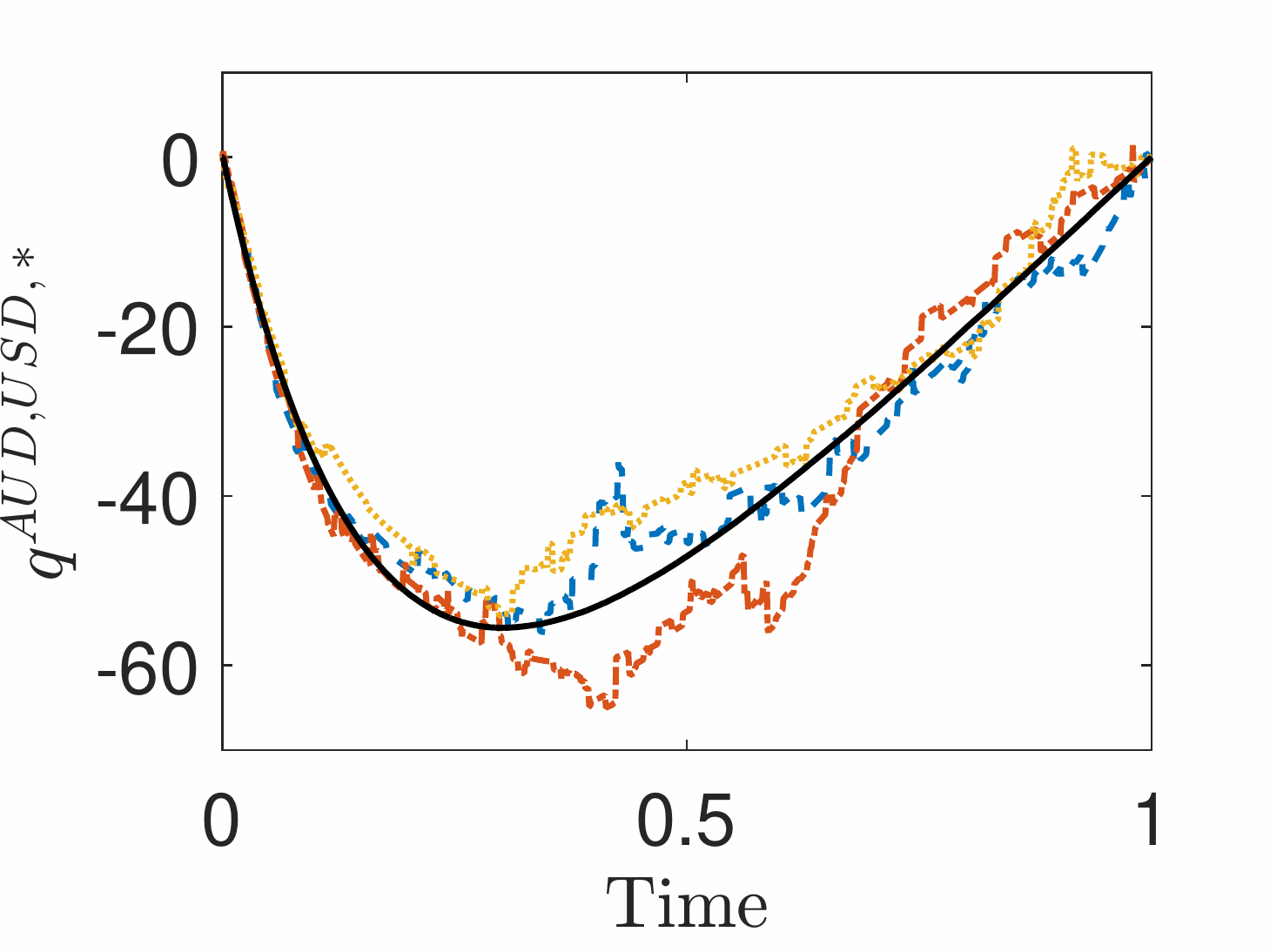}
\includegraphics[scale=0.33]{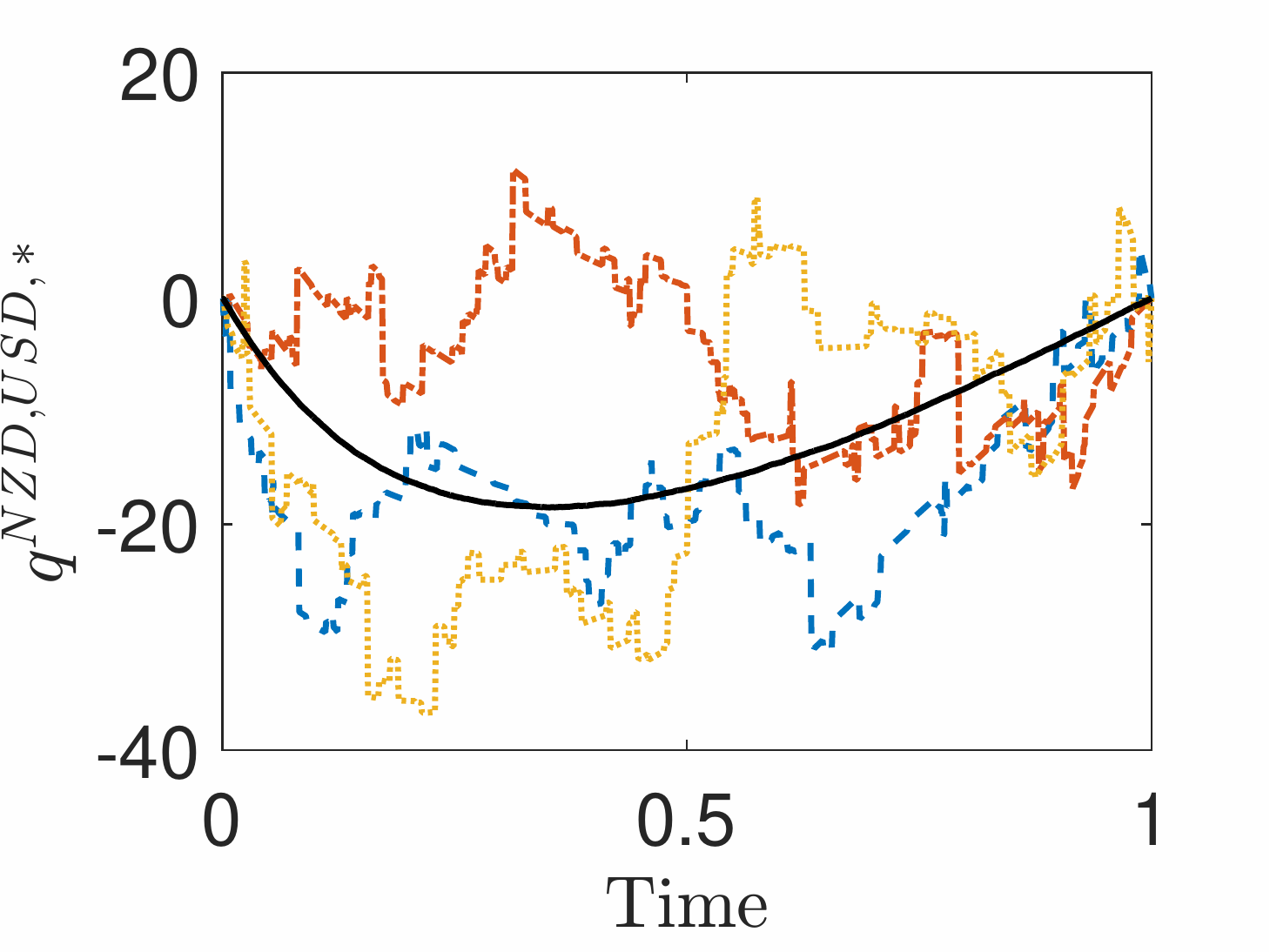}
\includegraphics[scale=0.33]{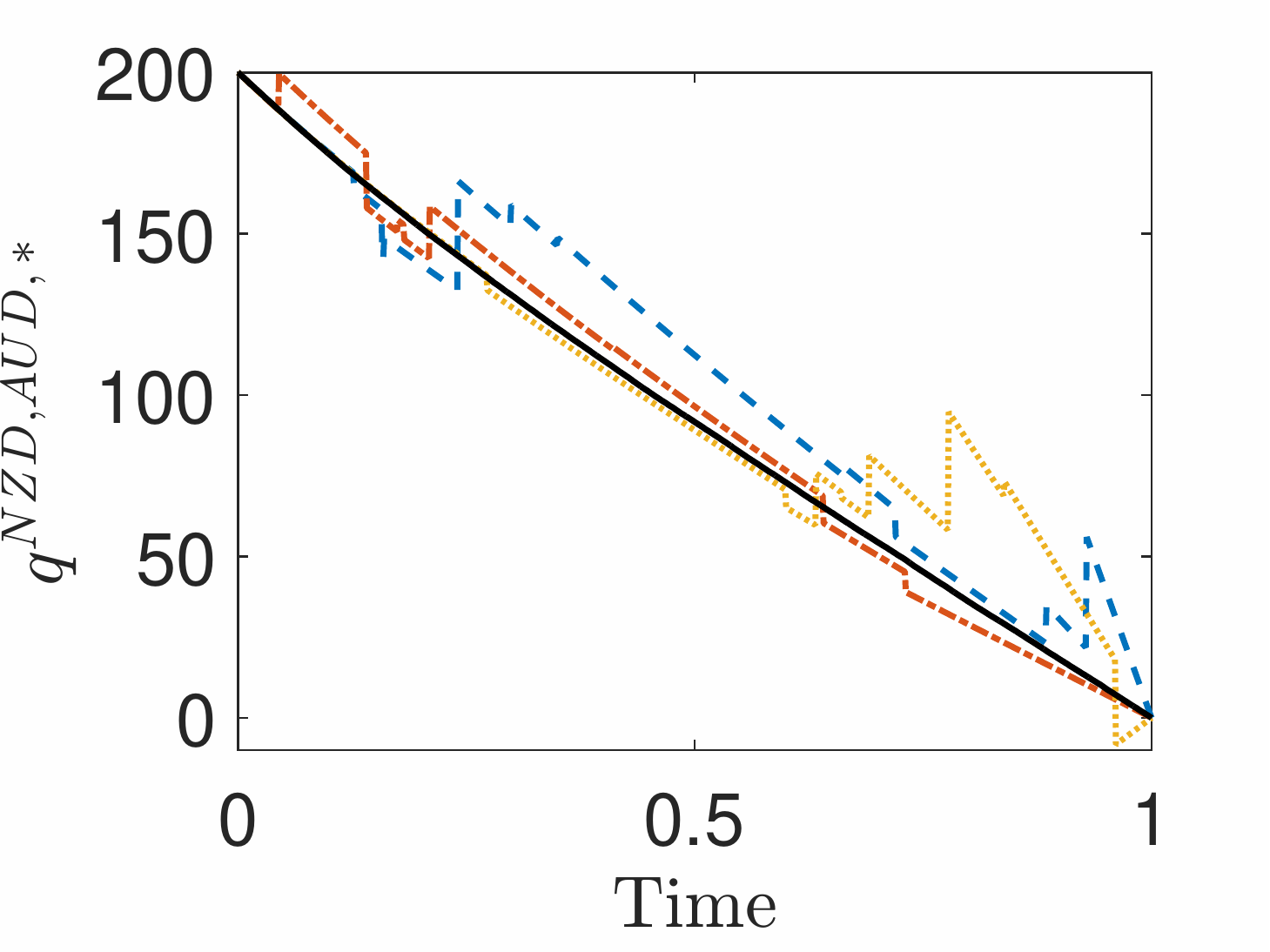}
\end{center}
\caption{Inventory paths for currency pairs, $\varphi = 0.1$. The solid line is the mean of all simulated  paths.}\label{fig: inventory}
\end{figure}

Figure \ref{fig: inventory} shows three inventory paths for each pair  and shows the mean inventory path (solid black line)  for the 10,000 simulations for each pair when the broker's degree of ambiguity aversion is $\varphi = 0.1$.   In the figure, the first picture corresponds to the very liquid pair (AUD, USD), followed by the picture for the liquid pair (NZD,USD), and the last picture is for the illiquid pair. We first discuss the mean inventory paths  and then discuss individual paths.

At the beginning of the trading window, the mean inventory paths  for the very liquid and liquid pairs show how the strategy builds a short position in both pairs while the position in the illiquid pair is gradually liquidated. The initial short positions in the very liquid and illiquid pairs are built to compensate for the broker's initial long position in the illiquid pair. At some point the (mean) strategy reverses the direction of trading in the two liquid pairs because the short positions must be closed by the terminal time. As expected, the strategy relies more on the very liquid  currency pair because trades in this pair have the smallest adverse impact in the quote currency received by the broker.

The sample inventory paths  in Figure \ref{fig: inventory}  show how individual simulations deviate from the mean behavior of the strategy as a result from the broker filling clients' orders.  For example,  the red dash-dot line in the middle picture shows that early on, the strategy  changes from selling the pair (NZD,USD) to purchasing it -- instead of building a short position in the pair (NZD, USD) as shown in the mean path. For that path we also see that the inventory $q^{NZD,USD}$ crosses zero a number of times before reaching the terminal date.  This optimal behavior becomes clear by looking at the speeds of trading and the broker's activity with her pool of clients, both of which we discuss below.

\begin{figure}[]
\begin{center}
\includegraphics[scale=0.33]{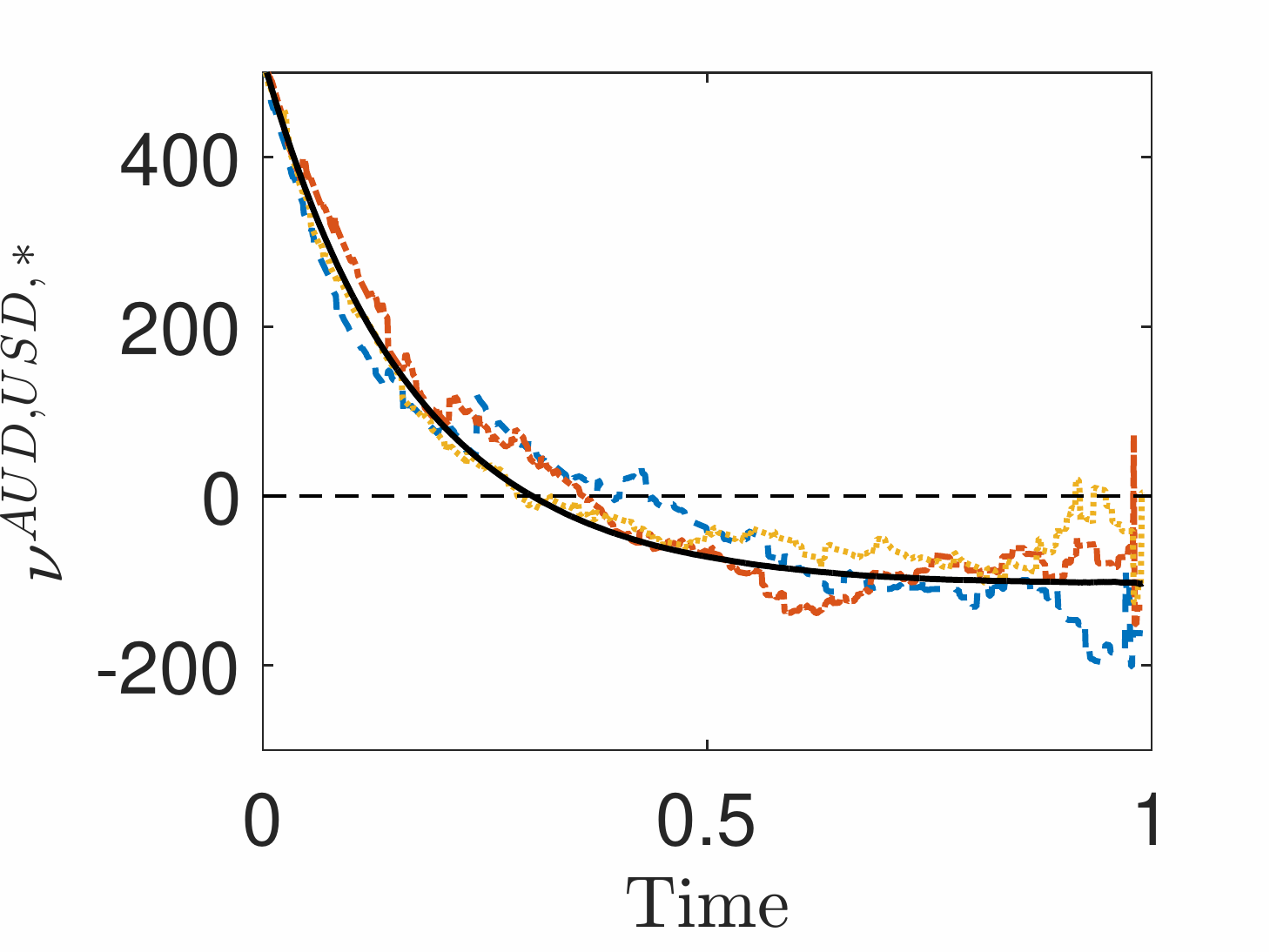}
\includegraphics[scale=0.33]{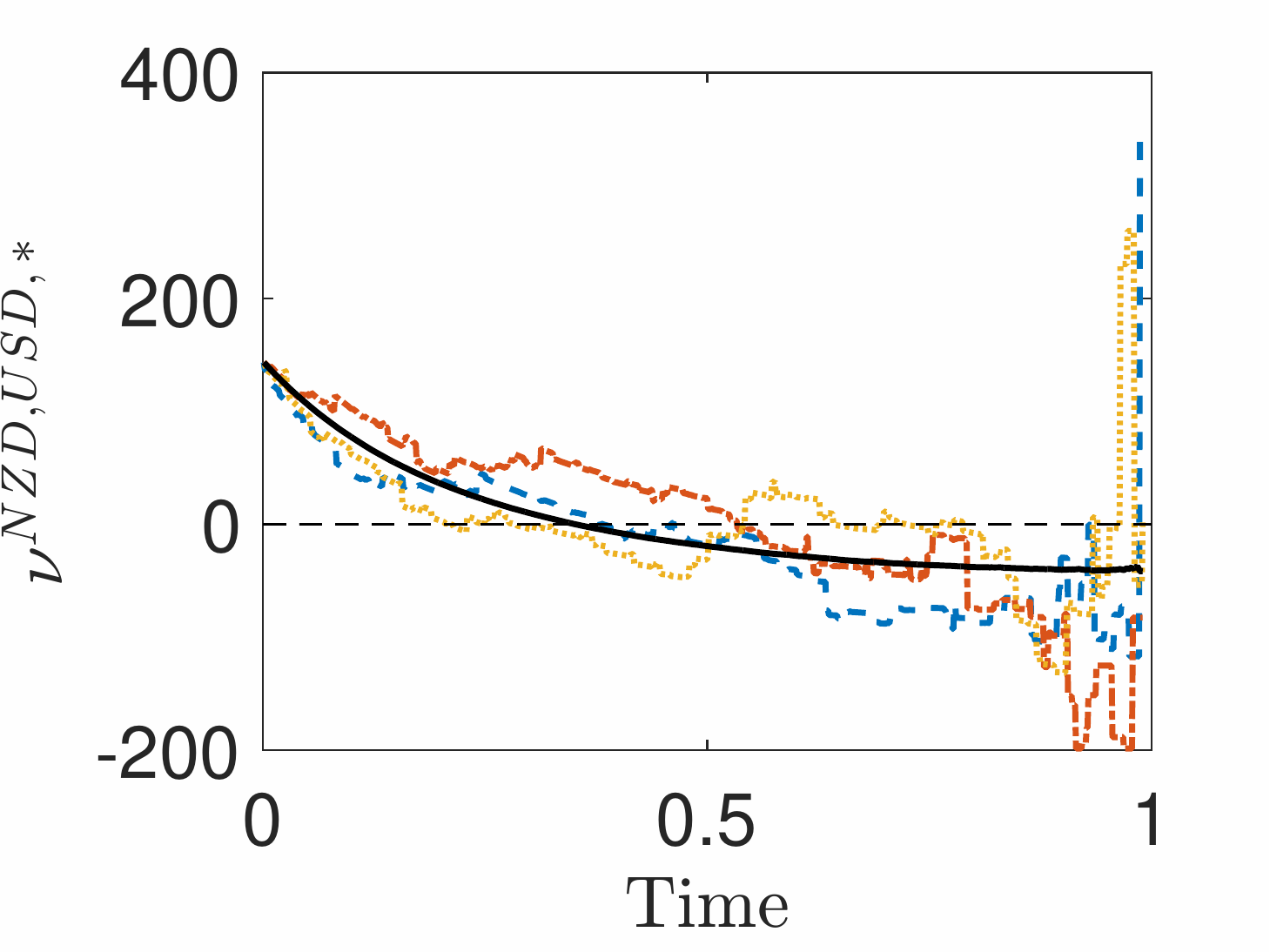}
\includegraphics[scale=0.33]{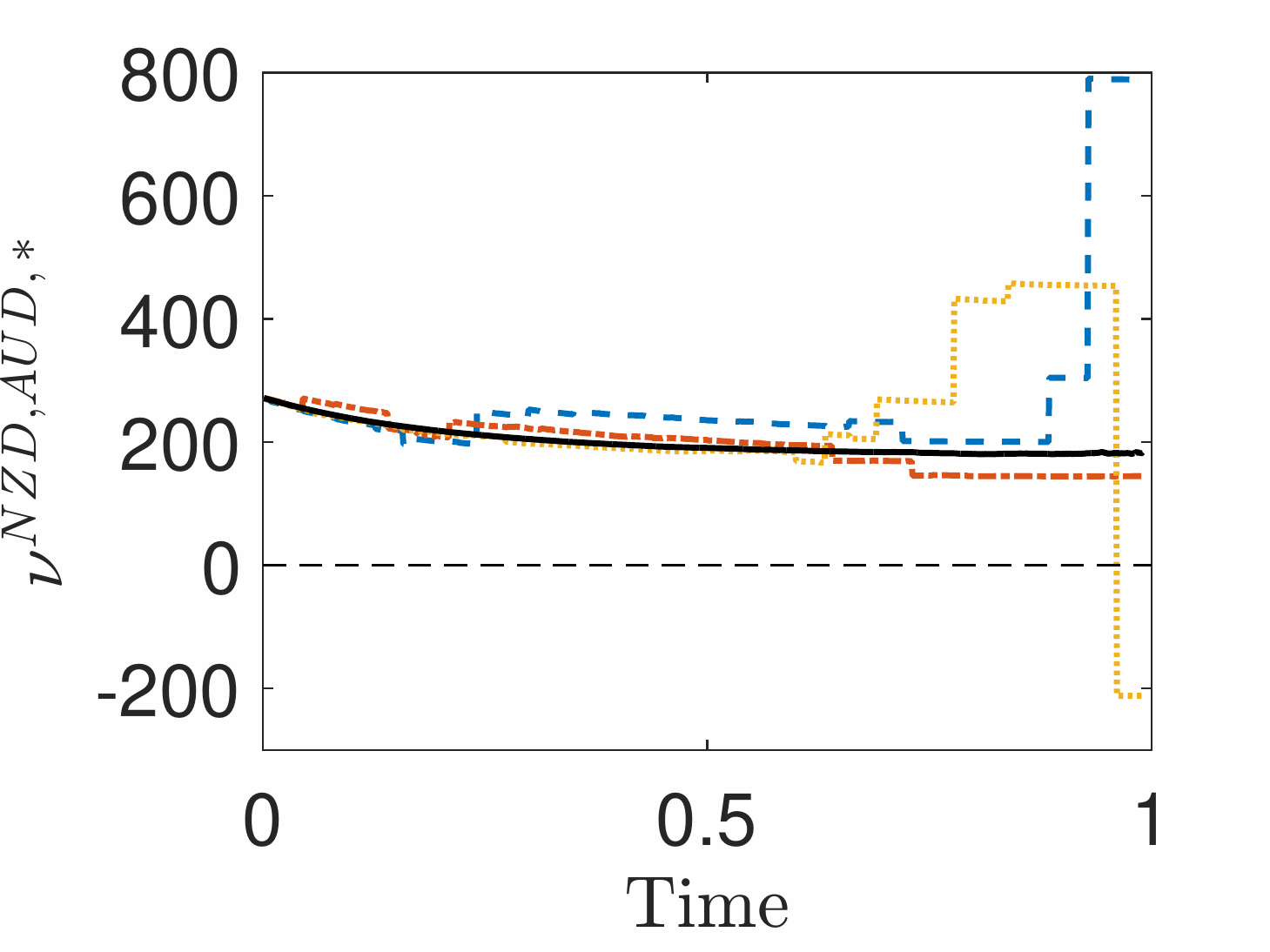}
\end{center}
\caption{Trading speed for each currency pair, $\varphi = 0.1$. Solid line represents the mean trading speed over all paths.}\label{fig: speed of trading}
\end{figure}

Figure \ref{fig: speed of trading} shows the optimal speeds of trading for the triplet. As above, the solid black line represents the average over 10,000 simulations.   Recall that if the continuous rate of trading $\nu^k$ is positive, then the strategy is selling the pair $k$. As already discussed, the strategy relies more on trading the most liquid pair (AUD,USD) because it exhibits the lowest temporary impact on the mid-exchange rate.

A number of factors determines the speed at which the strategy trades the pairs; one of which  is the exogenous order flow from clients. Every time the broker fills a client's order, the inventory in that pair jumps, and this causes the optimal speed of trading in the three pairs to change. There are cases in which the size of the client's order is large enough to reverse the direction of the broker's trading in that pair, i.e., the speed $\nu^k$ changes sign.

\begin{figure}[t]
\begin{center}
\includegraphics[scale=0.33]{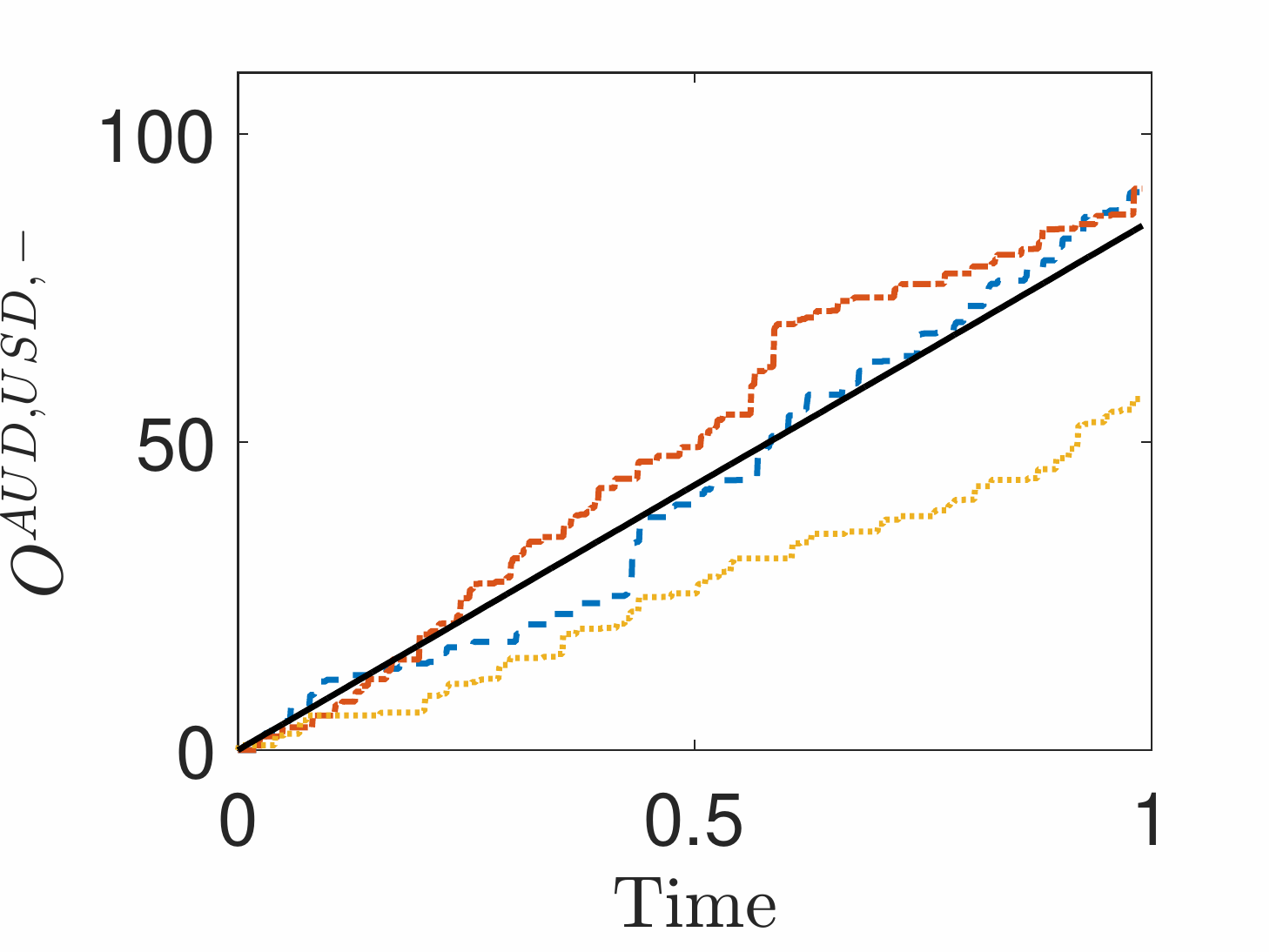}
\includegraphics[scale=0.33]{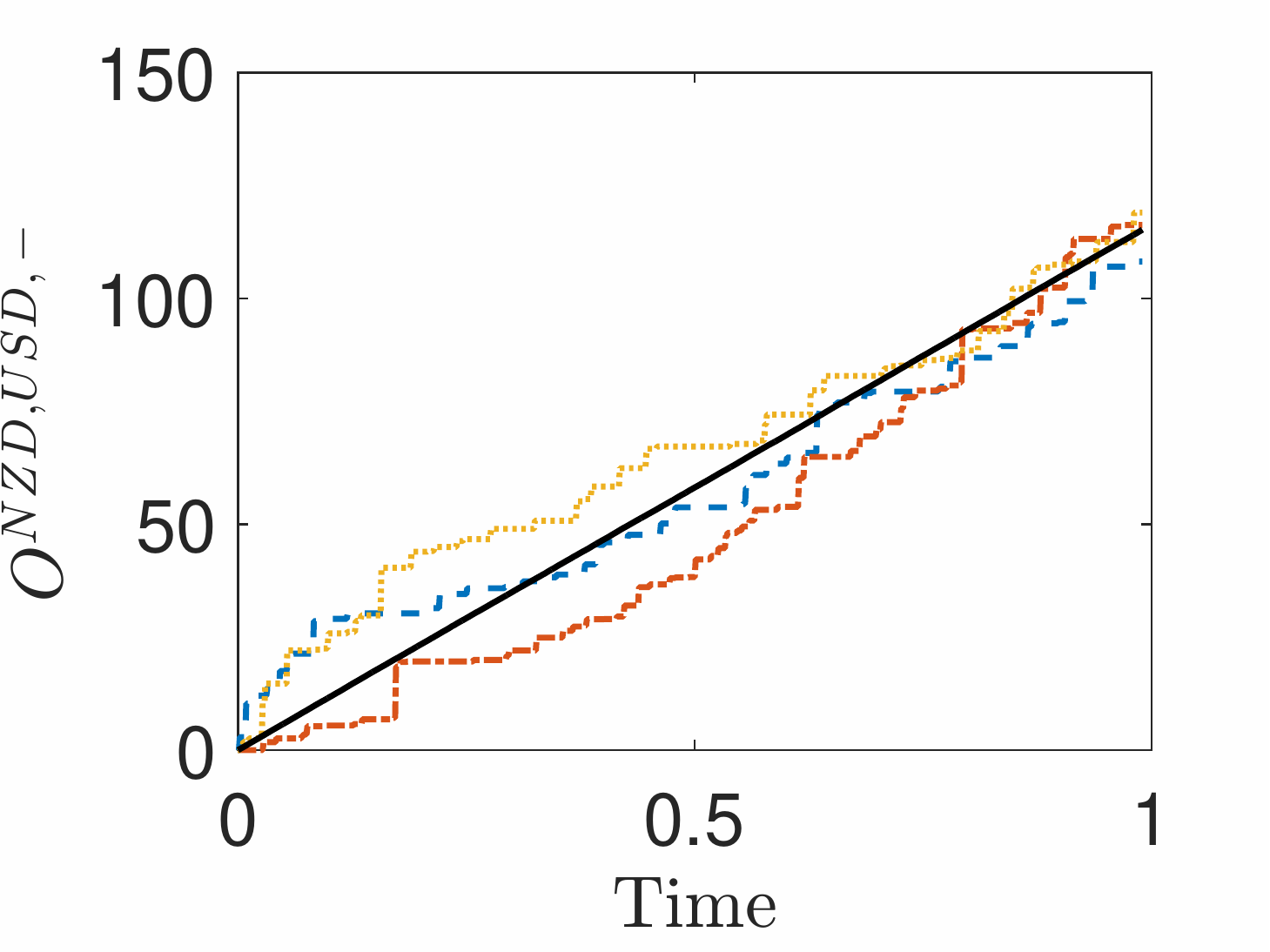}
\includegraphics[scale=0.33]{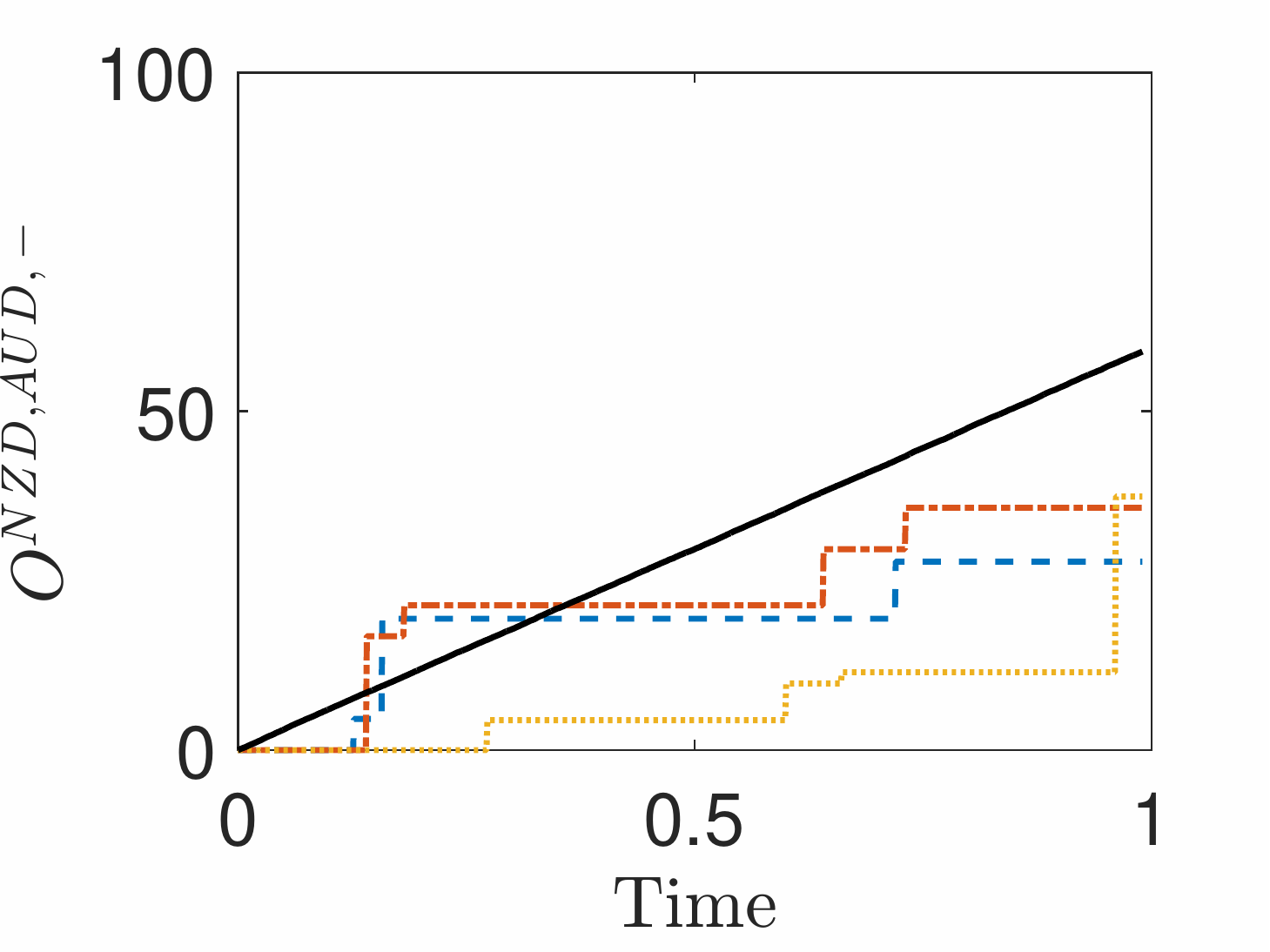}
\end{center}
\caption{Buy order flow from broker's clients, $\varphi=0.1$. Solid line is mean over all paths.}\label{fig: buy order flow}
\end{figure}

Figure \ref{fig: buy order flow} shows buy order flow (in the interest of space we do not show sell order flow) from the pool of clients which are filled by the broker. At the beginning of the trading window, the broker quickly fills around 20 lots of the currency pair (NZD,USD), which is above the average fills shown by the solid line in the second picture of the figure. Thus, the broker's inventory $q^{NZD,USD}$ drops very quickly and the speed at which the broker sells that pair in the Exchange also decreases.

The broker's degree of ambiguity aversion also affects the speeds of trading.  When the broker is ambiguous to the drift of the  mid-exchange rates, the trading strategy assumes that the dynamics of the rates are given by \eqref{eqn:dXtGBM2 under Q}, \eqref{eqn:dYtGBM2 under Q}, \eqref{eqn:dZtGBM2 under Q}, which for convenience we repeat here
\begin{eqnarray*}
dX_t &=& X_t\left(\sigma_x\, \kappa^{x,*}_t\,dt  + \sigma_x\, dW_t^{x,\bkappa^*} \right)\,,\\
dY_t &=& Y_t\left(\sigma_y\, \kappa^{y,*}_t\,dt  + \sigma_y\, dW_t^{y,\bkappa^*}\right) \,,\\
dZ_t &=& {Z_t}\left(\left(\mu_z +\sigma_z\,\kappa^{z ,*}_t\right)\,dt + \sigma_z\, dW_t^{z,\bkappa^*}\right)\,,
\end{eqnarray*}
where
\begin{equation}\label{eqn:drift adjustment of z}
\kappa^{z ,*}_t = \frac{\sigma_x\,\kappa^{x ,*}_t - \sigma_y\,\kappa^{y ,*}_t}{\sqrt{\sigma_x^2 + \sigma_y^2 - 2\,\rho\,\sigma_x\,\sigma_y}}\qquad\text{ and }\qquad dW_t^{z,\bkappa^*} = \frac{\left(\sigma_x\, dW_t^{x,\bkappa^*} - \sigma_y\, dW_t^{y,\bkappa^*}\right)}{\sqrt{\sigma_x^2 + \sigma_y^2 - 2\,\rho\,\sigma_x\,\sigma_y}}\,.
\end{equation}

In Figure \ref{fig: drift adjestment} we show the drift adjustments $\kappa^{x,*}_t$, $\kappa^{y,*}_t$, $\kappa^{z,*}_t$ that stem from ambiguity aversion. For each mid-exchange rate, the black solid line shows the mean of the drift adjustment for all the simulations. For the three pairs, the mean  of $\kappa^{k,*}$ is always negative   and gradually increases to zero by the end of the trading window. When the drift adjustment is negative, the strategy assumes that the mid-exchange rate is decreasing (on average), and therefore  it is optimal to sell the pair. Thus, everything else being equal, the effect of ambiguity aversion is to increase the speed of trading.

To gain insights into the effect of ambiguity aversion on the broker's optimal strategy we also discuss the results of the strategy if, everything else being equal,  the broker's initial inventory is $Q^x_0 = 0\,, Q^y_0 = 0\,, Q^z_0 = -200$, i.e., the initial position in the illiquid pair is short 200 contracts.  The equivalent to Figures \ref{fig: inventory}, and \ref{fig: speed of trading},  which describe the inventory, and speed of trading, are similar to those above but with the opposite sign. For example, the   mean inventory paths for the two liquid pairs show that the strategy builds long positions and gradually unwinds them.

Similarly, the drift adjustment when the broker's initial inventory is $Q^x_0 = 0\,, Q^y_0 = 0\,, Q^z_0 = 200$ is as that shown in Figure \ref{fig: drift adjestment}, but with the opposite sign. That is, for all three pairs, the mean adjustment $\kappa^{k,*}$ is positive and gradually decays to zero as the strategy approaches the terminal date.

Therefore, when the broker's initial position in the illiquid pair is short (resp. long), the ambiguity averse broker devises a strategy where the drift in the mid-exchange rates, relative to the model under the reference measure $\PP$,  is larger (resp. smaller). In the two cases discussed above ($Q^z_0=-200$ and $Q^z_0=200$), the drift of the pairs $x$ and $y$ is zero under the reference measure and the drift adjustment under the optimal measure is positive (resp. negative) if the initial  position in the illiquid pair is short (resp. long).  Hence, ambiguity aversion increases the speed at which the strategy enters positions in the liquid pairs of the triplet. The magnitude of the increase in the speed is linear in the  ambiguity aversion parameter $\varphi$.

\begin{figure}[t]
\begin{center}
\includegraphics[scale=0.33]{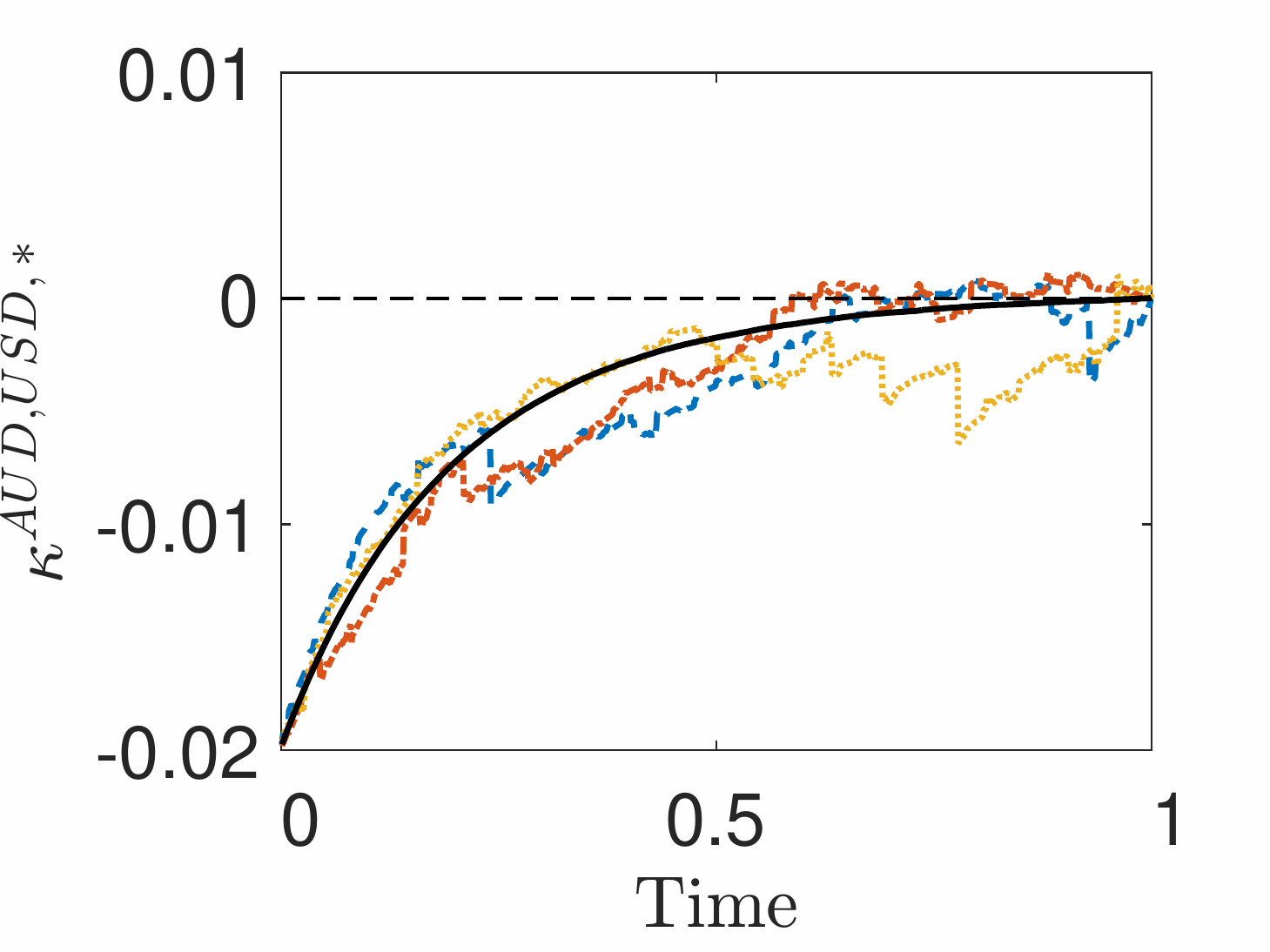}
\includegraphics[scale=0.33]{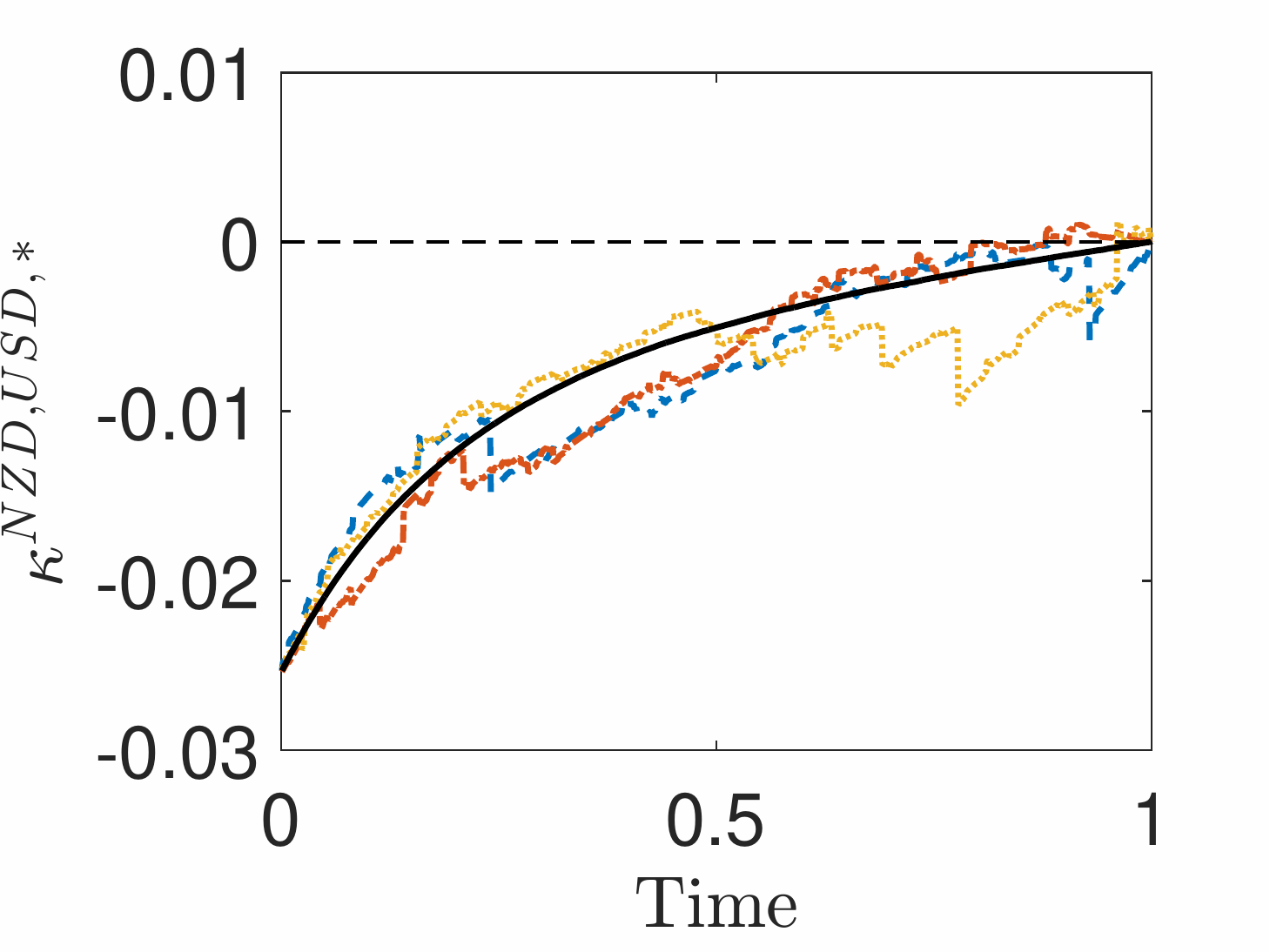}
\includegraphics[scale=0.33]{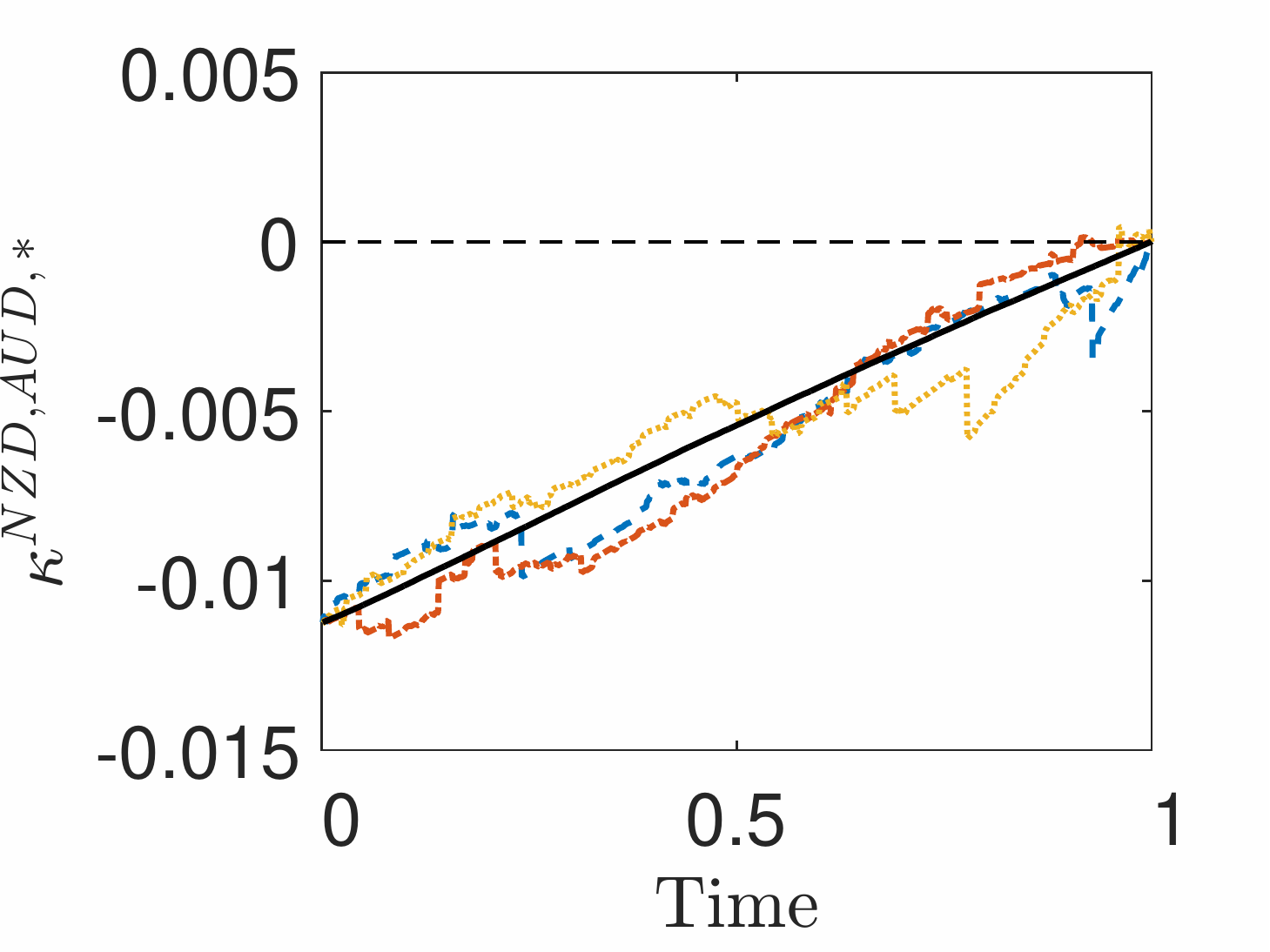}
\end{center}
\caption{Drift adjustments as result of ambiguity aversion, $\varphi=0.1$. Solid line is mean over all paths.}\label{fig: drift adjestment}
\end{figure}

\subsection{Added value of trading all pairs in the triplet}
The left panel in Figure \ref{fig:PnL} shows the profit and loss (P\&L) of the strategy  when the broker trades in the triplet (gray-fill histogram) or trades only the illiquid pair (white-fill histogram). For comparative purposes, when the broker trades the triplet of currency pairs, we assume that she also trades with her clients, so  $\lambda^{x,\pm} = 60$, $\lambda^{y,\pm} = 90$, $\lambda^{z,\pm} = 6$, and $\varphi = 0.1$. When the broker only trades the illiquid pair we employ the parameters   $\lambda^{x,\pm} = 0$, $\lambda^{y,\pm} = 0$, $\lambda^{z,\pm} = 6$, and $\varphi = 0$.\footnote{Recall that if $\varphi>0$ the broker finds it optimal to trade in all currency pairs regardless of the initial inventory in each pair.} The mean P\&L of the strategy that trades in the triplet is higher (solid vertical line) than the P\&L of the strategy when only the illiquid pair is traded (dash vertical line) -- the difference is approximately  $88.88$ USD per lot of  currency 1. The mean P\&L when the broker only trades the illiquid pair is $7.4562\times 10^5$, and when the broker trades in all the pairs the mean P\&L is $7.4571\times 10^5$.

The right panel in Figure \ref{fig:PnL} shows the mean P\&L against the standard deviation of the P\&L of the strategy for $\varphi\in[0,50]$ and $\lambda^{x,\pm} = 60$, $\lambda^{y,\pm} = 90$, $\lambda^{z,\pm} = 6$.  The lowest mean P\&L, which also corresponds to the highest standard deviation of the P\&L, corresponds to $\varphi=0$, i.e., when the broker does not doubt the reference model. As the value of the ambiguity aversion parameter $\varphi$ increases, the mean P\&L increases and the standard deviation of P\&L decreases up to a point ($\varphi=36$) after which the mean P\&L starts to decrease.  Thus, there is a range of the ambiguity aversion parameter where increasing the broker's ambiguity aversion improves the tradeoff between mean and standard deviation of the profitability of the strategy.

To investigate the additional value that stems from the opportunity to trade two liquid currency pairs, we assume the broker does not trade  with her clients, i.e.,  $\lambda^{x,\pm} =\lambda^{y,\pm} = \lambda^{z,\pm} = 0$, and define the relative improvement of P\&L as
\begin{eqnarray*}
	\Delta_{PnL}\% &=& \frac{PnL_\varphi - PnL_0}{PnL_0}\,,
\end{eqnarray*}
where $PnL_0$ and $PnL_\varphi$ are the P\&L with $\varphi = 0$ and $\varphi > 0$, respectively. We compute a `Sharpe ratio' of the relative improvement of the P\&L as follows:
\begin{eqnarray*}
	S_{\Delta_{PnL}\%} &=& \frac{\EE^{\tilde \PP}\,\left[\Delta_{PnL}\%\right]}{\sqrt{Var^{\tilde \PP}\left[\Delta_{PnL}\%\right]}}\,.
\end{eqnarray*}

The left-hand picture of Figure \ref{fig:ImprovePnL} shows $\log(1+S_{\Delta_{PnL}\%})$ against $\log(1+\varphi)$ for  $\varphi\in[0,50]$. When $\varphi = 0$, we see that $\Delta_{PnL}\%\equiv 0$, i.e., the broker does not doubt the reference model and only trades in the illiquid pair. As the level of ambiguity aversion increases, the Sharpe ratio increases and then decreases. Thus, there is a range of $\varphi$ where increasing the broker's ambiguity aversion improves the tradeoff between mean and standard deviation of the value of trading liquid pairs relative to trading the illiquid pair alone -- the maximum Sharpe ratio is $0.3421$ for $\varphi=16$.

The right-hand picture of Figure \ref{fig:ImprovePnL} shows the probability that relative P\&L improvement is greater than $x$, i.e. $P_{\Delta_{PnL}\% }(x) = \tilde \PP\left(\Delta_{PnL}\% > x\times 10^{-2}\right)$. When $\varphi = 0$, the relative P\&L improvement has a point probability mass of $1$ at $0$. The  dotted curve is for $\varphi = 0.1$, and the solid curve  is for $\varphi = 16$. For both curves, the probability of positive relative P\&L improvement is around $60\%$. The $\Delta_{PnL}\%$ has $\textrm{50}^{\textrm{th}}$-percentiles equal to 0.01\% and 0.03\% for $\varphi = 0.1$ and $\varphi = 16$, respectively. Thus, when the broker trades liquid pairs, she achieves (compared with the results obtained when the broker only trades the illiquid pair) a higher P\&L of $74.59$ USD ($\varphi = 0.1$) and $223.77$ USD ($\varphi = 16$) in half of the 10,000 trade simulations.

\subsection{Sensitivity to terminal penalty and ambiguity aversion}
In addition to Proposition \ref{prop6}, to illustrate the effect of the terminal liquidation penalty on: inventory, trading speed, and drift adjustment, we perform additional simulations. We assume that the broker does not trade with her pool of clients ($\lambda^{k,\pm}=0$) and employs the following terminal liquidation penalties: $\alpha_k = 1\times a_k$, $\alpha_k = 2.5\times a_k$, and $\alpha_k = 10^6\times a_k$, $k\in\{x,y,z\}$.

Figure \ref{fig: average inventory different terminal liquidation penalty} shows average inventory paths for each currency pair when $\varphi = 0.1$  and initial inventory $Q^x_0 = Q^y_0 = 0$, $Q^z_0 = 200$. As expected, as the value of the terminal liquidation penalty decreases, the  strategy first sells and then buys the two liquid pairs (AUD, USD), (NZD, USD), and sells the pair (NZD, AUD) throughout the trading period. Finally, when the value of  the  liquidation penalty is sufficiently low,  it is optimal for the strategy to have non-zero inventories at the end of trading of the trading horizon.

Figure \ref{fig: average speed of trading different terminal liquidation penalty} shows the average of the trading speed for each currency pair. As the value of the liquidation penalty parameter decreases we observe the following: the broker is willing to hold more inventory; the sell (buy) speed increases (decreases) in approximately the first (second) half of the trading period  for the liquid pairs (AUD,USD) and (NZD,USD); and   the sell speed decreases for the illiquid pair (NZD, AUD).

Figure \ref{fig: average drift adjestment different terminal liquidation penalty} shows the average drift adjustment, which is a result of the broker's ambiguity aversion, for each currency pair for various levels of the terminal penalty parameter. Observe that at the beginning of the trading horizon, the impact of the terminal liquidation penalty on the drift adjustment  is very small. As time evolves, the effect of the terminal penalty differs for each currency pair. For example,  the mean drift adjustment for the pair (NZD, AUD) is negative during the entire trading horizon, so it is always optimal to sell. Recall that the initial inventory in the pair (NZD, AUD) is 200 lots, so the impact of the negative drift adjustment on the strategy is to sell quicker than otherwise.

Figure \ref{fig:PnL different terminal liquidation penalty} shows the P\&L of a strategy when $\alpha_k = 1\times a_k$ (white-fill histogram) and $\alpha_k = 10^6\times a_k$ (gray-fill histogram), and $\lambda^{k,\pm}=0$, $k\in\{x,y,z\}$. The mean and standard deviation of the P\&L  for the strategies where the  terminal liquidation penalty is low are $7.4272\times 10^5$ and $501.96$ USD respectively. When the  terminal liquidation penalty is large, the mean and standard deviation of the P\&L per lot are $7.4571\times 10^5$ and $496.05$ USD, respectively. The average cost of executing large inventories at time $T$ is $2.9922\times 10^3$ USD.\footnote{We compute the costs of unwinding terminal inventory as follows.  Let $Q_T^k$ represent the outstanding inventory at the terminal time $T$, and let $\Delta_T$ be the time interval to unwind $Q_T^k$. The contribution to the P\&L of unwinding the inventory is
\begin{eqnarray}\label{eqn: terminal liquidation PnL}
	\sum_{k \in \{x,y,z\}}\,Q_T^k\times\max\left(\widehat{K}_T\,\left(1 - a_k\,\tfrac{Q_T^k}{\Delta_T}\right),\,0\right)\,,
\end{eqnarray}
where $\widehat{K}_T = X_T\,1_{\{k\in\{x,z\}\}} + Y_T\,1_{\{k=y\}}$ and $\Delta_T = 10^{-3}$ hours, which is the time step  we use in the simulations.}

Finally, we investigate the behavior of the optimal strategy when the value of the ambiguity aversion parameter is high. We let $\varphi \in \{0.1,1,10\}$ and assume that the broker ensures full liquidation by the end of the terminal horizon, so  $\alpha_k = 10^6\times a_k$, $k\in\{x,y,z\}$ and  $\lambda^{k,\pm}=0$, $k\in\{x,y,z\}$, i.e., the broker does not trade with her pool of clients. Below, in Figures \ref{fig: average inventory different ambiguity aversion}, \ref{fig: average speed of trading different ambiguity aversion},   \ref{fig: average drift adjestment different ambiguity aversion},  the black, blue, and red solid lines represent cases of $\varphi = 0.1$, $1$, $10$, respectively. Figure \ref{fig: average inventory different ambiguity aversion} shows the average inventory paths for the three currency pairs. Figure \ref{fig: average speed of trading different ambiguity aversion} shows average trading speeds in the currency pairs. Finally, Figure \ref{fig: average drift adjestment different ambiguity aversion} shows the average of the drift adjustment for each currency pair.

\begin{figure}[t]
\begin{center}
\includegraphics[scale=0.33]{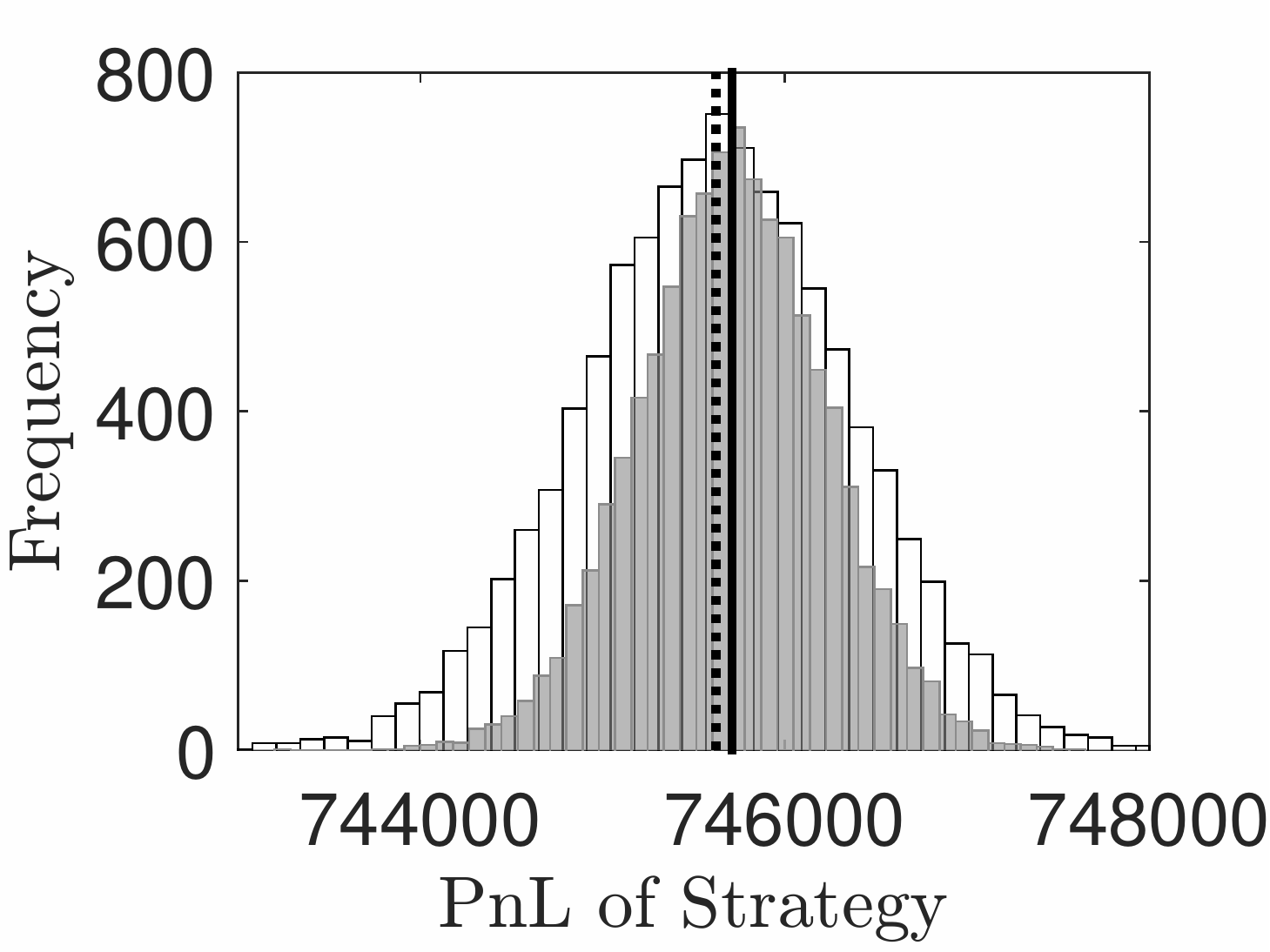}
\includegraphics[scale=0.33]{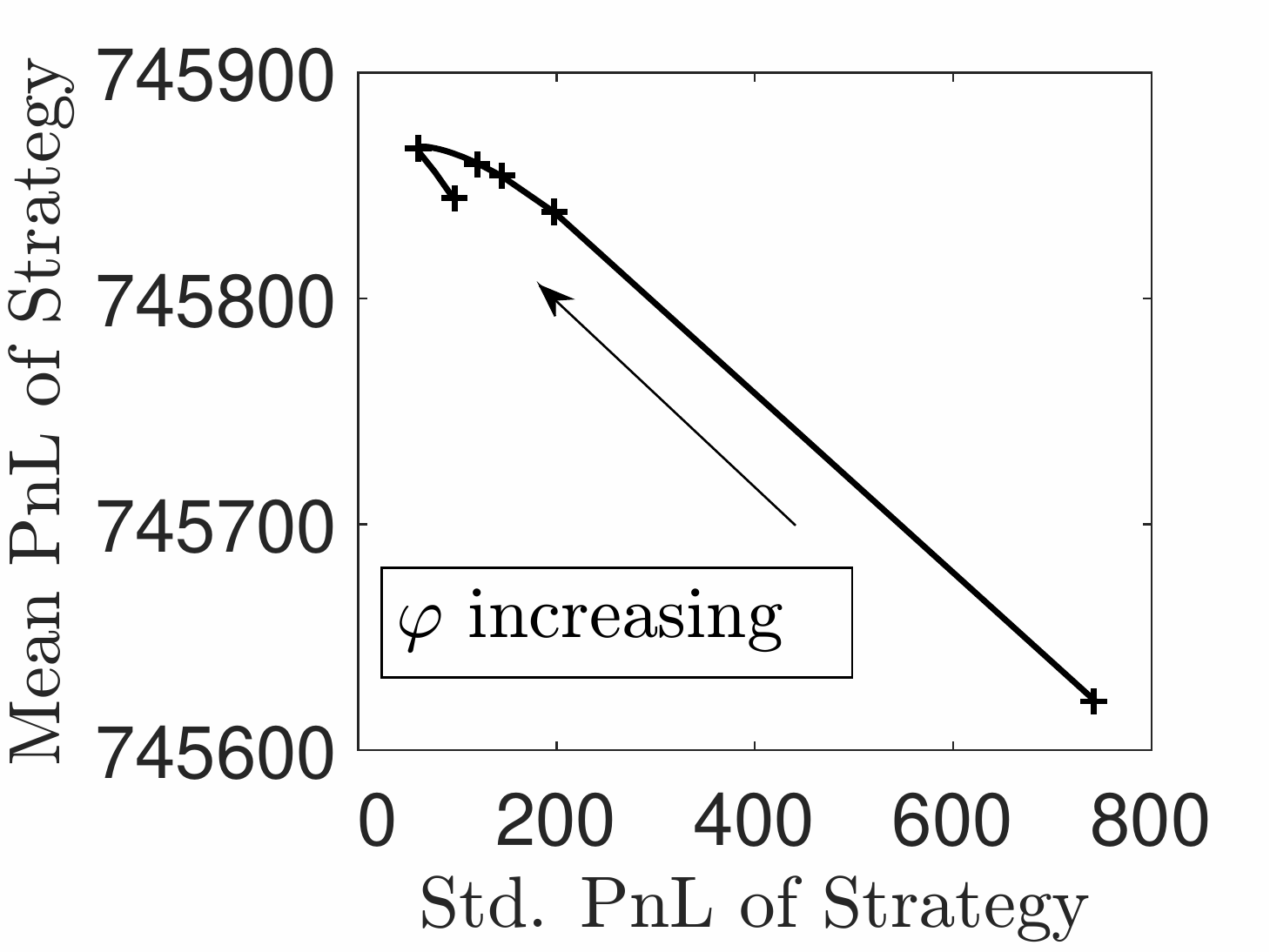}
\end{center}
\caption{Left panel: P\&L of strategy trading in triplet \tianyi{($\varphi = 0.1$, $\lambda^{x,\pm} = 60$, $\lambda^{y,\pm} = 90$, $\lambda^{z,\pm} = 6$)} or in illiquid pair only \tianyi{($\varphi = 0$, $\lambda^{x,\pm} = 0$, $\lambda^{y,\pm} = 0$, $\lambda^{z,\pm} = 6$).} Right panel: mean and standard deviation of P\&L for  \tianyi{$\varphi\in[0,50]$, $\lambda^{x,\pm} = 60$, $\lambda^{y,\pm} = 90$, $\lambda^{z,\pm} = 6$.}}\label{fig:PnL}
\end{figure}

\begin{figure}[t]
	\begin{center}
		\includegraphics[scale=0.33]{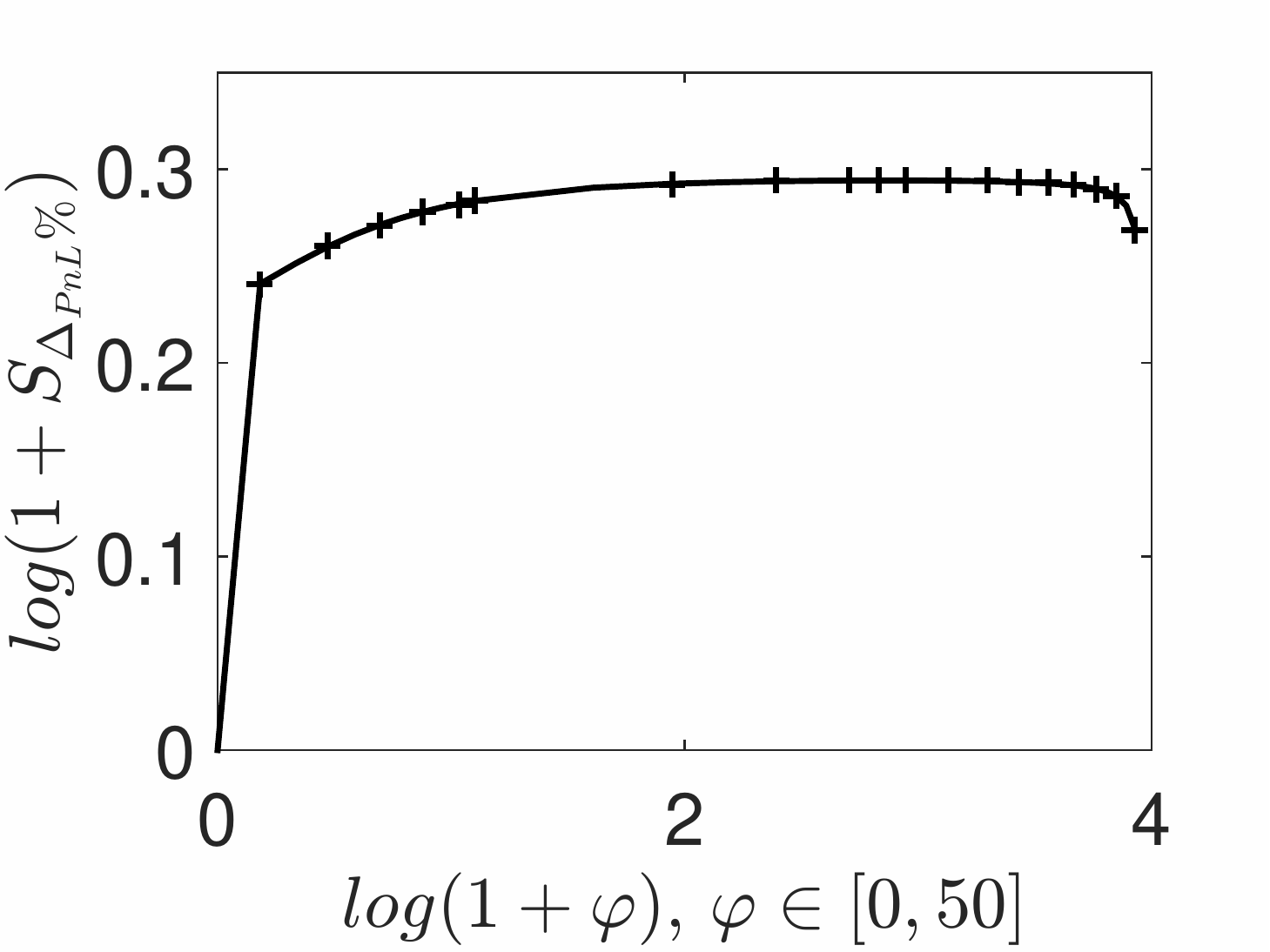}
		\includegraphics[scale=0.33]{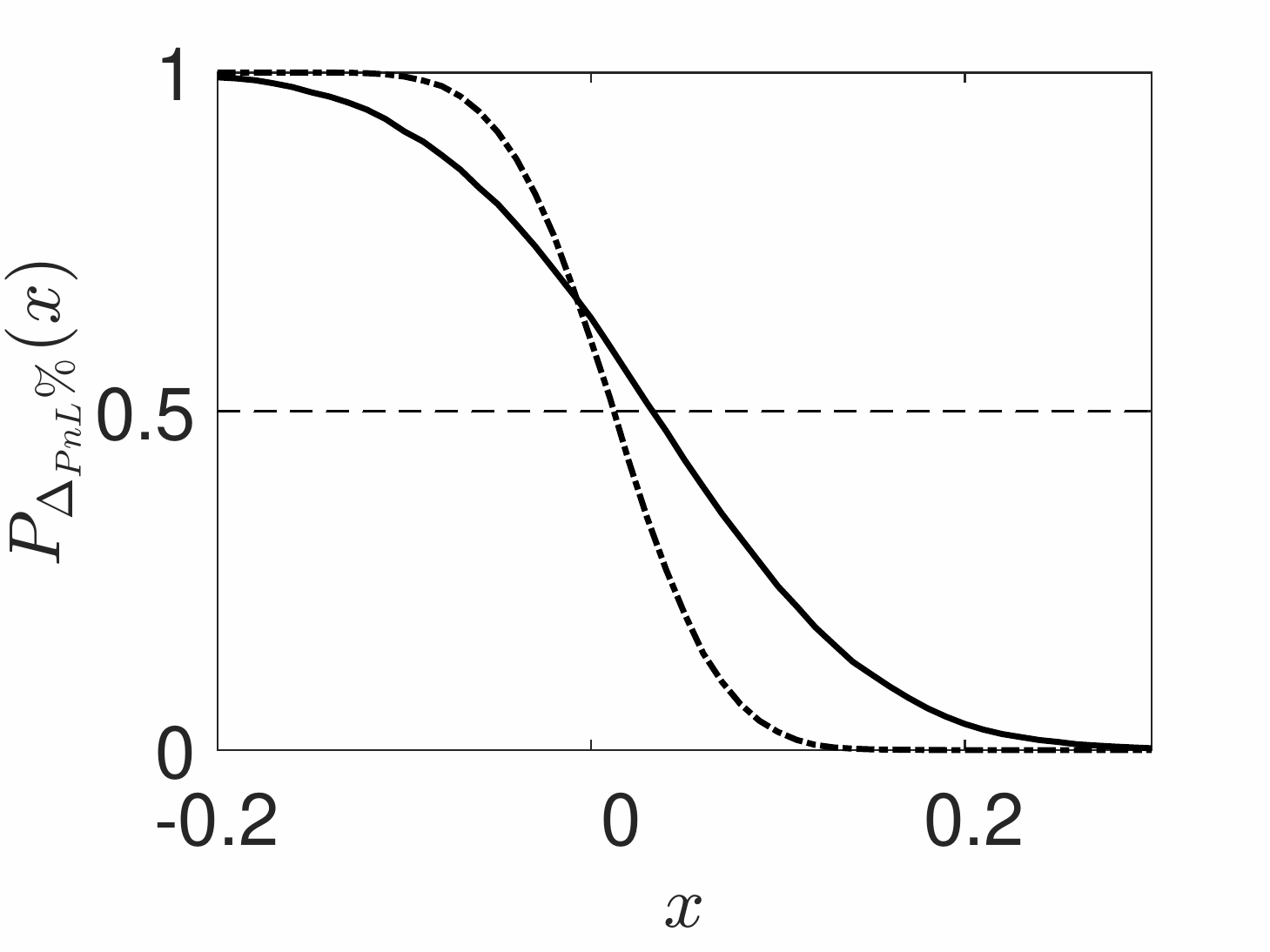}
	\end{center}
	\caption{\tianyi{Left panel: Sharpe ratio of relative P\&L improvement of strategy trading in triplet for a range of $\varphi$. Right panel: The probability of relative P\&L improvement greater than a given value $x$ in percentage, the dotted line is for $\varphi = 0.1$ and the solid line is for $\varphi = 16$.}}\label{fig:ImprovePnL}
\end{figure}

\begin{figure}[]
	\begin{center}
		\includegraphics[scale=0.33]{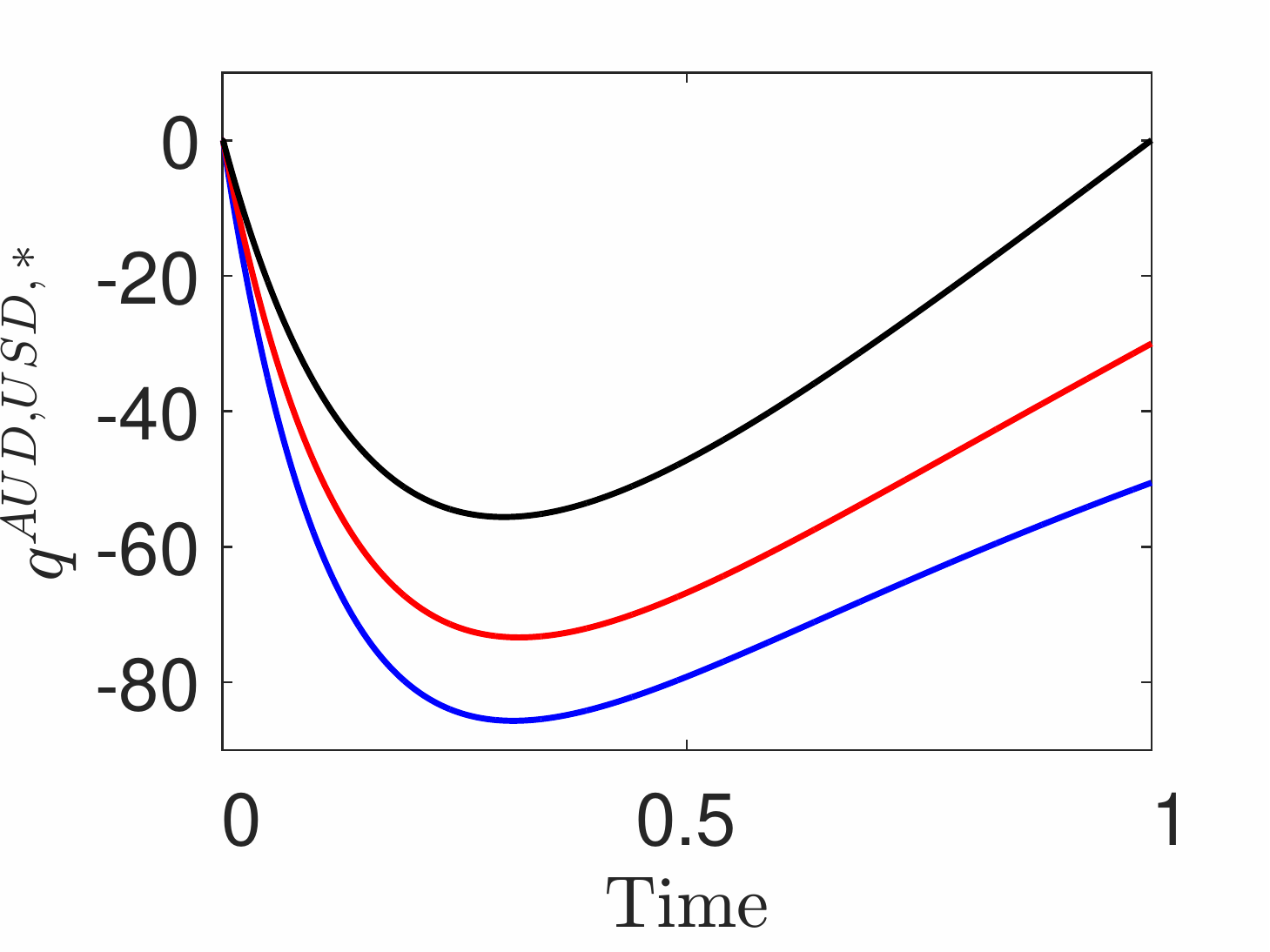}
		\includegraphics[scale=0.33]{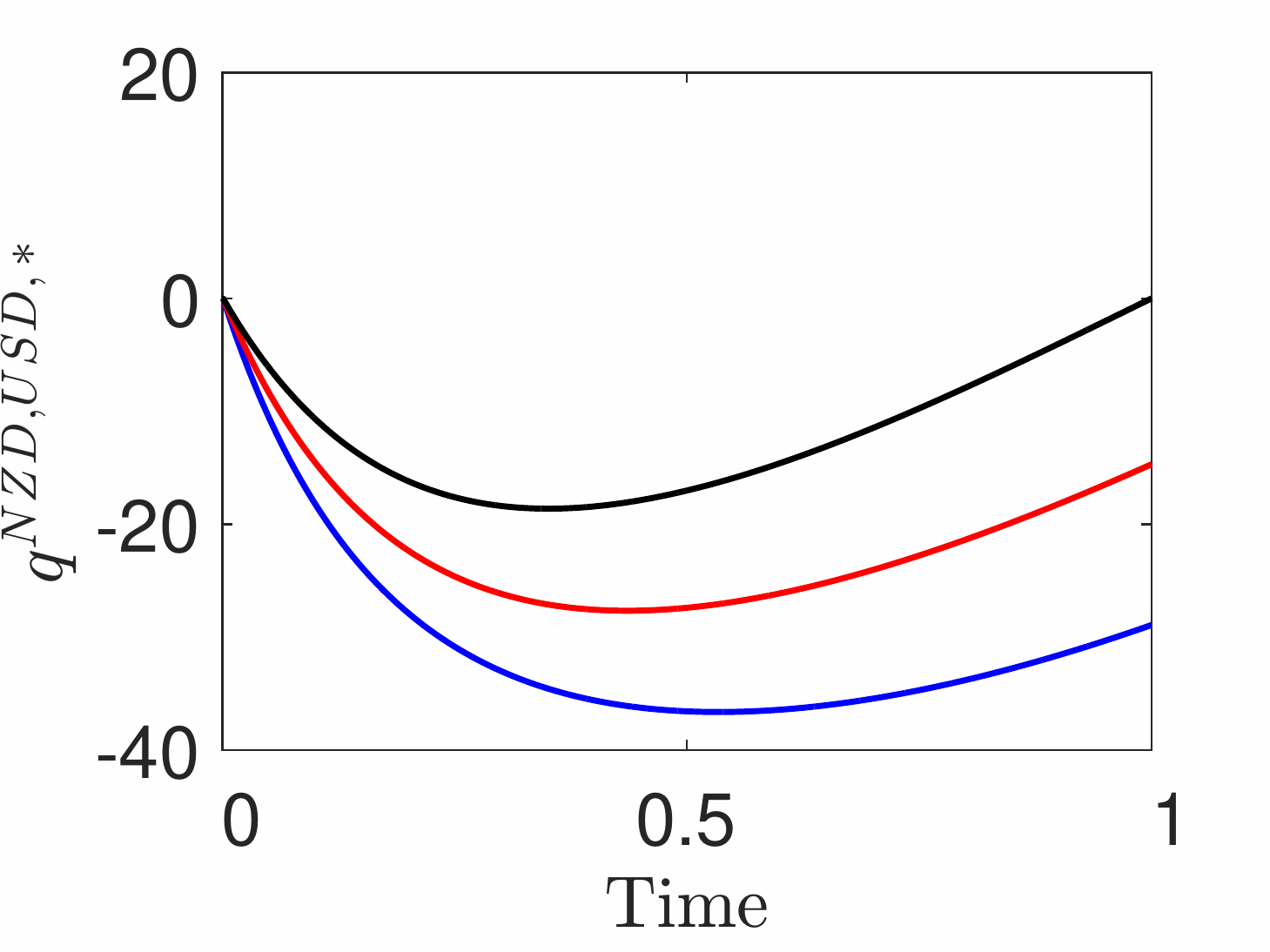}
		\includegraphics[scale=0.33]{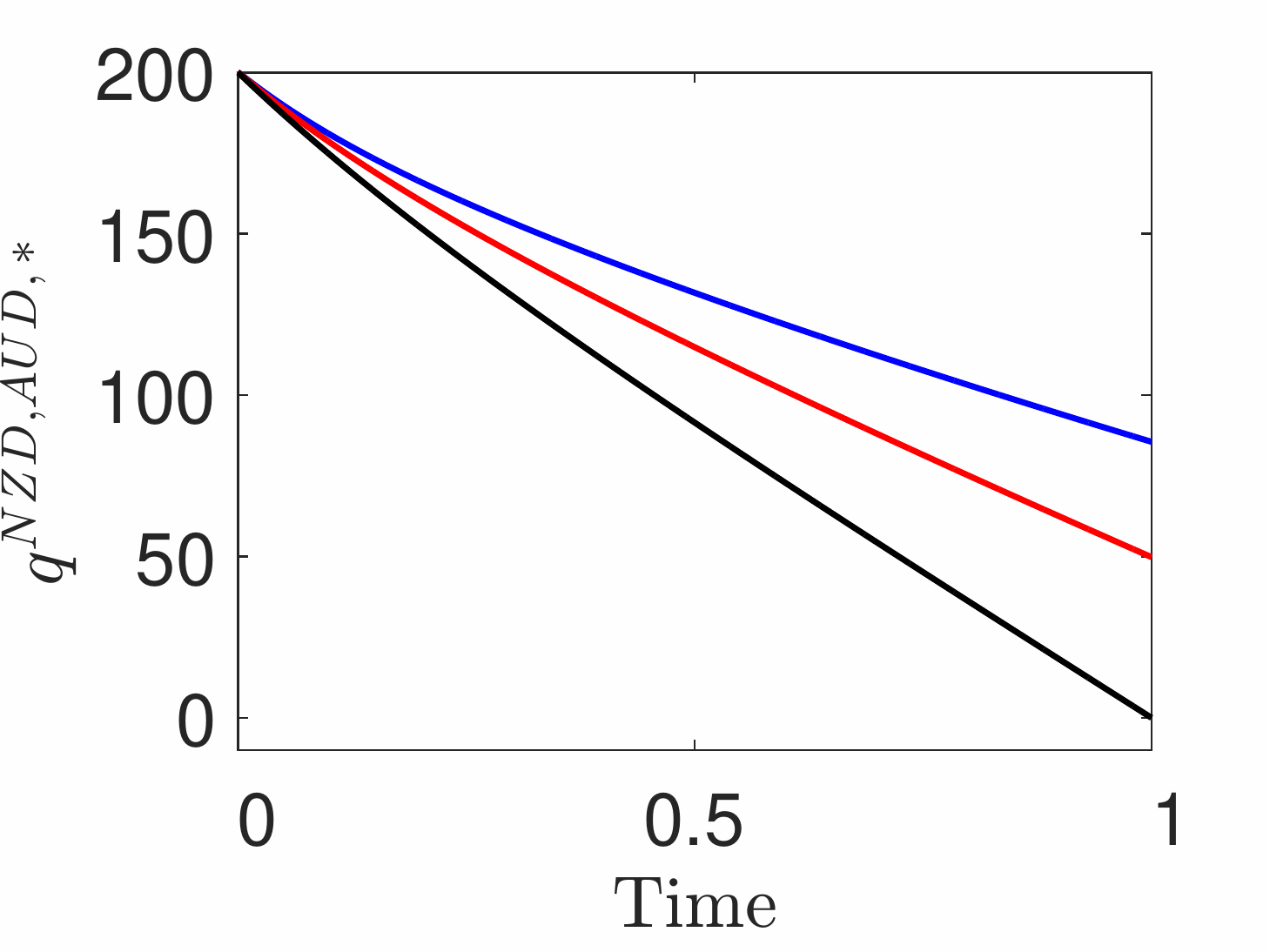}
	\end{center}
	\caption{Average inventory over all simulated paths for currency pairs, $\varphi = 0.1$, $\lambda^{k,\pm} = 0$, $k\in\{x,y,z\}$. The blue, red, black solid lines are for $\alpha_k = 1\times a_k$, $\alpha_k = 2.5\times a_k$, and $\alpha_k = 10^6\times a_k$, $k\in\{x,y,z\}$, respectively.}\label{fig: average inventory different terminal liquidation penalty}
\end{figure}

\begin{figure}[]
	\begin{center}
		\includegraphics[scale=0.33]{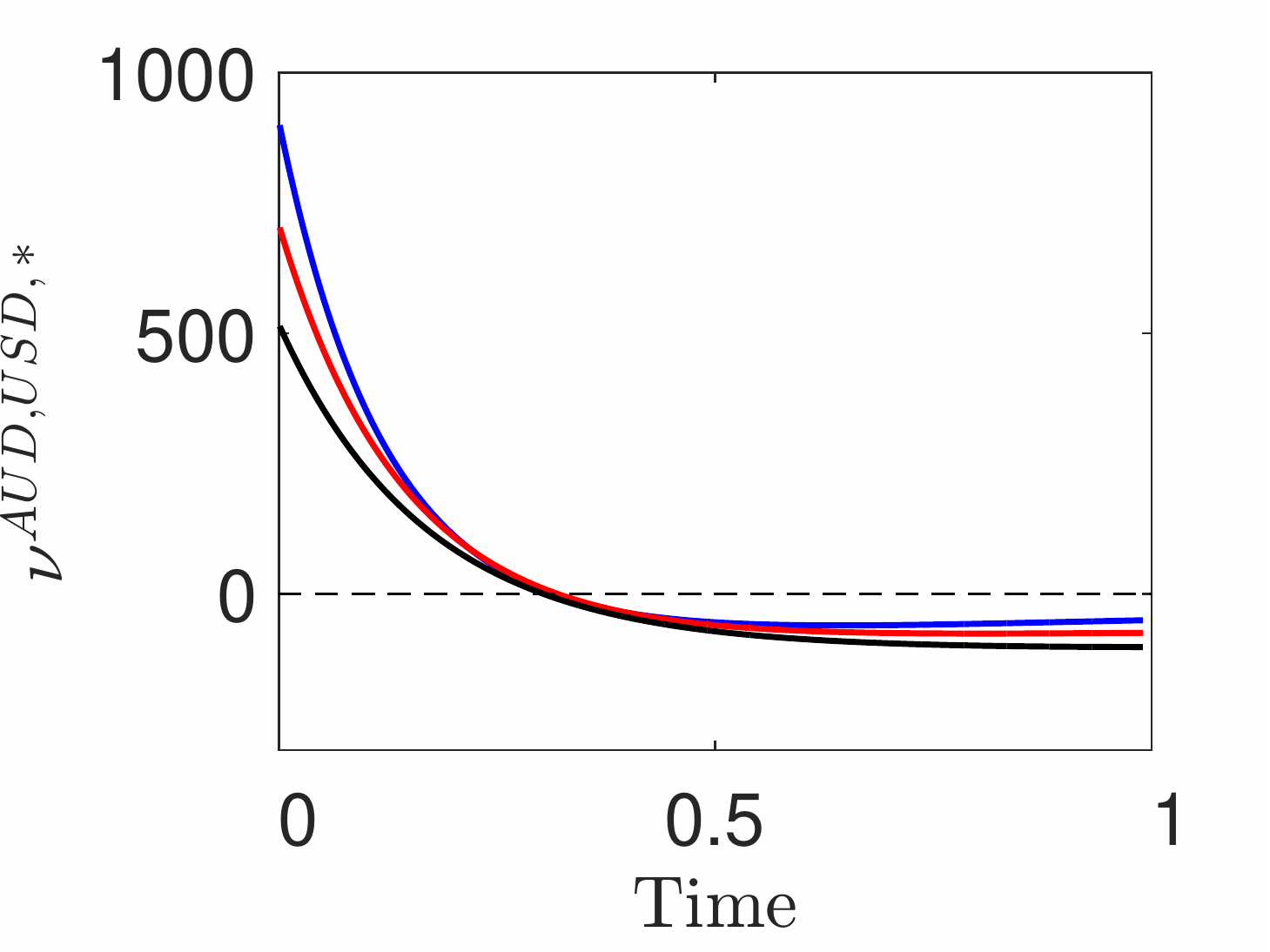}
		\includegraphics[scale=0.33]{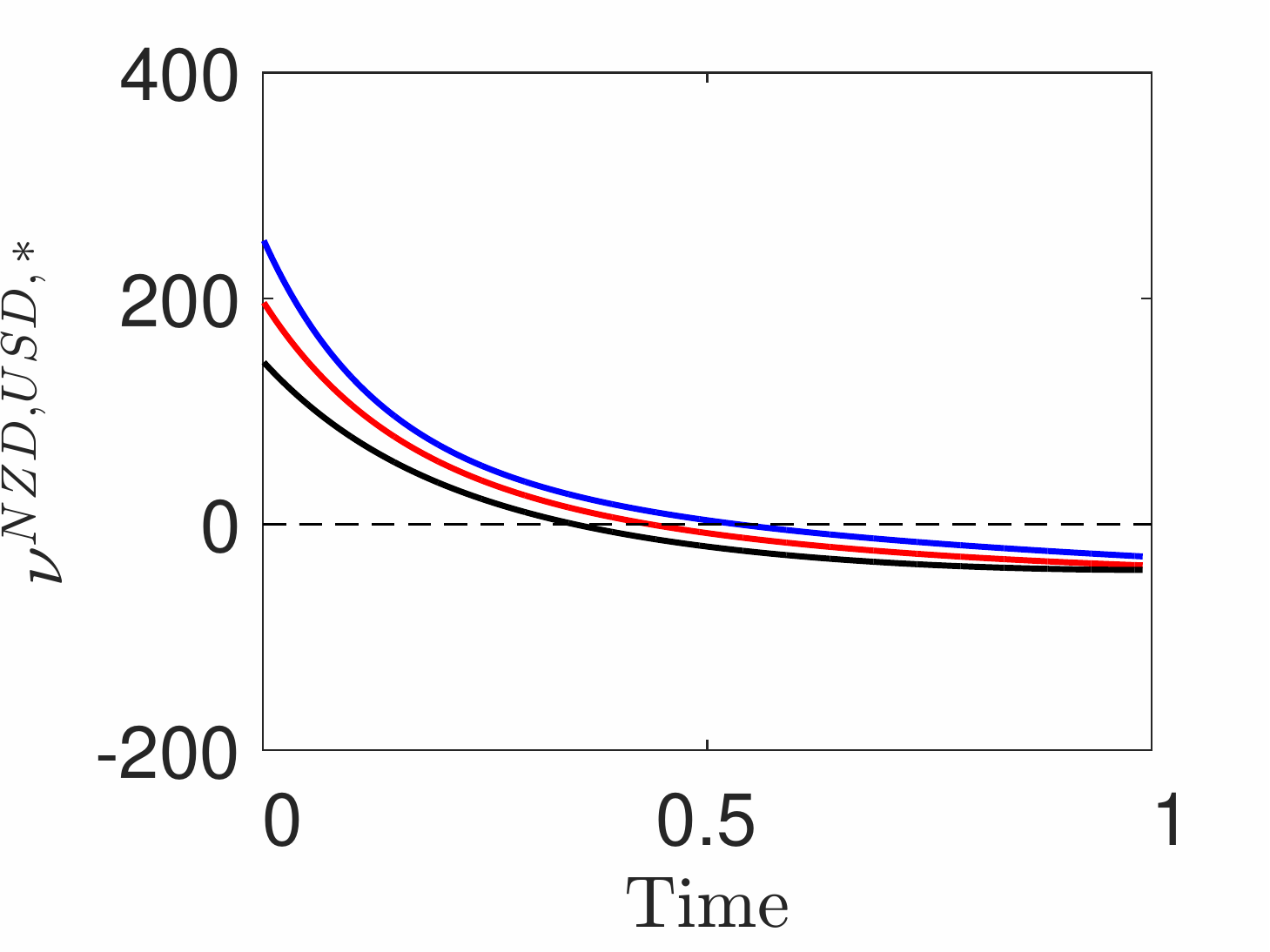}
		\includegraphics[scale=0.33]{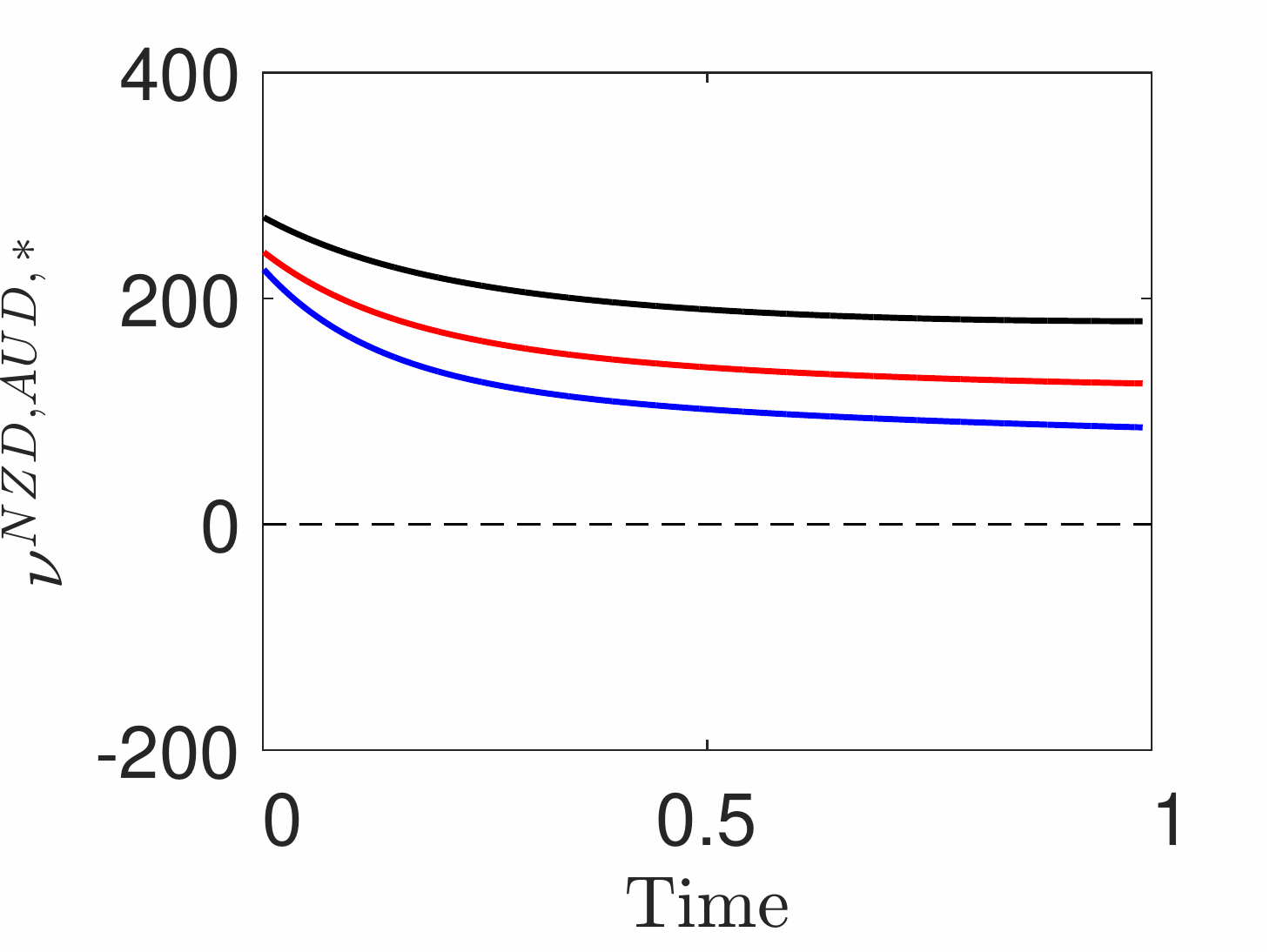}
	\end{center}
	\caption{Average trading speed over all simulated paths for each currency pair, $\varphi = 0.1$, $\lambda^{k,\pm} = 0$, $k\in\{x,y,z\}$. The blue, red, black solid lines represent $\alpha_k = 1\times a_k$, $\alpha_k = 2.5\times a_k$, and $\alpha_k = 10^6\times a_k$, $k\in\{x,y,z\}$, respectively.}\label{fig: average speed of trading different terminal liquidation penalty}
\end{figure}

\begin{figure}[t]
	\begin{center}
		\includegraphics[scale=0.33]{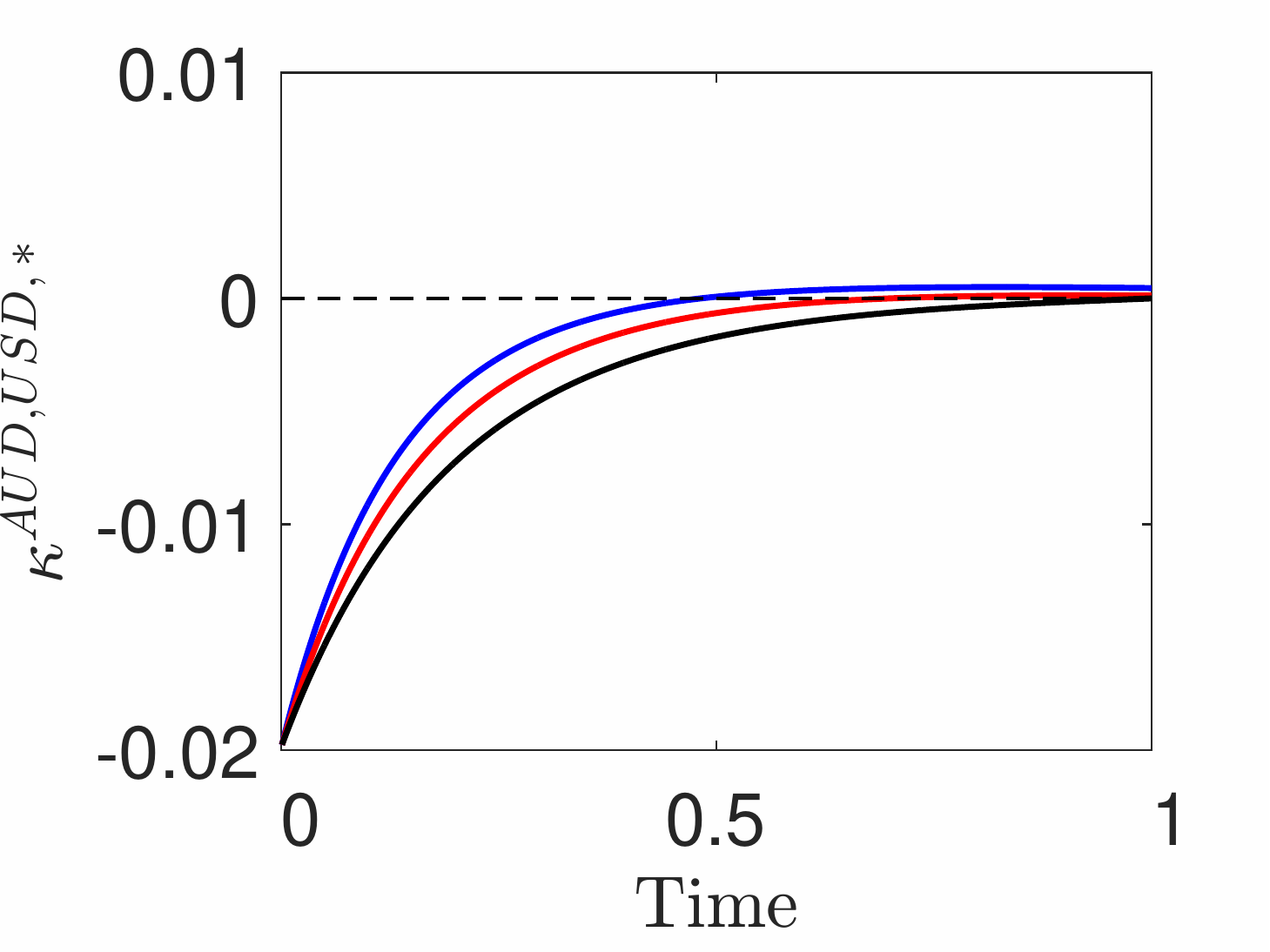}
		\includegraphics[scale=0.33]{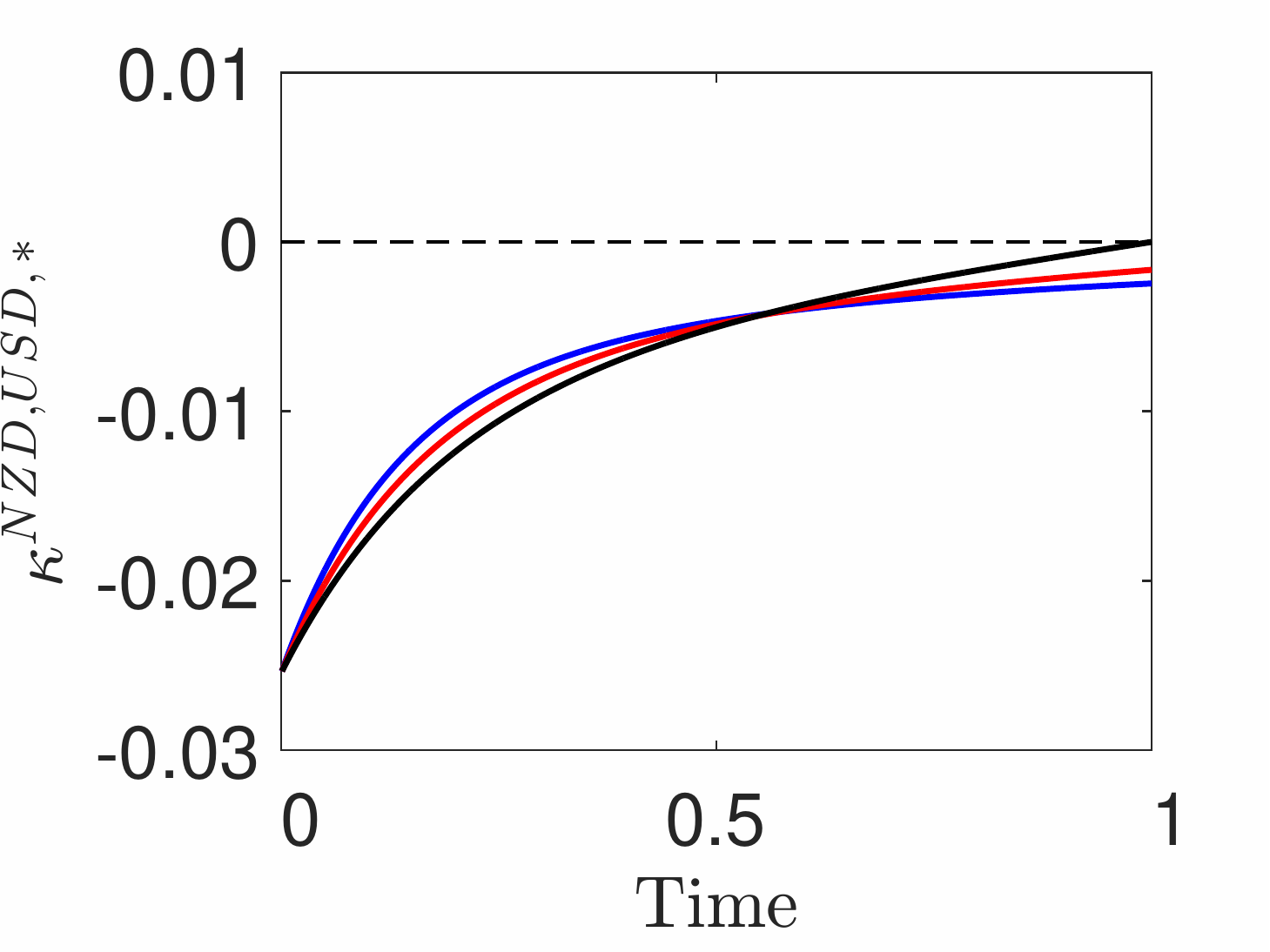}
		\includegraphics[scale=0.33]{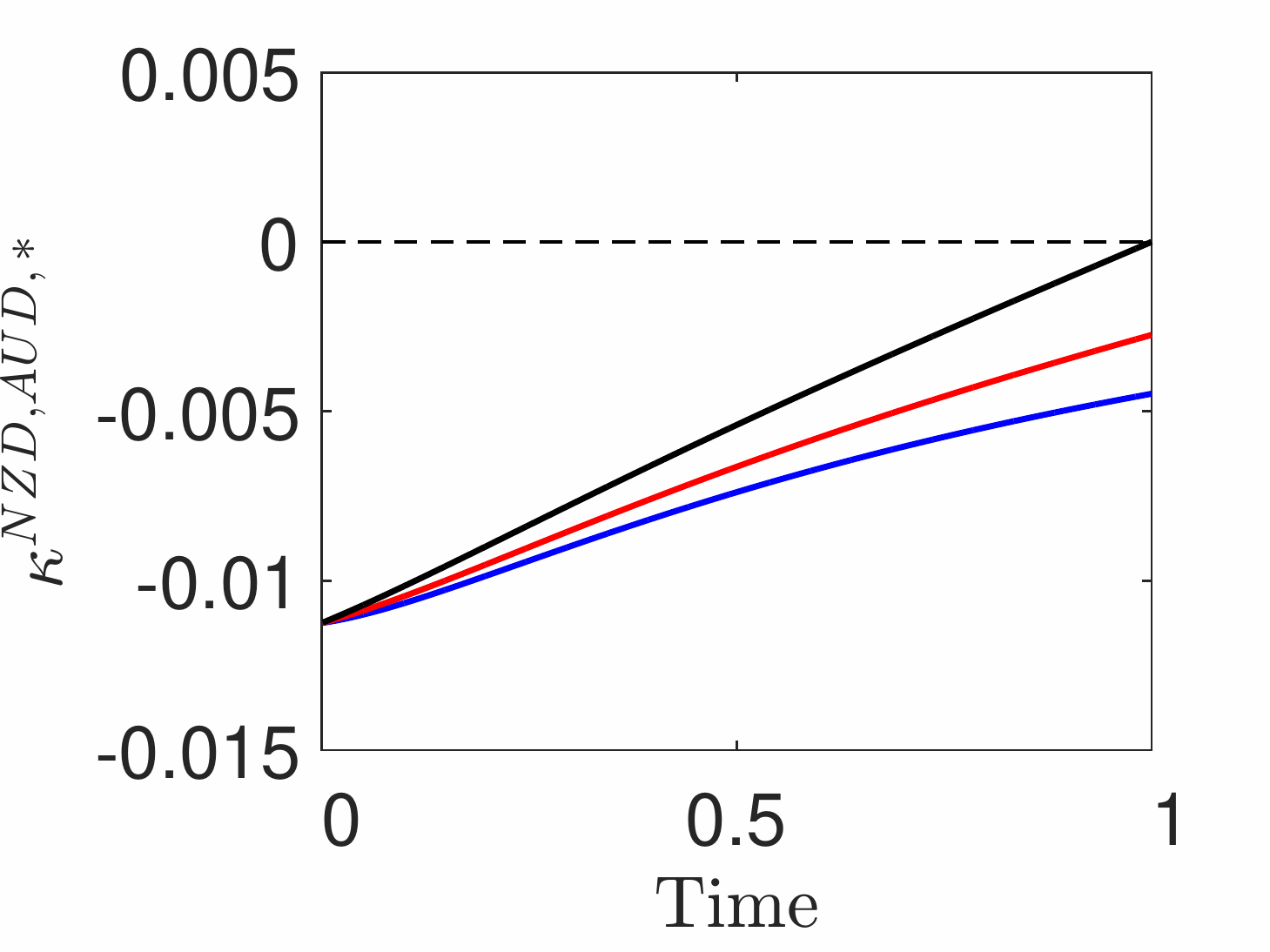}
	\end{center}
	\caption{Average drift adjustments over all simulated paths as result of ambiguity aversion, $\varphi=0.1$, and no trades with broker's clients, $\lambda^{k,\pm} = 0$, $k\in\{x,y,z\}$. The blue, red, black solid lines are for $\alpha_k = 1\times a_k$, $\alpha_k = 2.5\times a_k$, and $\alpha_k = 10^6\times a_k$, $k\in\{x,y,z\}$, respectively.}\label{fig: average drift adjestment different terminal liquidation penalty}
\end{figure}

\begin{figure}[t]
	\begin{center}
		\includegraphics[scale=0.33]{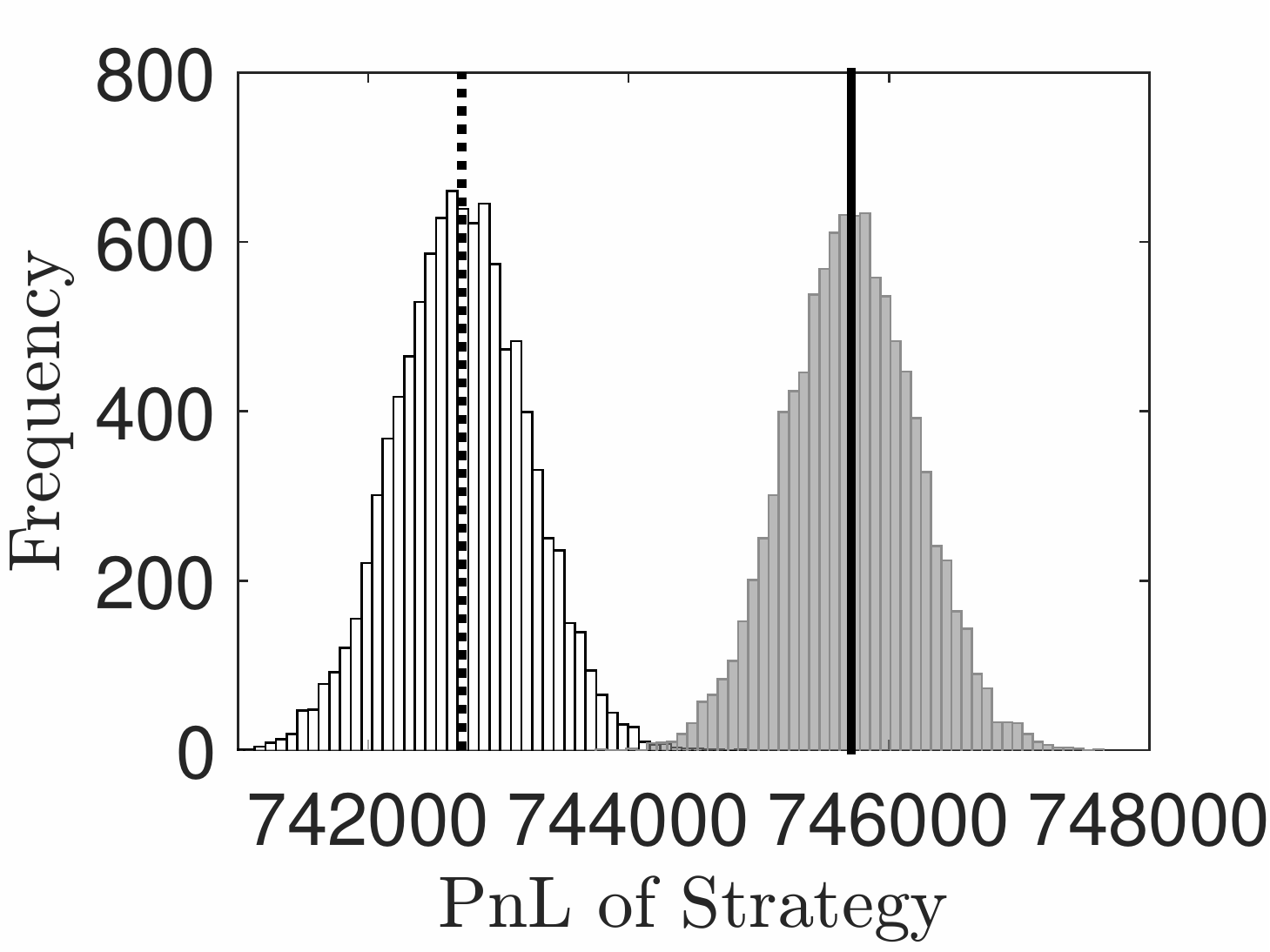}
	\end{center}
	\caption{P\&L of strategy trading in triplet with $\varphi = 0.1$, $\lambda^{k,\pm} = 0$, $k\in\{x,y,z\}$. The plot with black edge white fill is P\&L for $\alpha_k = 1\times a_k$; and the plot with gray edge and fill is P\&L for $\alpha_k = 10^6\times a_k$, where $k\in\{x,y,z\}$.}\label{fig:PnL different terminal liquidation penalty}
\end{figure}

\begin{figure}[]
	\begin{center}
		\includegraphics[scale=0.33]{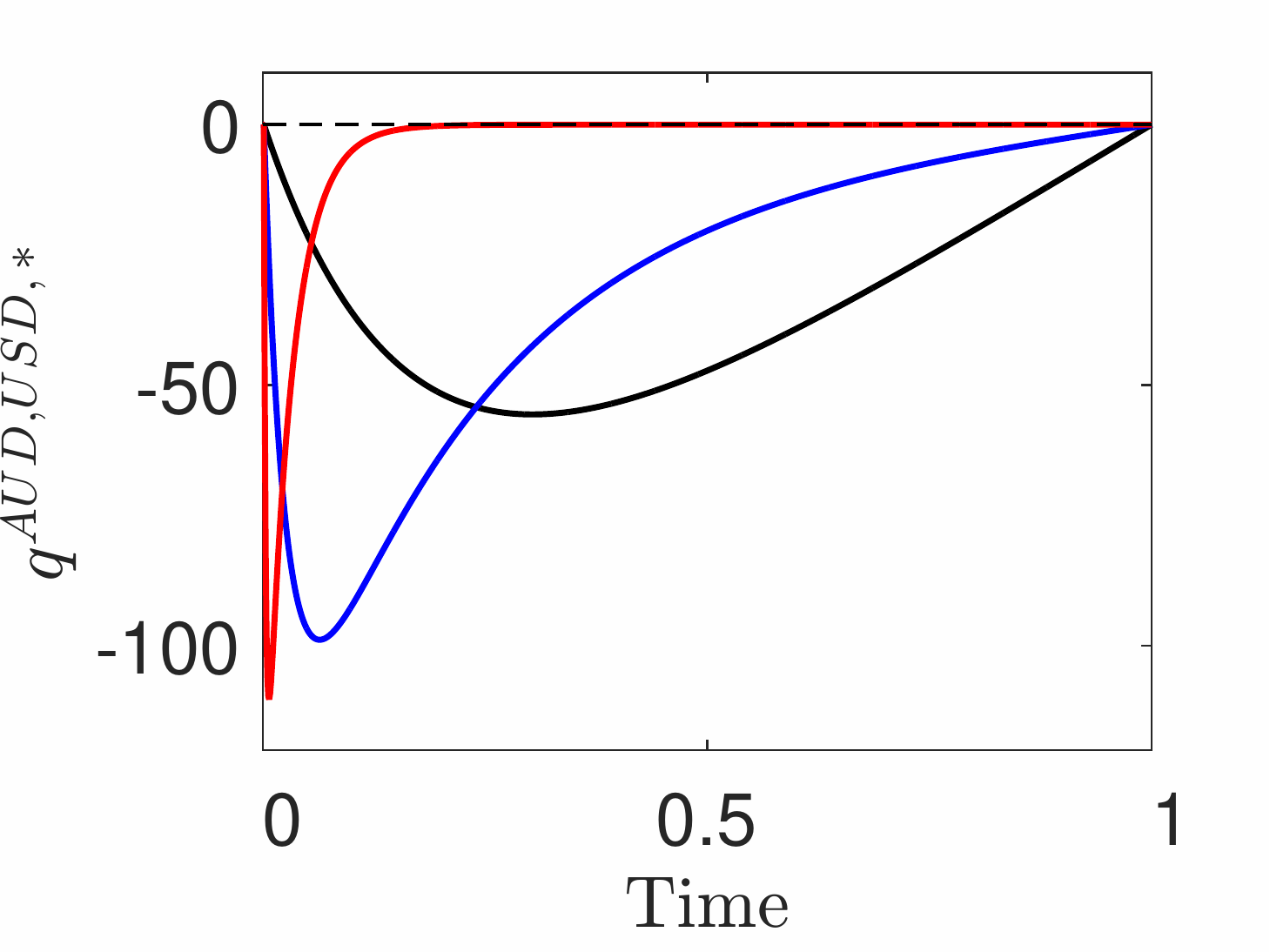}
		\includegraphics[scale=0.33]{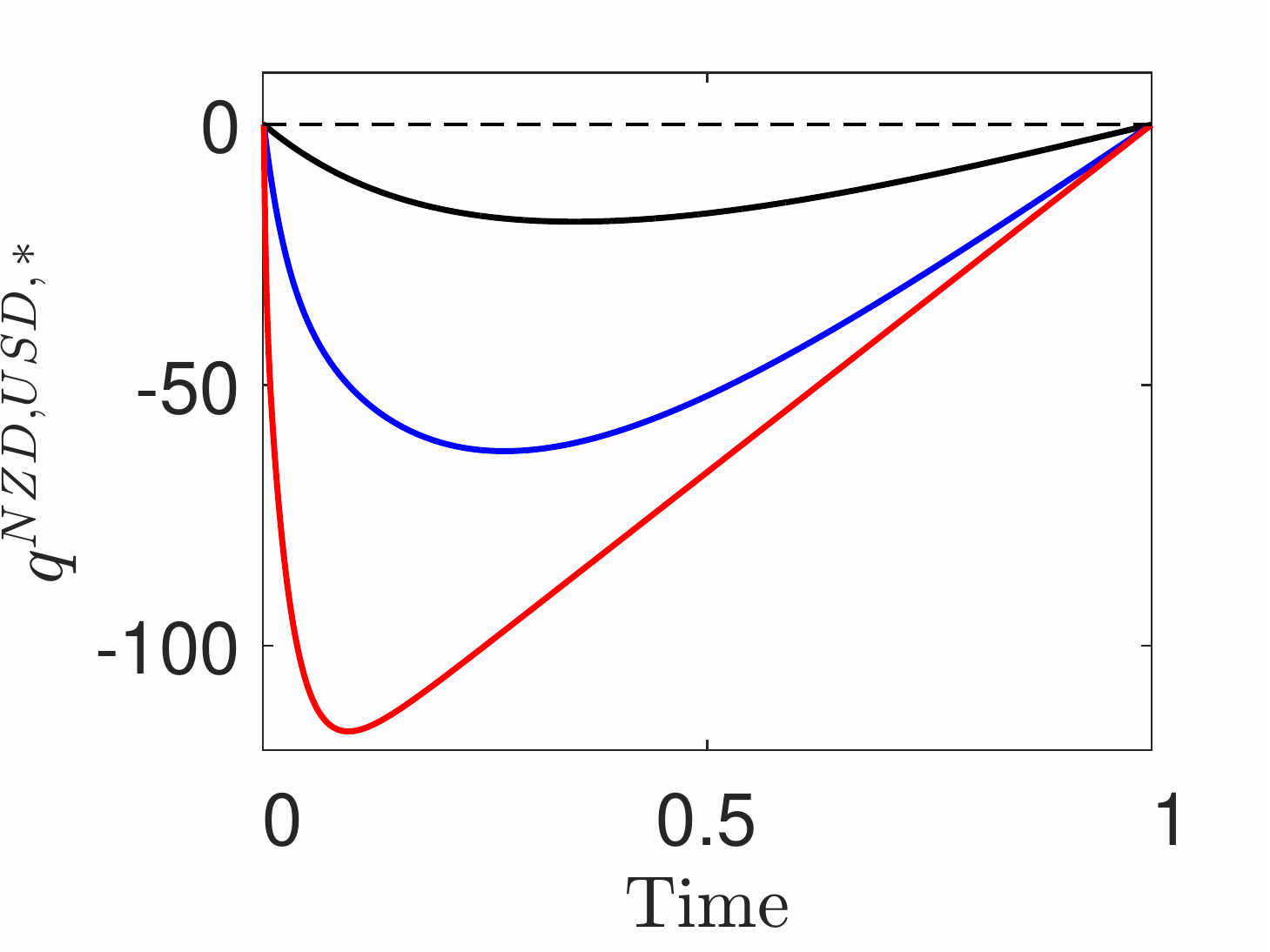}
		\includegraphics[scale=0.33]{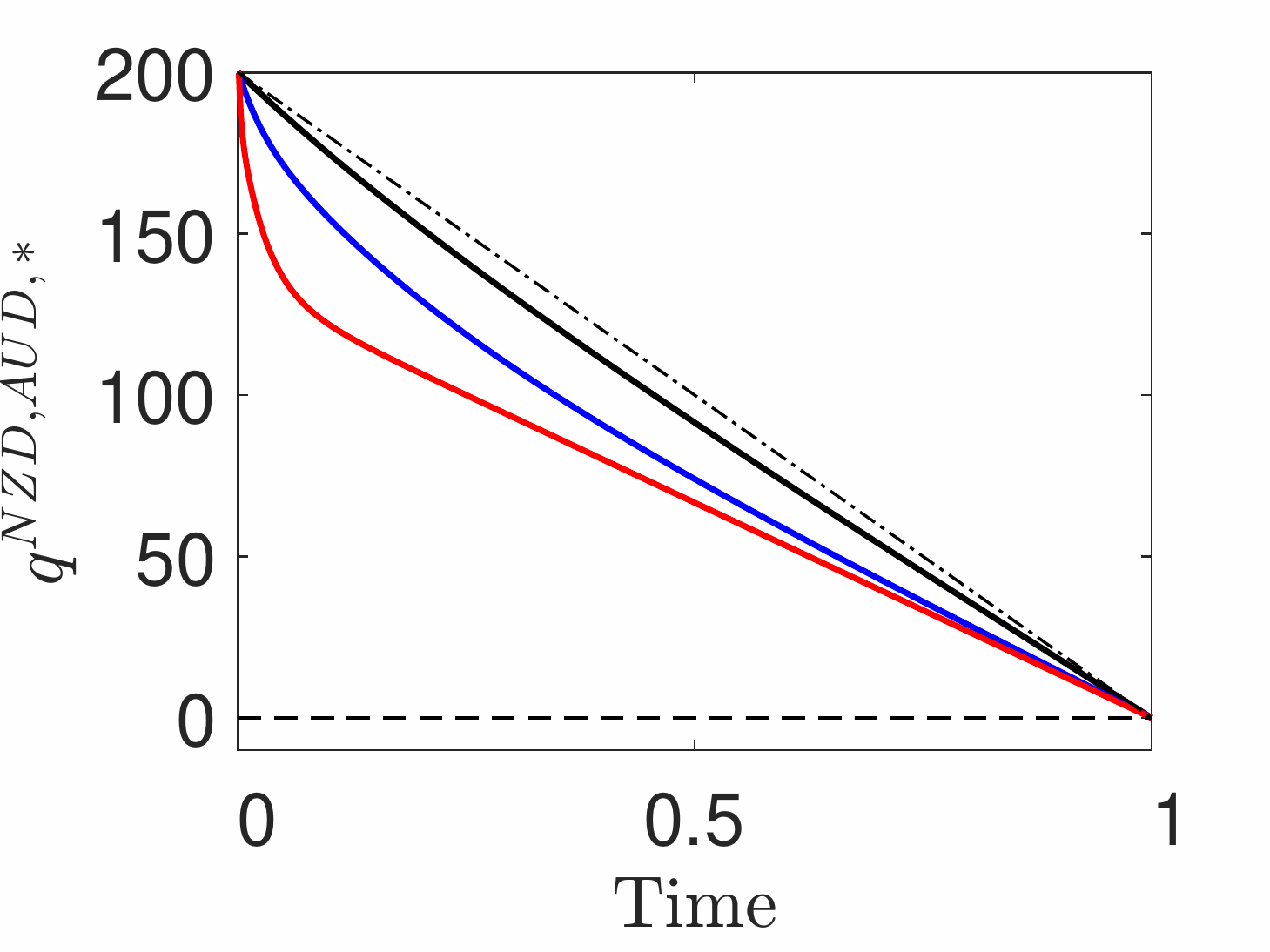}
	\end{center}
	\caption{Average inventory over all simulated paths for three currency pairs, $\alpha_k = 10^6\times a_k$, $\lambda^{k,\pm} = 0$, $k\in\{x,y,z\}$. The red, blue, black solid lines are for $\varphi = 10$, $\varphi = 1$, and $\varphi = 0.1$ respectively.}\label{fig: average inventory different ambiguity aversion}
\end{figure}

\begin{figure}[]
	\begin{center}
		\includegraphics[scale=0.33]{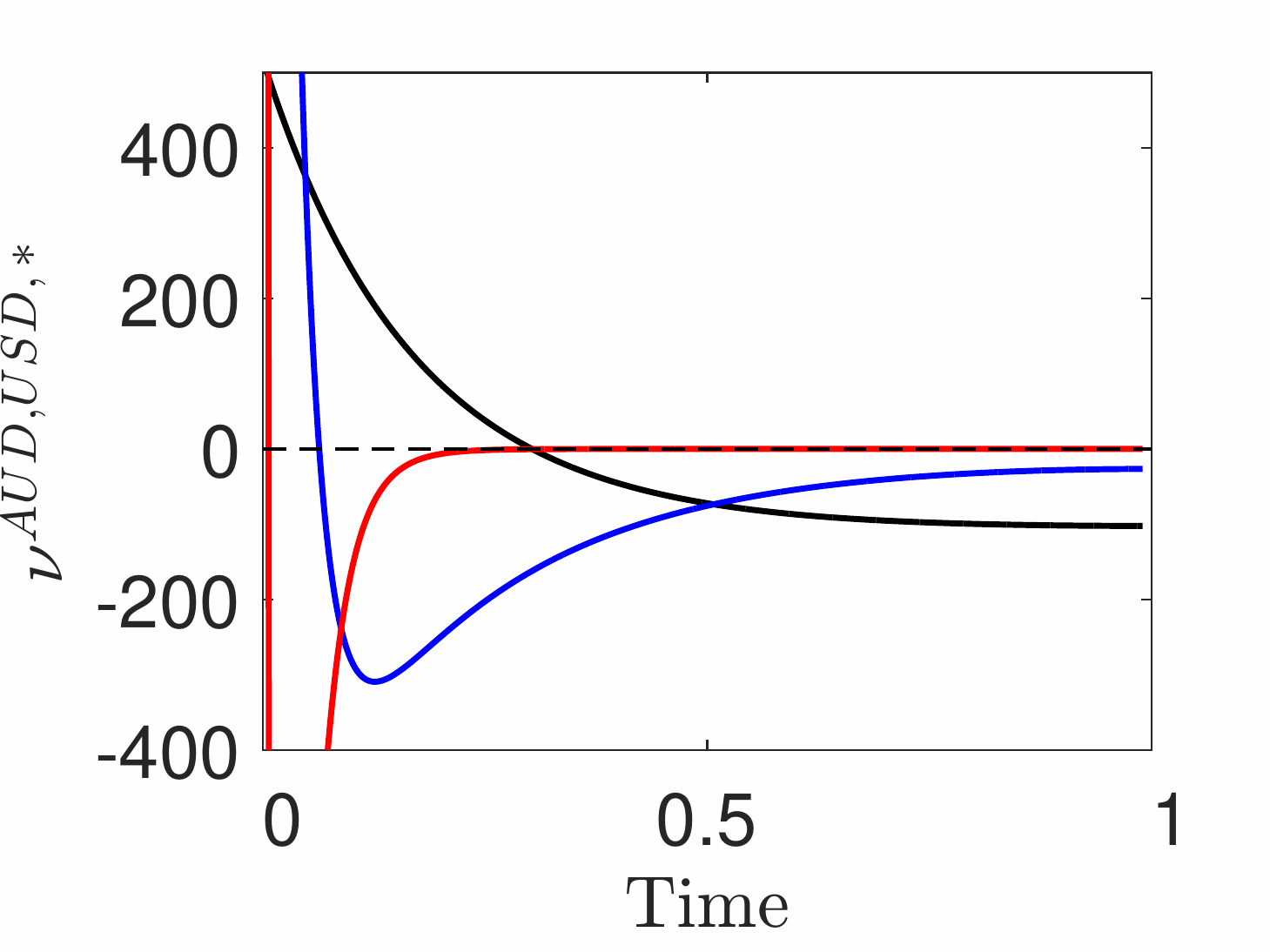}
		\includegraphics[scale=0.33]{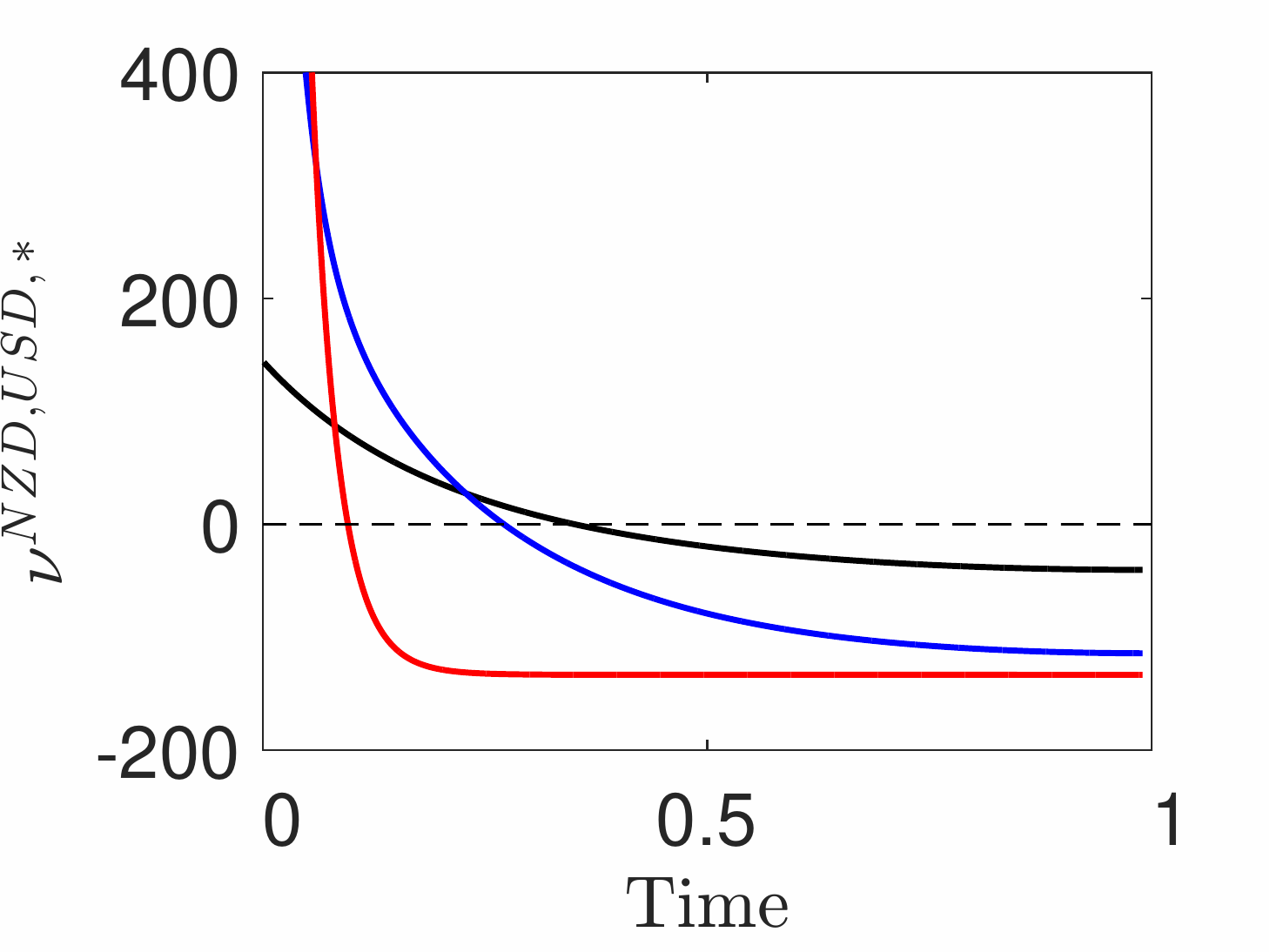}
		\includegraphics[scale=0.33]{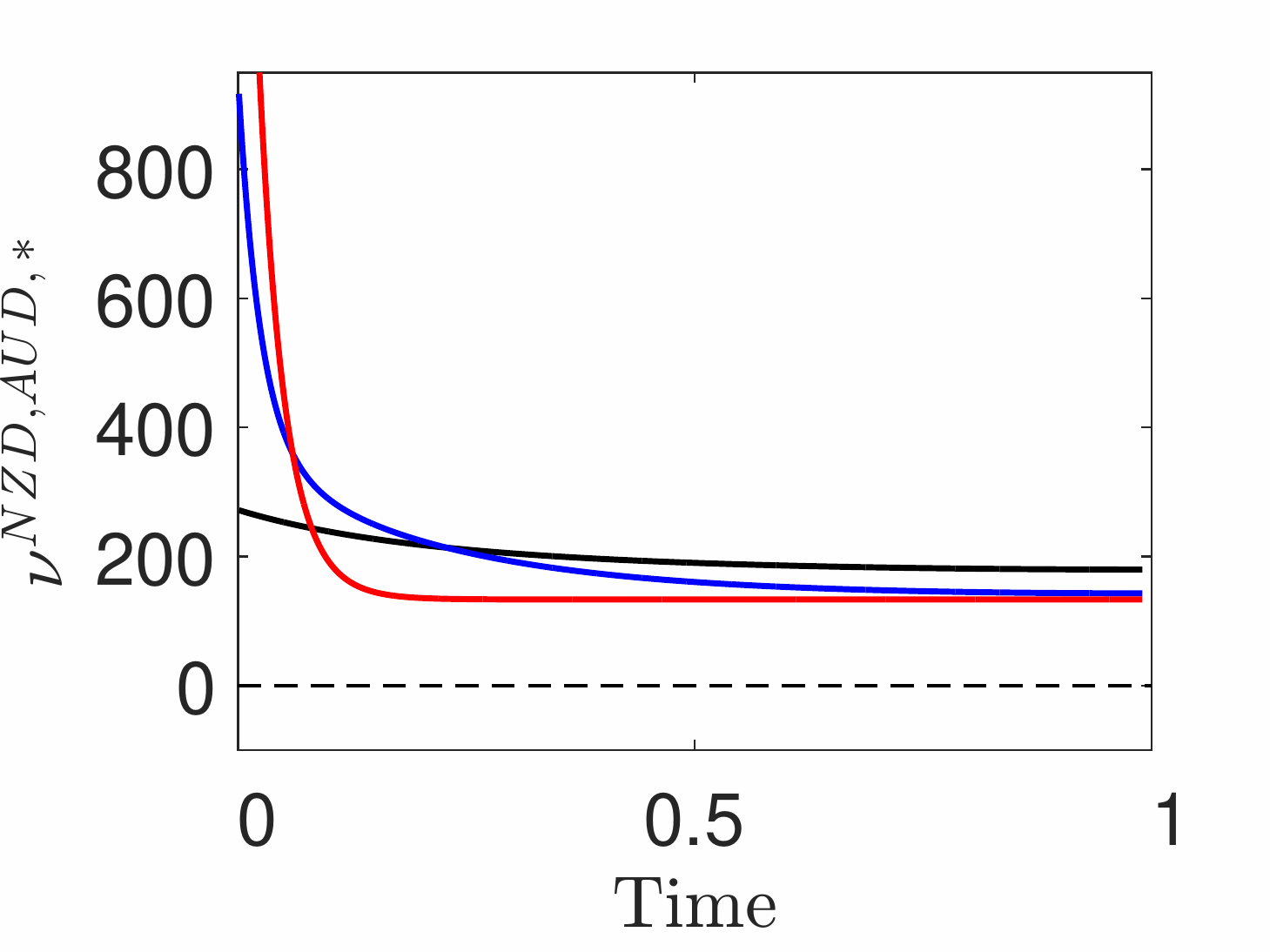}
	\end{center}
	\caption{Average trading speed over all simulated paths for each currency pair, $\alpha_k = 10^6\times a_k$, $\lambda^{k,\pm} = 0$, $k\in\{x,y,z\}$. The red, blue, black solid lines are for $\varphi = 10$, $\varphi = 1$, and $\varphi = 0.1$ respectively. The horizontal black dash lines indicate level of $0$ inventories. The black dotted dash line in the right panel represents liquidation with constant speed.}\label{fig: average speed of trading different ambiguity aversion}
\end{figure}

\begin{figure}[t]
	\begin{center}
		\includegraphics[scale=0.33]{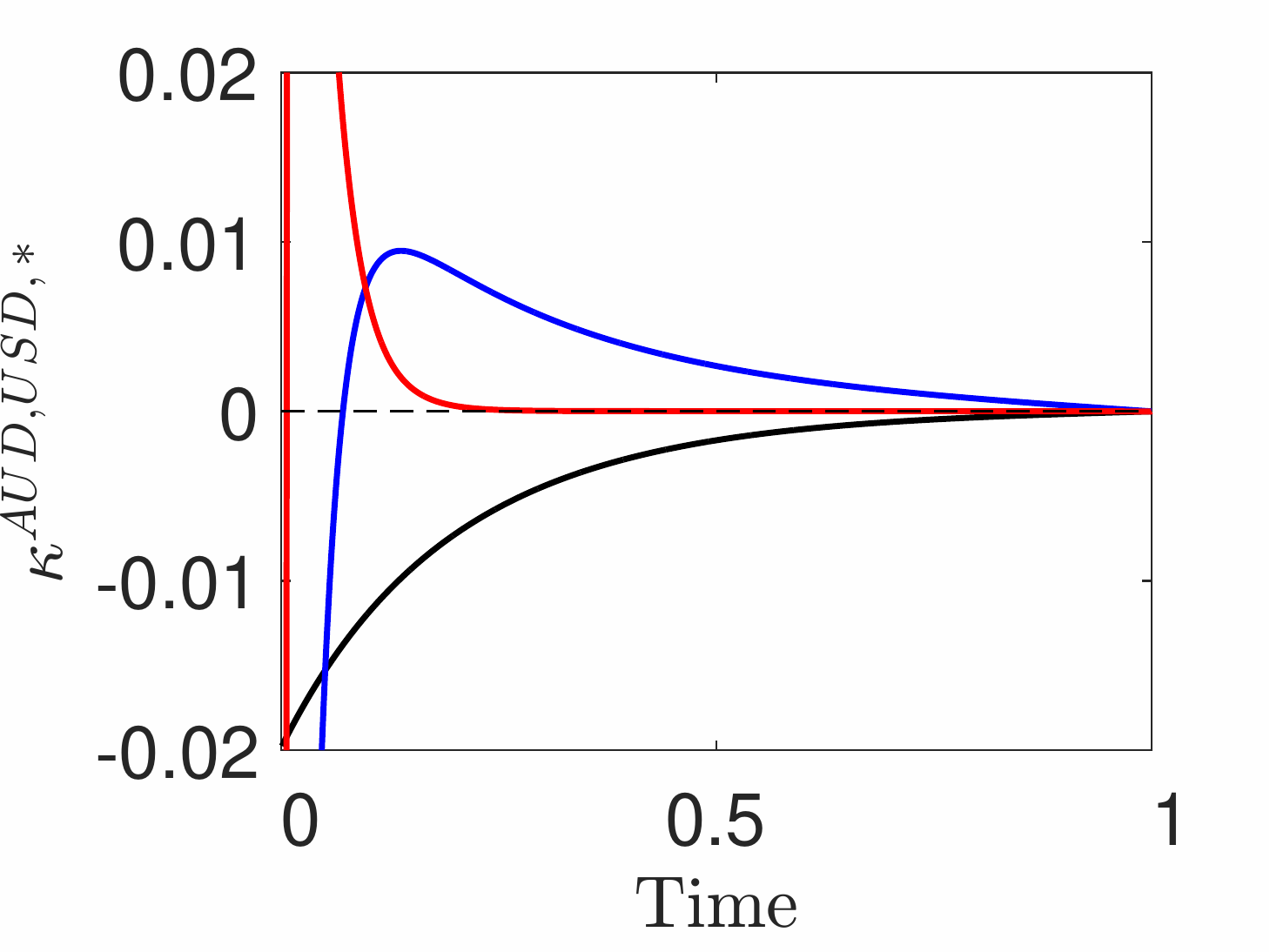}
		\includegraphics[scale=0.33]{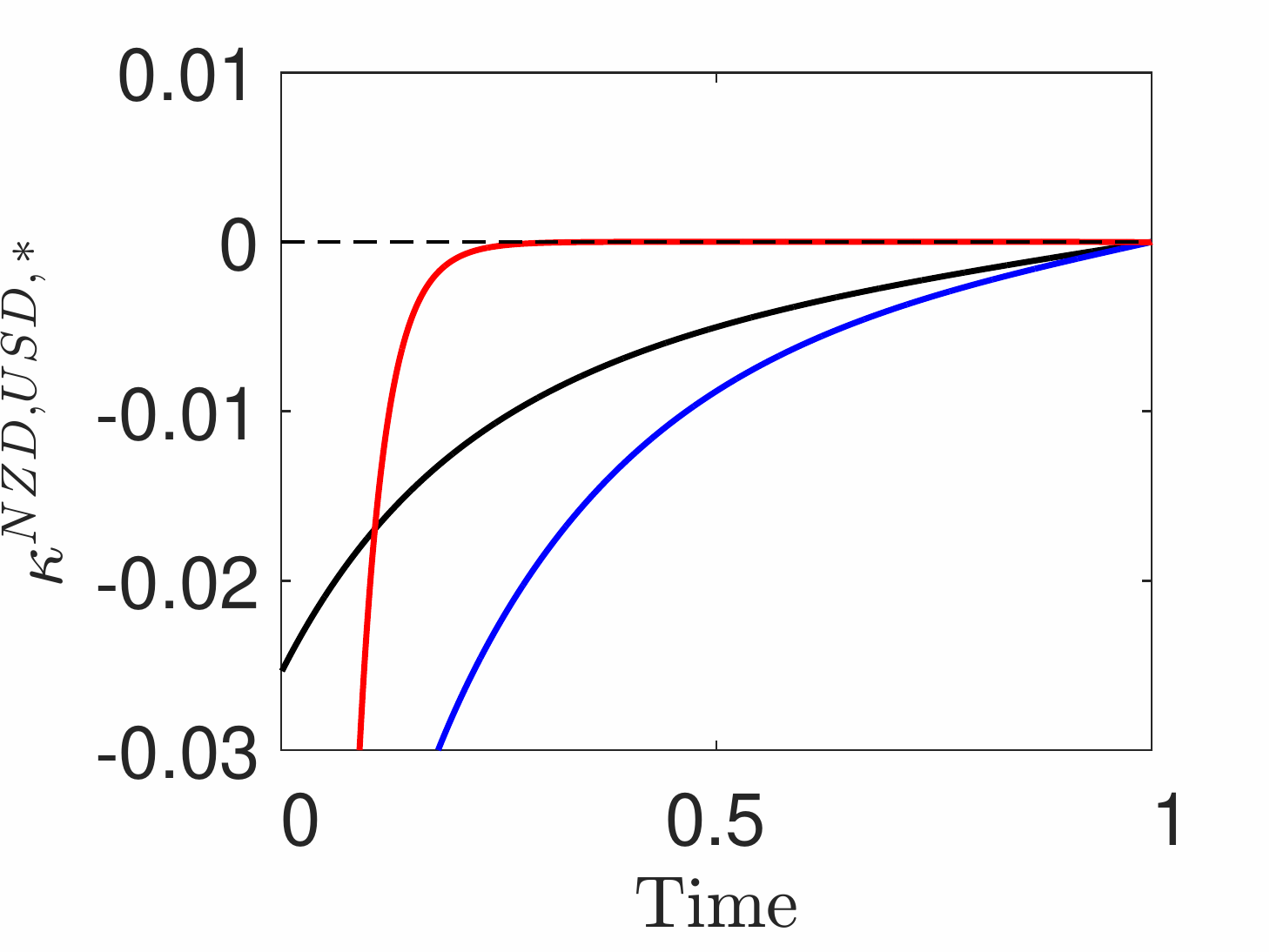}
		\includegraphics[scale=0.33]{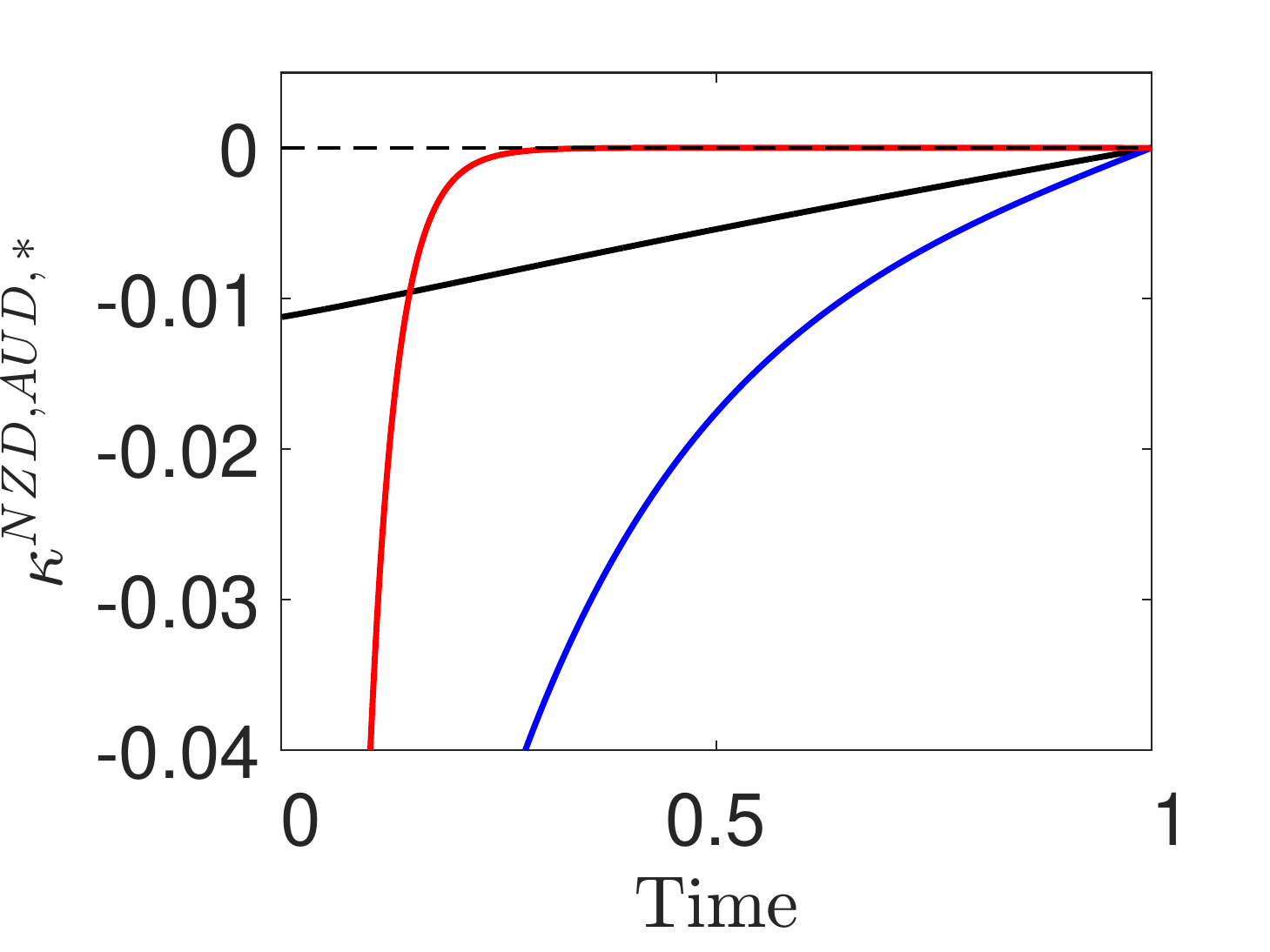}
	\end{center}
	\caption{Average drift adjustments over all simulated paths as result of ambiguity aversion for all three currency pairs, $\alpha_k = 10^6\times a_k$, $\lambda^{k,\pm} = 0$, $k\in\{x,y,z\}$. The red, blue, black solid lines are for $\varphi = 10$, $\varphi = 1$, and $\varphi = 0.1$ respectively.}\label{fig: average drift adjestment different ambiguity aversion}
\end{figure}

\section{Conclusions}\label{sec: conclusions}
We showed how a broker executes a large position in an illiquid foreign exchange currency pair. The broker's strategy considers taking positions on two other more liquid currency pairs. These additional pairs are chosen to form a triplet where, by no arbitrage, one of the pairs becomes redundant.

When the broker is ambiguity neutral, i.e., fully trusts the reference model,   the optimal strategy in each currency pair is independent of the inventory held in the other two pairs. Thus, the ambiguity neutral broker's strategy is to liquidate each pair independently from the other pairs. On the contrary, when the broker makes her model robust to misspecification, we show that  the optimal trading strategy in each currency pair is affected by the inventory holdings in the other two pairs.

We use simulations to compare the tradeoff between mean  and standard deviation of the Profit and Loss  (P\&L) of the optimal liquidation strategy. We show that when the broker makes her model robust to model misspecification  the mean P\&L of the strategy increases and the standard deviation of the P\&L decreases (for a range of ambiguity aversion parameter).

\begin{appendices}

\section{Proof of Proposition \ref{prop10}}
\label{appendix_I}

\begin{proof}
We show that $H$ can be solved analytically and it takes the form \eqref{eqn:ansatzH under P}. First, the supremum in \eqref{eqn: DPE under P} is obtained point-wise by
\begin{equation*}
  \nu^{k,*} = \frac{\whk\,\partial_{\mcX}H-\partial_{q^k}H}{2\,a_k\,\whk\,\partial_{\mcX}H}\,,\qquad \text{for } k \in \{x,y,z\} \text{ and } \whk = k\,1_{\{k \in \{x,y\}\}} + x\,1_{\{k=z\}}\,.
\end{equation*}
Substituting the form \eqref{eqn:ansatzH under P} into the above leads to
\begin{equation}\label{eqn:v under P 1}
\nu^{k,*} = \tfrac{1}{2\,a_k}(h_1^k(t) + 2\,h_2^k(t)\,q^k).
\end{equation}
Substituting the feedback control \eqref{eqn:v under P 1}  and \eqref{eqn:ansatzH under P} into \eqref{eqn: DPE under P} we obtain a new equation which is closed under the form of the solution \eqref{eqn:ansatzH under P}. Collecting the $x\,(q^x)^2$, $y\,(q^y)^2$ and $x\,(q^z)^2$ terms, we find the ODEs
\begin{equation}
-\partial_th_2^k+\tfrac{1}{a_k}(h_2^k)^2-\mu_{\whk}\,h_2^k = 0\,,\label{eqn:h2k-ode}
\end{equation}
with terminal conditions $h_2^k(T) = \alpha_k$. These ODEs have the explicit solutions
\begin{equation}
\label{eqn:mfh2k}
h_2^k(t) = \begin{cases}
\left( 1 - \upsilon_k\,e^{-\mu_{\whk}(T-t)}\right)^{-1}\,
a_k\,\mu_{\whk}\,, & \mu_{\whk}\ne 0\,,\\
\left( \alpha_k^{-1} + a_k^{-1}(T-t)\right)^{-1}, & \mu_{\whk}=0\,,
\end{cases}
\end{equation}
with constant $\upsilon_k:=(1-\tfrac{a_k}{\alpha_k}\,\mu_{\whk})$.

Next, collecting $x\,q^x$, $y\,q^y$ and $x\,q^z$ terms, and setting $\theta^{k,\pm} := \int_{\RRR^+}r\,F^{k,\pm}(dr)$, we have
\begin{equation}
	-\partial_th_1^k+ \mu_{\whk}(1-h_1^k)
+\tfrac{1}{a_k}\,h_2^k\,h_1^k-2\,\gamma^k_-\,h_2^k = 0\,,\label{eqn:h1k-ode}
\end{equation}
with terminal conditions $h_1^k(T)=0$, and $\gamma^k_-:=(\lambda^{k,+}\theta^{k,+}-\lambda^{k,-}\theta^{k,-})$. These equations admit the explicit solutions
\begin{equation}
\label{eqn:mfh1k}
h_1^k(t) =
\begin{cases}
\left( 1 - \upsilon_k \,e^{-\mu_{\whk}\,(T-t)}\right)^{-1}\,\left\{(1-2\,a_k\,\gamma^k_{-})\,
\mu_{\whk}\,\left(t-T\right)+ \upsilon_k\,(1-e^{-\mu_{\whk}\,(T-t)})\right\}\,,
& \mu_{\whk}\ne 0, \\
2\,\gamma_-^k\,\left( \alpha_k^{-1} + a_k^{-1}\,(T-t)\right)^{-1}\,(T-t)\,,
& \mu_{\whk} = 0\,.
\end{cases}
\end{equation}

Finally, collecting $x$ and $y$ terms, and let $\eta^{k,\pm} := \int_{\RRR^+}r^2\,F^{k,\pm}(dr)$, we have
\begin{subequations}
	\begin{align}
\begin{split}
	0 = -(\partial_t+\mu_x)h_0^x&+\psi^x+\psi^z
-\gamma_-^x\,h_1^x(t)-\gamma_-^z\,h_1^z(t)
-\delta^x\,h_2^x(t)-\delta^z\,h_2^z(t)\\
&+\tfrac{1}{4a_x}\,(h_1^x)^2+\tfrac{1}{4a_z}\,(h_1^z)^2\,,
\end{split}
\label{eqn:eqh1}
\\
\begin{split}
	0 = -(\partial_t+\mu_y)h_0^y&+\psi^y-\gamma^y_-\,h_1^y(t)-\delta^y\,h_2^y(t)+\tfrac{1}{4a_y}\,(h_1^y)^2\,,
\end{split}\label{eqn:eqh2}
\end{align}
\end{subequations}
with terminal conditions $h_0^x(T)=h_0^y(T)=0$, and constants $\delta^k:=(\lambda^{k,+}\eta^{k,+}+\lambda^{k,-}\eta^{k,-})$
and
$\psi^k:=(c_k^{-}\lambda^{k,-}\,\eta^{k,-}+c_k^{+}\lambda^{k,+}\,\eta^{k,+})$. These ODEs can be solved by quadratures and are given by
\begin{subequations}
	\begin{align}
	h_0^x(t) =&\, \sum_{\ell\in\{x,z\}}\left\{\mfh_1^\ell(t)+\mfh_2^\ell(t) + \mfh_3^\ell(t) +\mfh_4^\ell(t)\right\}\,,
\label{eqn:h1}
\\
	h_0^y(t) =&\,\mfh_1^y(t)+\mfh_2^y(t)+\mfh_3^y(t)+\mfh_4^y(t)\,,
\label{eqn:h2}
\end{align}
\label{eqn:mfh0k}
\end{subequations}
where, for all $\mu_{\whk}$
\begin{equation}
\begin{split}
\mfh_1^k(t) &= -\int_t^T\,e^{\mu_{\whk}\,(u-t)}\,\psi^k\,du\,,
\quad
\mfh_2^k(t) = \int_t^T\,e^{\mu_{\whk}\,(u-t)}\,\gamma_-^k\,h_1^k(u)\,du\,,
\\
\mfh_3^k(t) &= \int_t^T\,e^{\mu_{\whk}\,(u-t)}\,\delta^k\,h_2^k(u)\,du\,,
\quad
\mfh_4^k(t) = -\int_t^T\,e^{\mu_{\whk}\,(u-t)}\,\tfrac{1}{4\,a_k}\,\left(h_1^k(u)\right)^2\,du\,.
\end{split}
\label{eqn:mfh0k component mu neq 0}
\end{equation}

\qed
\end{proof}

\section{Proof of Theorem \ref{thm:VerificationP}} \label{sec:Proof-Verification-P}

Since \eqref{eqn:ansatzH under P} is classical, we only need to check that the controls \eqref{eqn:opt-controls-P} are admissible. To check that $\bnu^*$ is admissible,  it suffices to show that $\EE^{\PP}[\int_0^t\,(Q_u^{k,\bnu^*})^2\,du]<+\infty$ because $\nu^{k,*}$ is affine in $Q^{k,\bnu^*}$. It further suffices to show $\EE^{\PP}[(Q_t^{k,\bnu^*})^2]<+\infty$ on $t \in [0,T]$.

From \eqref{eqn:Q dynamics}, we have
\begin{equation}\label{eqn:optimal Q_t SDE}
Q^{k,\bnu^*}_t = Q^{k,\bnu^*}_0-\int_0^t \tfrac{1}{2\,a_k} \left( h_1^k(u) + 2\,h_2^k(u)\,Q^{\bnu^*,k}_u\right)\,du-O^{k,-}_{t}+O^{k,+}_{t}\,.
\end{equation}
Use the integrating factor $\pi_t^k:=e^{\frac{1}{a_k}\int_0^t h_2^k(s)ds}$ to find
\begin{equation}\label{eqn:optimal Q_t Solution}
Q^{k,\bnu^*}_t = Q^{k,\bnu^*}_0 \,\frac{\pi^k_0}{\pi^k_t} - \tfrac{1}{2a_k}\int_0^t \frac{\pi^k_u}{\pi^k_t}\,h_1^k(u)\,du + \int_0^t \frac{\pi^k_u}{\pi^k_t}\, (dO^{k,+}_{u}-dO^{k,-}_{u})\,.
\end{equation}

From \eqref{eqn:mfh2k} we see that $h_2^k(t)$ does not change sign on the interval $t\in[0,T]$. Next, if $0<\mu_{\whk}<\alpha_k/a_k$, then $\exists\; C_0>C_1>0$, s.t. $h_2^k(t) = \left(C_0-C_1\,e^{-\mu_{\whk}\,(T-t)}\right)^{-1}$. Hence, by explicit integration, $\pi^k_t = e^{C_2 t}\,\left(C_3 - C_4\,e^{-\mu_{\whk}\,(T-t)}\right)^{C_5}\,,$
where $C_2 = 1/a_k\,C_0 > 0$, $C_3 =  C_0 / (C_0 - C_1e^{-\mu_{\whk}\,T}) > C_4 =  C_1/ (C_0 - C_1e^{-\mu_{\whk}\,T}) > 0$, $C_5 = -1/a_k\,\mu_{\whk}\,C_0 < 0$.

Next, by Ito's isometry for jump processes, we have
\begin{align}\label{eqn:Expectation of pi_u}
\EE^{\PP}\left[ \left( \int_0^t \frac{\pi^k_u}{\pi^k_t}\,dO_u^{k,\pm} \right)^2\right]
\le  \int_0^t \left(\tfrac{\pi^k_u}{\pi^k_t}\right)^2 \lambda^{k,\pm}\theta^{k,\pm}\,du<\infty\,, \qquad \forall 0\le t\le T\,,
\end{align}
where the last inequality follows because $C_3>C_4$ and therefore $\pi_t^k$ remains bounded on the domain $t\in[0,T]$. Similar arguments hold when $- \alpha_k /a_k<\mu_{\whk}<0$.

When $\mu_{\whk}=0$, from \eqref{eqn:mfh2k}, we have $h_2^k(t) = \left( \alpha_k^{-1} + a_k^{-1}(T-t)\right)^{-1}$. Hence, by explicit integration, $\pi^k_t = \left( \alpha_k^{-1} + a_k^{-1}(T-t)\right)^{-1}\,\left( \alpha_k^{-1} + a_k^{-1}\,T\right)$. We have the analogous inequality \eqref{eqn:Expectation of pi_u}. Hence the result follows.
\qed

\section{Proof of Proposition \ref{prop11}}
\label{appendix_L}

We proof the case $\mu_{\whk}\neq 0$ ($\mu_{\whk} = 0$ is similar). From \eqref{eqn:mfh2k} and \eqref{eqn:mfh1k} we have
\begin{equation*}
\label{eqn:mfh2k alpha infty}
\lim_{\alpha_k\to +\infty}h_2^k(t) =
\left( 1 - e^{-\mu_{\whk}(T-t)}\right)^{-1}\,a_k\,\mu_{\whk}\,,
\end{equation*}
and
\begin{equation*}
\label{eqn:mfh1k alpha infty}
\lim_{\alpha_k\to +\infty}h_1^k(t) = \left( 1 - e^{-\mu_{\whk}\,(T-t)}\right)^{-1}\,(1-2\,a_k\,\gamma^k_{-})
\,\mu_{\whk}\,\left(t-T\right) + 1\,.
\end{equation*}
Then
\begin{equation*}
\lim_{t\to T}\left\{\left( 1 - e^{-\mu_{\whk}(T-t)}\right)^{-1}\,a_k\,\mu_{\whk}\right\}\,(T-t) = a_k\,,
\end{equation*}
and
\begin{equation*}
\lim_{t\to T}\left( 1 - e^{-\mu_{\whk}\,(T-t)}\right)^{-1}\,(1-2\,a_k\,\gamma^k_{-})
\mu_{\whk}\left(t-T\right) + 1 = 2\,a_k\,\gamma^k_{-}\,.
\end{equation*}

From \eqref{eqn:v under P 1}, as $\alpha_k\to +\infty$ and $T-t \ll 1$, we obtain
\begin{equation}
	\lim_{\alpha_k\to+\infty} \nu_t^{k,*} = (T-t)^{-1}\,Q^{k,\bnu^*}_t + \gamma_-^k +o(T-t)\,,\label{eqn:nuk t to T lambda neq 0}
\end{equation}
recall $\gamma^k_- =(\lambda^{k,+}\theta^{k,+}-\lambda^{k,-}\theta^{k,-})$  and $\theta^{k,\pm} := \int_{\RRR^+}r\,F^{k,\pm}(dr)$.

The SDE of the inventory process is
\begin{equation}
dQ_u^{k,\bnu^*} = -\left[(T-u)^{-1}\,Q^{k,\bnu^*}_u + \gamma_-^k\right]\,du + d(O^{k,+}_u-O^{k,-}_u)+o(T-t)\,,\qquad u\in[t,T]\,,
\end{equation}
which has solution
\begin{equation}\label{eqn: Q T-t << 1}
Q_s^{k,\bnu^*}=\frac{T-s}{T-t}\,Q_t^{k,\bnu^*}+\int_t^s \frac{T-s}{T-u}\,\left[d(O^{k,+}_u-O^{k,-}_u)-\gamma_-^k\,du\right]+o(T-t)\,, \qquad s\in[t,T]\,.
\end{equation}
As $s\to T$ the controlled inventory $Q_s^{k,\bnu^*}$ vanishes. In other words, the strategy guarantees that the broker completely liquidates her position by maturity.

If we substitute $Q_s^{k,\bnu^*}$ given in \eqref{eqn: Q T-t << 1} into \eqref{eqn:nuk t to T lambda neq 0}, the speed of trading becomes
\begin{equation}\label{eqn: nu of t to T alpha to infinity}
\nu_s^{k,*} = \frac{Q_t^{k,\bnu^*}}{T-t}+\int_t^s \frac{1}{T-u}\,\left[d(O^{k,+}_u-O^{k,-}_u)-\gamma_-^k\,du\right]
+\frac{\gamma_-^k}{T-s}
+o(T-t)\,, \qquad s\in[t,T]\,,
\end{equation}
and if we assume that the expected order flow from  clients $\gamma_-^k\ne0$, it is easy to see that the broker's speed of trading is infinitely fast close to maturity. Let $\gamma_-^k = 0$. Then,   the speed of trading becomes
\begin{equation*}
\begin{split}
\nu_s^{k,*} = &\,\frac{Q_t^{k,\bnu^*}}{T-t}+\int_t^s \frac{1}{T-u}\,(dO^{k,+}_u-\lambda^{k,+}\theta^{k,+}du)-\int_t^s \frac{1}{T-u}\,(O^{k,-}_u-\lambda^{k,-}\theta^{k,-}du)
+o(T-t)\,. \qquad
\end{split}
\end{equation*}
Note that
\begin{equation*}
\begin{split}
&\EE\left[\left(\int_t^s\frac{1}{T-u}\,(dO^{k,\pm}_u
-\lambda^{k,\pm}\theta^{k,\pm}du)\right)^2\right] = \lambda^{k,\pm}\,\eta^{k,\pm}\left(\frac{1}{T-s}-\frac{1}{T-t}\right)\,,
\end{split}
\end{equation*}
and recall $\eta^{k,\pm} := \int_{\RRR^+}r^2\,F^{k,\pm}(dr)>0$.

If $\lambda^{k,\pm} = 0$, we therefore have $\EE\left[\,\int_t^T (\nu^{k,*}_s)^2\,ds\,\right]< \infty$ , and hence, in this case, the strategy is admissible, and from \eqref{eqn: Q T-t << 1} we see that the broker completely liquidates her position by $T$.

If $\lambda^{k,\pm} > 0$, we have $\EE\left[\,\int_t^T (\nu^{k,*}_s)^2\,ds\,\right] = \infty$, which is not an admissible strategy.
\qed

\section{Proof of Proposition \ref{prop6}}
\label{appendix_G}

We prove the case $\mu_{\whk} = 0$, the proof is similar when $\mu_{\whk} \neq 0$. The control $\tilde \bnu^*$, given in \eqref{eqn:vkaprox} and repeated here for convenience is
\begin{equation*}
\begin{split}
\tilde{\nu}^{k,*} =&\, \frac{\whk-\partial_{q^k}H_0}{2\,a_k\,\whk} - \varphi\,\frac{\partial_{q^k}H_1}{2\,a_k\,\whk}\,,\quad\,k\in\{x,y,z\}\textrm{ and }\,\whk = k\,1_{\{k\in\{x,y\}\}} + x\,1_{\{k=z\}}\,.
\end{split}
\end{equation*}
When $\lambda^{k,\pm}=0$, the function $H_0$ in \eqref{eqn:ansatzH0} simplifies to
\begin{equation*}
\begin{split}
	H_0(t,x,y,\bbq) = &\,x\,q^x+y\,q^y+x\,q^z-h_{02}^x(t)\,x\,(q^x)^2-h_{02}^y(t)\,y\,(q^y)^2-h_{02}^z(t)\,x\,(q^z)^2\,.\label{eqn:ansatzH0simplified}
\end{split}
\end{equation*}
In the limit $\alpha_k\to\infty$ the first term in $\tilde{\nu}^{k,*}$, see \eqref{eqn:mfh2k}, becomes
\begin{equation*}
	\lim_{\alpha_k\to +\infty}\frac{\whk-\partial_{q^k}H_0}{2\,a_k\,\whk} = \frac{q^k}{T-t}\,.
\end{equation*}
By \eqref{eqn: solution H11}, \eqref{eqn: solution H12}, \eqref{eqn: solution H13}, we have
\begin{equation*}
\begin{split}
H_{11}(t,q^x,q^z)	= \int_t^T&\,-e^{(2\,\mu_x+\sigma_x^2)\,(u-t)}\\
&\times\, \frac{1}{2}\,\sigma_x^2\,\EE_{t,q^x,q^z}\left[\,\bigg(\mfQ^x_u + \mfQ^z_u - h_{02}^x(u)\,(\mfQ^x_u)^2 - h_{02}^z(u)\,(\mfQ^z_u)^2\bigg)^2\right]\,du\,,\\
H_{12}(t,q^y) = \int_t^T&\,-e^{(2\,\mu_y+\sigma_y^2)\,(u-t)}\,\frac{1}{2}\,\sigma_y^2\,\EE_{t,q^y}\left[\,\bigg(\mfQ^y_u - h_{02}^y(u)\,(\mfQ^y_u)^2\bigg)^2\right]\,du\,,\\
H_{13}(t,q^x,q^y,q^z) = \int_t^T&\,-e^{(\mu_x+\mu_y+\rho\,\sigma_x\,\sigma_y)\,(u-t)}\\
&\times\, \,\rho\,\sigma_x\,\sigma_y\,\EE_{t,q^x,q^y,q^z}\bigg[\,\bigg(\mfQ^x_u + \mfQ^z_u - h_{02}^x(u)\,(\mfQ^x_u)^2 - h_{02}^z(u)\,(\mfQ^z_u)^2\bigg)\,\\
&\times\,\bigg(\mfQ^y_u - h_{02}^y(u)\,(\mfQ^y_u)^2\bigg)\bigg]\,du\,,
\end{split}
\end{equation*}
and by \eqref{eqn:auxQ dynamics}, \eqref{eqn:betak}, the auxiliary processes $\mfQ^k_u$, $k\in\{x,y,z\}$ are deterministic functions of $u$ and satisfy
\begin{equation*}
d\mfQ^k_u = -\frac{h_{02}^k(u)}{a_k}\,\mfQ^k_u\,du\,,\quad\,u\in[t,T]\,,
\end{equation*}
with initial condition $\mfQ^k_t = q^k$. The solution is
\begin{equation*}
	\mfQ^k_u = \frac{a_k+\alpha_k\,(T-u)}{a_k+\alpha_k\,(T-t)}\,q^k\,.
\end{equation*}
After evaluating the integrals, taking derivatives with respect to $q^k$, taking the limit $\alpha_k\to \infty$, and  in the limit $t\to T$,  the second term in $\tilde{\nu}^{k,*}$ becomes
\begin{equation*}
\lim_{\substack{\alpha_k\to +\infty\,,\\\,k\in\{x,y,z\}}}- \varphi\,\frac{\partial_{q^k}H_1}{2\,a_k\,\whk} = \varphi\,\sigma_{\whk}\,\frac{q^k}{T-t}\,C_{\whk}(x,y,q^x,q^y,q^z) + o(T-t)\,,
\end{equation*}
where
\begin{equation*}
\begin{split}
C_{x}(x,y,q^x,q^y,q^z) =&\, \sigma_x\,\left[a_x\,x\,(q^x)^2 + a_z\,x\,\,(q^z)^2\right] + \rho\,\sigma_y\,a_y\,y\,(q^y)^2\,,\\
C_{y}(x,y,q^x,q^y,q^z) =&\, \rho\,\sigma_x\,\left[a_x\,x\,(q^x)^2 + a_z\,x\,\,(q^z)^2\right] + \sigma_y\,a_y\,y\,(q^y)^2\,.
\end{split}	
\end{equation*}
\qed

\end{appendices}

\section*{References}
\bibliographystyle{chicago}
\bibliography{Triplet}

\end{document}